\documentclass[11pt]{article}
\usepackage{fullpage,graphicx,amsmath,amsfonts,amssymb,amsthm,bbold,dsfont,authblk,cite,xcolor,multirow}
\usepackage{algorithm, algpseudocode}
\usepackage[colorlinks]{hyperref}
\usepackage[capitalise, nameinlink]{cleveref}
\usepackage{braket}
\usepackage{tikz}
\usepackage{pgfplots}
\pgfplotsset{compat=1.18}
\usepackage{quantikz}
\usepackage{bm}
\usepackage[labelfont=bf]{caption}
\usepackage{siunitx}

\title{Optimization by Decoded Quantum Interferometry}

\author[1]{Stephen P.~Jordan\footnote{stephenjordan@google.com}}
\author[1]{Noah Shutty\footnote{shutty@google.com}}
\author[2]{Mary Wootters}
\author[1]{Adam Zalcman}
\author[1,3]{\\Alexander Schmidhuber}
\author[1,4]{Robbie King}
\author[1]{Sergei V.~Isakov}
\author[1]{\\ Tanuj Khattar}
\author[1]{Ryan Babbush}
\affil[1]{\small{\it{Google Quantum AI, Venice, CA 90291}}}
\affil[2]{\small{\it{Departments of Computer Science and Electrical Engineering, Stanford University, Stanford, CA 94305}}}
\affil[3]{\small{\it{Center for Theoretical Physics, Massachusetts Institute of Technology, Cambridge, MA 02139}}}
\affil[4]{\small{\it{Department of Computing and Mathematical Sciences, Caltech, Pasadena, CA 91125}}}

\begin{document}

\date{}

\newcommand{\eq}[1]{(\ref{#1})}                   
\newcommand{\sect}[1]{\S\ref{#1}}                 
\renewcommand{\th}{^{\textrm{th}}}                

\renewcommand{\F}{\mathbb{F}} 
\newcommand{\tr}{\mathrm{tr}} 
\renewcommand{\abstractname}{\vspace{-3\baselineskip}}

\definecolor{codegreen}{rgb}{0,0.6,0}
\definecolor{codegray}{rgb}{0.5,0.5,0.5}
\definecolor{codepurple}{rgb}{0.58,0,0.82}
\definecolor{backcolour}{rgb}{0.95,0.95,0.92}

\newtheorem{theorem}{Theorem}[section]
\newtheorem{definition}{Definition}[section]
\newtheorem{conjecture}{Conjecture}[section]
\newtheorem{lemma}{Lemma}[section]
\newtheorem{fact}{Fact}[section]
\newtheorem{corollary}{Corollary}[section]

\theoremstyle{definition}
\newtheorem{remark}{Remark}[section]

\bibliographystyle{unsrt}

\maketitle

\begin{abstract}
Achieving superpolynomial speedups for optimization has long been a central goal for quantum algorithms. Here we introduce Decoded Quantum Interferometry (DQI), a quantum algorithm that uses the quantum Fourier transform to reduce optimization problems to decoding problems. For approximating optimal polynomial fits over finite fields, DQI achieves a superpolynomial speedup over known classical algorithms. The speedup arises because the problem's algebraic structure is reflected in the decoding problem, which can be solved efficiently. We then investigate whether this approach can achieve speedup for optimization problems that lack algebraic structure but have sparse clauses. These problems reduce to decoding LDPC codes, for which powerful decoders are known. To test this, we construct a max-XORSAT instance where DQI finds an approximate optimum significantly faster than general-purpose classical heuristics, such as simulated annealing. While a tailored classical solver can outperform DQI on this instance, our results establish that combining quantum Fourier transforms with powerful decoding primitives provides a promising new path toward quantum speedups for hard optimization problems.
\end{abstract}

\section{Introduction}
\label{sec:introduction}

NP-hardness results suggest that finding exact optima and even sufficiently good approximate optima for worst-case instances of many optimization problems is likely out of reach for polynomial-time algorithms both classical and quantum \cite{T14}. Nevertheless, there remain combinatorial optimization problems, such as the closest vector problem, for which there is a large gap between the best approximation achieved by a polynomial-time classical algorithm \cite{AKS01} and the strongest complexity-theoretic inapproximability result \cite{M15}. When considering average-case complexity such gaps become more prevalent, as few average-case inapproximability results are known. These gaps present a potential opportunity for quantum computers, namely achieving in polynomial time an approximation that requires superpolynomial time to achieve using known classical algorithms.

Quantum algorithms for combinatorial optimization have been the subject of intense research over the last three decades \cite{FGGLLP01, FGG14, H18, BFM22, DP22, KBF23, SLC24}, which has uncovered some evidence of possible superpolynomial quantum speedup for certain optimization problems \cite{FGRV25, LZW23, CLZ22, PUW24, S22, VHG21, EH22}. Nevertheless, the problem of finding superpolynomial quantum advantage for optimization is extremely challenging and remains largely open. 

Here, we propose a quantum algorithm for optimization that uses interference patterns as its main underlying principle. We call this algorithm Decoded Quantum Interferometry (DQI). DQI uses a Quantum Fourier Transform to arrange that amplitudes interfere constructively on symbol strings for which the objective value is large, thereby enhancing the probability of obtaining good solutions upon measurement. Most prior approaches to quantum optimization have been Hamiltonian-based \cite{FGGLLP01, FGG14}, with a notable exception being the superpolynomial speedup due to Chen, Liu, and Zhandry \cite{CLZ22} for finding short lattice vectors, which uses Fourier transforms and can be seen as an ancestor of DQI. Whereas Hamiltonian-based quantum optimization methods are often regarded as exploiting the local structure of the optimization landscape (\textit{e.g.} tunneling across barriers \cite{DBI16}), our approach instead exploits sparsity that is routinely present in the Fourier spectrum of the objective functions for combinatorial optimization problems and can also exploit more elaborate structure in the spectrum if present.

Before presenting evidence that DQI can efficiently obtain approximate optima not achievable by known polynomial-time classical algorithms, we quickly illustrate the essence of the DQI algorithm by applying it to max-XORSAT. We use max-XORSAT as our first example because, although it is not the problem on which DQI has achieved its greatest success, it is the context in which DQI is simplest to explain.

Given an $m \times n$ matrix $B$ with $m > n$, the max-XORSAT problem is to find an $n$-bit string $\mathbf{x}$ satisfying as many as possible among the $m$ linear mod-2 equations $B \mathbf{x} = \mathbf{v}$. Since we are working modulo 2 we regard all entries of the matrix $B$ and the vectors $\mathbf{x}$ and $\mathbf{v}$ as coming from the finite field $\mathbb{F}_2$. The max-XORSAT problem can be rephrased as maximizing the objective function 
\begin{equation}
    \label{eq:obj1}
    f(\mathbf{x}) = \sum_{i=1}^m (-1)^{v_i + \mathbf{b}_i \cdot \mathbf{x}}.
\end{equation}
where $\mathbf{b}_i$ is the $i\th$ row of $B$. Thus $f(\mathbf{x})$ is the number among the $m$ linear equations that are satisfied minus the number unsatisfied. 

From \eq{eq:obj1} one can see that the Hadamard transform of $f$ is extremely sparse: it has $m$ nonzero amplitudes, which are on the strings $\mathbf{b}_1,\ldots,\mathbf{b}_m$. The state $\sum_{\mathbf{x} \in \mathbb{F}_2^n} f(\mathbf{x}) \ket{\mathbf{x}}$ is thus easy to prepare. Simply prepare the superposition $\sum_{i=1}^m (-1)^{v_i} \ket{\mathbf{b}_i}$ and apply the quantum Hadamard transform. (Here, for simplicity, we have omitted normalization factors.) Measuring the state $\sum_{\mathbf{x} \in \mathbb{F}_2^n} f(\mathbf{x}) \ket{\mathbf{x}}$ in the computational basis yields a biased sample, where a string $\mathbf{x}$ is obtained with probability proportional to $f(\mathbf{x})^2$, which slightly enhances the likelihood of obtaining strings of large objective value relative to uniform random sampling.

To obtain stronger enhancement, DQI prepares states of the form
\begin{equation}
    \label{eq:Pdef}
    \ket{P(f)} = \sum_{\mathbf{x} \in \mathbb{F}_2^n} P(f(\mathbf{x})) \ket{\mathbf{x}},
\end{equation}
where $P$ is an appropriately normalized degree-$\ell$ polynomial. The Hadamard transform of such a state always takes the form
\begin{equation}
    \label{eq:step1}
    \sum_{k=0}^\ell \frac{w_k}{\sqrt{\binom{m}{k}}} \sum_{\substack{\mathbf{y} \in \mathbb{F}_2^m \\ |\mathbf{y}|=k}} (-1)^{\mathbf{v} \cdot \mathbf{y}} \ket{B^T \mathbf{y}},
\end{equation}
for some coefficients $w_0,\ldots,w_\ell$. Here $|\mathbf{y}|$ denotes the Hamming weight of the bit string $\mathbf{y}$. The DQI algorithm prepares $\ket{P(f)}$ in five steps. The first step is to prepare the superposition $\sum_{k=0}^\ell w_k \ket{D_{m,k}}$, where
\begin{equation}
\ket{D_{m,k}} = \frac{1}{\sqrt{\binom{m}{k}}} \sum_{\substack{\mathbf{y} \in \mathbb{F}_2^m \\ |\mathbf{y}| = k}} \ket{\mathbf{y}}
\end{equation}
is the Dicke state of weight $k$. Preparing such superpositions over Dicke states can be done using $\mathcal{O}(m^2)$ quantum gates using the techniques of \cite{BE22, WT24}. Second, the phase $(-1)^{\mathbf{v} \cdot \mathbf{y}}$ is imposed by applying the Pauli product $Z_1^{v_1} \otimes \ldots \otimes Z_m^{v_m}$. Third, the quantity $B^T \mathbf{y}$ is computed into an ancilla register using a reversible circuit for matrix multiplication. This yields the state
\begin{equation}
    \sum_{k=0}^\ell \frac{w_k}{\sqrt{\binom{m}{k}}} \sum_{\substack{\mathbf{y} \in \mathbb{F}_2^m \\ |\mathbf{y}|=k}} (-1)^{\mathbf{v} \cdot \mathbf{y}} \ket{\mathbf{y}} \ket{B^T \mathbf{y}}.
\end{equation}
The fourth step is to use the value $B^T \mathbf{y}$ to infer $\mathbf{y}$, which can then be subtracted from $\ket{\mathbf{y}}$, thereby bringing it back to the all zeros state, which can be discarded. (This is known as ``uncomputation'' \cite{NC10}.) The fifth and final step is to apply a Hadamard transform to the remaining register, yielding $\ket{P(f)}$. This sequence of steps is illustrated in Fig. \ref{fig:schematic}.

\begin{figure}
\begin{center}
        \includegraphics[width=0.5\textwidth]{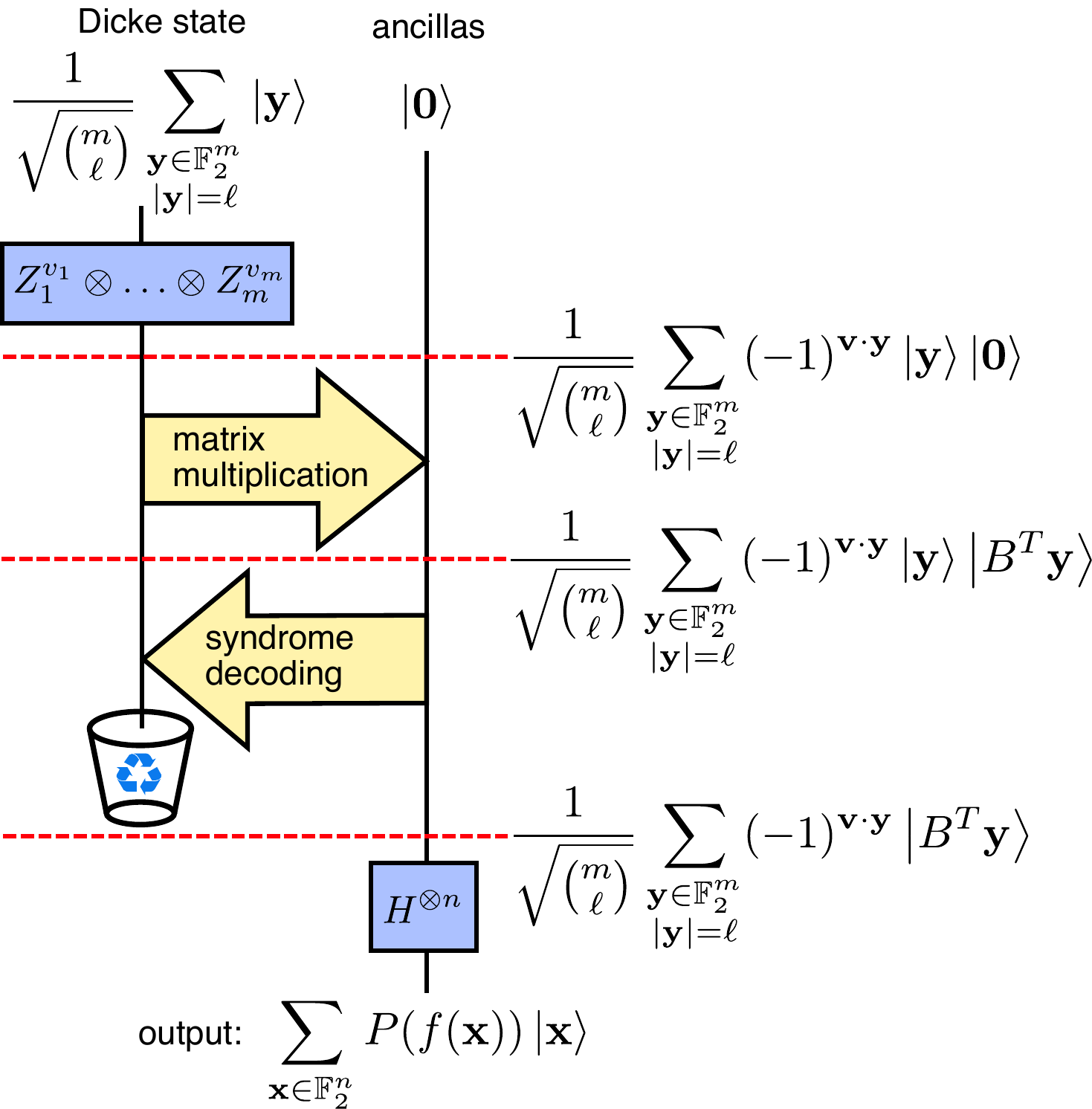}
    \caption{\label{fig:schematic} A schematic illustration of the steps of the DQI algorithm. Since the initial Dicke state is of weight $\ell$, the final polynomial $P$ is of degree $\ell$. Here, for simplicity, we take $w_\ell = 1$ and $w_k=0$ for all $k\neq \ell$.}
    \end{center}
\end{figure}

The fourth step, in which $\ket{\mathbf{y}}$ is uncomputed, is not straightforward because $B$ is a nonsquare matrix and thus inferring $\mathbf{y}$ from $B^T \mathbf{y}$ is an underdetermined linear algebra problem. However, we also know that $|\mathbf{y}| \leq \ell$. The problem of solving this underdetermined linear system with a Hamming weight constraint is precisely the syndrome decoding problem for the classical error correcting code $C^\perp = \{ \mathbf{d} \in \mathbb{F}_2^m : B^T \mathbf{d} = \mathbf{0} \}$ with up to $\ell$ errors. 

In general, syndrome decoding is an NP-hard problem \cite{BMT78}. However, when $B$ is very sparse or has certain kinds of algebraic structure, the decoding problem can be solved by polynomial-time classical algorithms even when $\ell$ is large (\textit{e.g.} linear in $m$). By solving this decoding problem using a reversible implementation of such a classical decoder one uncomputes $\ket{\mathbf{y}}$ in the first register. If the decoding algorithm requires $T$ quantum gates, then the number of gates required to prepare $\ket{P(f)}$ is $\mathcal{O}(T + m^2)$.

Approximate solutions to the optimization problem are obtained by measuring $\ket{P(f)}$ in the computational basis. The higher the degree of the polynomial in $\ket{P(f)}$, the greater one can bias the measured bit strings toward solutions with large objective value. However, this requires solving a harder decoding problem, as the maximum number of errors is equal to the degree of $P$. Next, we summarize how, by making optimal choice of $P$ and judicious choice of decoder, DQI can be a powerful optimizer for some classes of problems.

\section{Results}
\label{sec:results}

Although DQI can be applied more broadly, the most general optimization problem that we apply DQI to in this paper is max-LINSAT, which we define as follows.

\begin{definition}\label{def:linsat}
    Let $\mathbb{F}_p$ be a finite field and let $B \in \mathbb{F}_p^{m \times n}$. For each $i=1,\ldots,m$, let $F_i \subset \mathbb{F}_p$ be an arbitrary subset of $\mathbb{F}_p$, which yields a corresponding constraint $\sum_{j=1}^n B_{ij} x_j \in F_i$. The max-LINSAT problem is to find $\mathbf{x} \in \mathbb{F}_p^n$ satisfying as many as possible of these $m$ constraints.
\end{definition}
\noindent
We focus primarily on the case that $p$ has at most polynomially large magnitude and the subsets $F_1,\ldots,F_m$ are given as explicit lists. The max-XORSAT problem is the special case where $p=2$ and $|F_i| = 1$ for all $i$.

Consider a max-LINSAT instance where the sets $F_1,\ldots,F_m$ each have size $r$. Let $\langle s \rangle$ be the expected number of constraints satisfied by the symbol string sampled in the final measurement of the DQI algorithm. Suppose we have a polynomial-time algorithm that can correct up to $\ell$ bit flip errors on codewords from the code $C^\perp = \{ \mathbf{d} \in \mathbb{F}_p^m : B^T \mathbf{d} = \mathbf{0} \}$. Then, in polynomial time, DQI achieves the following approximate optimum to the max-LINSAT problem
\begin{equation}
    \label{eq:semicircle_general1}
    \frac{\langle s \rangle}{m} = \left( \sqrt{\frac{\ell}{m}\left(1 - \frac{r}{p}\right)} + \sqrt{\frac{r}{p}\left(1 - \frac{\ell}{m}\right)}\right)^2
\end{equation}
if $\frac{r}{p}\leq 1-\frac{\ell}{m}$ and $\frac{\langle s \rangle}{m}=1$ otherwise. See Theorem \ref{thm:semicircle} for the precise statement in the case of perfect decoding and Theorem \ref{thm:imperfect_decoding} for the analogous statement in the presence of decoding errors. This is achieved by a specific optimal choice of the coefficients $w_0,\ldots,w_\ell$, which can be classically precomputed in polynomial time, as described in \sect{sec:semicircle}.

Note that $r/p$ is the fraction of constraints that would be satisfied if the variables were assigned uniformly at random. In the case $r/p=1/2$, \eq{eq:semicircle_general1} becomes the equation of a semicircle, as illustrated in Fig.~\ref{fig:semicircle}. Hence we informally refer to \eq{eq:semicircle_general1} as the ``semicircle law.''

Via \eq{eq:semicircle_general1}, any result on decoding a class of linear codes implies a corresponding result regarding the performance of DQI for solving a class of combinatorial optimization problems which are dual to these codes. This enables two new lines of research in quantum optimization. The first is to harvest the coding theory literature for rigorous theorems on the performance of decoders for various codes and obtain as corollaries guarantees on the approximation achieved by DQI for corresponding optimization problems. The second is to perform computer experiments to determine the empirical performance of classical heuristic decoders, which through equation \eq{eq:semicircle_general1} can be compared against the empirical performance of classical heuristic optimizers. In this manner DQI can be benchmarked instance-by-instance against classical heuristics, even for optimization problems far too large to attempt on present-day quantum hardware. We next describe our results so far from each of these two lines of research.

We first use rigorous decoding guarantees to analyze the performance of DQI on the following problem.

\begin{definition}
    \label{def:OPI}
    Given integers $n < p - 1$ with $p$ prime, an instance of the Optimal Polynomial Intersection (OPI) problem is as follows. Let $F_1,\ldots,F_{p-1}$ be subsets of the finite field $\mathbb{F}_p$. Find a polynomial $Q \in \mathbb{F}_p[y]$ of degree at most $n-1$ that maximizes $f_\text{OPI}(Q) = |\{y \in \{1,\ldots,p-1\} : Q(y) \in F_y\}|$, \textit{i.e.} intersects as many of these subsets as possible.
\end{definition}

\noindent
An illustration of this problem is given in Fig. \ref{fig:curvefitting}. 

\begin{figure}
    \centering
\begin{tikzpicture}[domain=0:11,scale=.5]
    \draw(0,0) -- (11,0);
    \draw(0,0) -- (0,11);
    \node[anchor=west] at (11.5,0) {$\mathbb{F}_p$};
    \node[anchor=south] at (0,11.5) {$\mathbb{F}_p$};
    \foreach \i in {2,4,6,8, 10}
    {
    \draw(\i,-.25) --(\i, 0);
   }
\draw[thick, orange, fill=orange!10, rounded corners] (1.75, 3.75) rectangle (2.25, 9.25);
\draw[thick, orange, fill=orange!10, rounded corners] (3.75, .25) rectangle (4.25, 2.25);
\draw[thick, orange, fill=orange!10, rounded corners] (3.75, 3.75) rectangle (4.25, 6.25);
\draw[thick, orange, fill=orange!10, rounded corners] (3.75, 7.75) rectangle (4.25, 10.25);
\draw[thick, orange, fill=orange!10, rounded corners] (5.75, 5.75) rectangle (6.25, 10.25);
\draw[thick, orange, fill=orange!10, rounded corners] (7.75, .75) rectangle (8.25, 5.25);
\draw[thick, orange, fill=orange!10, rounded corners] (9.75, 4.75) rectangle (10.25, 9.25);
 \node at(6,-.75) {$y_1$};
 \node[orange](l) at (6.5, 12) {$F_{y_1}$};
 \draw[->, orange] (l) to [out=-90,in=90] (6,10.3);
    \draw[color=blue, very thick] plot(\x, {.05 * \x^3 - \x^2 + 5*\x + 2}) node[right] {$Q_1(y)$}; 
    \draw[color=red, very thick,dashed] plot(\x, {-.05 * \x^3 + .8*\x^2 - 3*\x + 8}) node[right] {$Q_2(y)$};
\end{tikzpicture}
\caption{A stylized example of the Optimal Polynomial Intersection (OPI) problem.  For $y_1 \in \mathbb{F}_p$, the orange set above the point $y_1$ represents $F_{y_1}$.  Both of the polynomials $Q_1(y)$ and $Q_2(y)$ represent solutions that have a large objective value, as they each intersect all but one set $F_y$.}
    \label{fig:curvefitting}
\end{figure}

In \sect{sec:OPI} we show that OPI is a special case of max-LINSAT over $\mathbb{F}_p$ with $m = p-1$ constraints in which $B$ is a Vandermonde matrix and thus $C^\perp$ is a Reed-Solomon code. Syndrome decoding for Reed-Solomon codes can be solved in polynomial time out to half the distance of the code, \textit{e.g.} using the Berlekamp-Massey algorithm \cite{B15}. Consequently, in DQI we can take $\ell = \lfloor \frac{n+1}{2} \rfloor$. For the regime where $r/p$ and $n/p$ are constants and $p$ is taken asymptotically large, the fraction of satisfied constraints achieved by DQI using the Berlekamp Massey decoder can be obtained by substituting $\frac{\ell}{m} = \frac{n}{2p}$ into \eq{eq:semicircle_general1}.

OPI and special cases of it have been studied in several domains. In the coding theory literature, OPI is studied under the name \emph{list-recovery}, and in the cryptography literature it is studied under the name \emph{noisy polynomial reconstruction/interpolation}~\cite{NP99,BN00}. OPI can also be viewed as a generalization of the polynomial approximation problem, studied in \cite{GM14, SW05, S05}, in which each set $F_i$ is a contiguous range of values in $\mathbb{F}_p$. In \sect{sec:classical} we analyze the algorithms from these literatures and find that, for the parameter regime addressed by DQI, the best approximation achieved in polynomial time classically is $\frac{1}{2} + \frac{n}{2p}$, via Prange's algorithm. As shown in Figure \ref{fig:OPI_comparison}, for $r/p = 1/2$ and any fixed $0 < n/p < 1$, DQI with the Berlekamp-Massey decoder exceeds the satisfaction fraction achieved by Prange's algorithm in the limit of large $p$. Classically, the only methods we are aware of to exceed the satisfaction fraction achieved by Prange's algorithm are brute force search or slight refinements thereof, which have exponential runtime. Thus, DQI achieves superpolynomial speedup for this problem, assuming no polynomial-time algorithm is found that can match the satisfaction fraction that DQI achieves.

Currently, there are no results directly showing that the OPI problem in the parameter regime that we consider is classically intractable under any standard complexity-theoretic or cryptographic assumptions. However, such results are known for certain limiting cases of the OPI problem, and we propose the task of extending these results to regimes more relevant to DQI for future research. The hardness of the special case of OPI when $|f_i^{-1}(+1)| = 1$, in a certain parameter regime, has been proposed as a cryptographic assumption in \cite{NP06}, which has not been broken to our knowledge. Finding exact optima for OPI with $|f_i^{-1}(+1)| = 1$ can be cast as maximum-likelihood decoding for Reed-Solomon codes, which is known to be NP-hard~\cite{GV05,GGG18}. Finding sufficiently good approximate optima is known to be as hard as discrete log~\cite{CW07,CW08}, but these hardness results do not match the parameter regime addressed by DQI.

As a concrete example, for $n \simeq p/10$ and $r/p \simeq 1/2$, the fraction of constraints satisfied by Prange's algorithm is $0.55$, whereas DQI achieves $1/2 + \sqrt{19}/20 \simeq 0.7179$. As a specific point of comparison, we challenge the algorithms community to beat this by a classical polynomial time algorithm. Interestingly, for these parameters, one statistically expects that solutions satisfying all $p-1$ constraints exist, but they apparently remain out of reach of polynomial time algorithms both quantum and classical.

To find classically intractable instances of OPI solvable by DQI with minimal quantum resources, we find it is advantageous to choose $n/p \simeq r/p \simeq 1/2$. For these parameters DQI achieves satisfaction fraction 0.933. As discussed in \sect{sec:resources}, achieving this using classical algorithms known to us has prohibitive computational cost for $p$ as small as 521. The dominant cost in DQI+BM is the reversible implementation of the subroutine to find the shortest linear feedback shift register (LFSR) used in the Berlekamp-Massey algorithm. In \sect{sec:resources} we use Qualtran \cite{harrigan2024expressing} to find that at $p=521$ the LFSR can be found using approximately $1 \times 10^8$ logical Toffoli gates and $9 \times 10^3$ logical qubits.

We next use computer experiments to benchmark the performance of DQI against classical heuristics on average-case instances from certain families of max-XORSAT with sparse $B$. DQI reduces such problems to decoding problems on codes with sparse parity check matrices. Such codes are known as Low Density Parity Check (LDPC) codes. Polynomial-time classical algorithms such as belief propagation (BP) can decode randomly sampled LDPC codes up to numbers of errors that nearly saturate information-theoretic limits \cite{Gal62, RU01, MM09}. This makes sparse max-XORSAT an enticing target for DQI.  Although we use max-XORSAT as a convenient testbed for DQI, other commonly-studied optimization problems such as max-$k$-SAT could be addressed similarly. Specifically, consider any binary optimization problem in which the objective function counts the number of satisfied constraints, where each constraint is a Boolean function of at most $k$ variables. By taking the Hadamard transform of the objective function, one converts such a problem into an instance of weighted max-$k$-XORSAT, where the number of variables is unchanged and the number of constraints has been increased by at most a factor of $2^k$.

Although we are able to analyze the asymptotic average case performance of DQI rigorously we do not restrict the classical competition to algorithms with rigorous performance guarantees. Instead, we choose to set a high bar by also attempting to beat the empirical performance of classical heuristics that lack such guarantees.

Through careful tuning of sparsity patterns in $B$, we are able to find some families of sparse max-XORSAT instances for which DQI with standard belief propagation decoding finds solutions satisfying a larger fraction of constraints than we are able to find using a comparable number of computational steps by any of the general-purpose classical optimization heuristics that we tried, which are listed in Table \ref{tab:xorboard}. However, unlike our OPI example, we do not put this forth as a potential example of superpolynomial quantum advantage. Rather, we are able to construct a tailored classical algorithm specialized to these instances which, with seven minutes of runtime, finds solutions where the fraction of constraints satisfied slightly beats DQI+BP. As discussed in \sect{sec:wins}, our tailored heuristic is a variant of simulated annealing that assigns temperature-dependent weights to the terms in the cost function determined by how many variables they contain.

The comparison against simulated annealing is complicated by the fact that, as shown in \sect{sec:sa_convergence}, the fraction of clauses satisfied by simulated annealing increases as a function of the duration of the anneal. Thus there is not a unique sharply-defined number indicating the maximum satisfaction fraction reachable by simulated annealing. DQI reduces our sparsity-tuned max-XORSAT problem to an LDPC decoding problem that our implementation of belief propagation solves in approximately 8 seconds on a single core, excluding the time used to load and parse the instance. Thus, a natural point of comparison is the result obtained by simulated annealing with similar runtime. By running our optimized C++ implementation of simulated annealing for 8 seconds, we are only able to reach 0.764. If we allow parallel execution of multiple anneals and increase our runtime allowance, we are able to eventually replicate the satisfaction fraction achieved by DQI+BP using simulated annealing. The shortest anneal that achieved this used five cores and ran for 73 hours, \textit{i.e.} five orders of magnitude longer than our belief propagation decoder. Although dependent on implementation details, we can take this ratio of runtimes as a rough indicator of the ratio of computational steps. In the context of DQI the decoder would need to be implemented as a reversible circuit and subject to overhead due to quantum error correction, so this should not be interpreted as an indicator of quantum versus classical runtime.

\begin{figure}
\begin{center}
\begin{tikzpicture}[scale=1.35]
  \draw[thick] (0,0) rectangle (4,4);

  \node at (2, -0.4) {$n/p$};
  \node at (-0.6, 2) {$\braket{s}/p$};

  \draw[thick] (0,-0.1) -- (0,0.1);
  \node[below] at (0, -0.1) {$0$};

  \draw[thick] (4,-0.1) -- (4,0.1);
  \node[below] at (4, -0.1) {$1$};

  \draw[thick] (.4,3.9) -- (.4,4.1);
  \node[above] at (.4, 4.1) {$\frac{1}{10}$};

\draw[thick] (2,3.9) -- (2,4.1);
\node[above] at (2, 4.1) {$\frac{1}{2}$};

  \draw[thick] (-0.1, 0) -- (0.1, 0);
  \node[left] at (-0.1, 0) {$\tfrac{1}{2}$};

  \draw[thick] (-0.1, 4) -- (0.1, 4);
  \node[left] at (-0.1, 4) {$1$};

  \draw[domain=0:4, smooth, samples=200, variable=\x, thick, blue]
    plot ({\x}, {8 * (0.5 + sqrt((\x/4)/2 * (1 - (\x/4)/2)) - 0.5)});

  \draw[domain=0:4, smooth, samples=200, variable=\x, thick, red]
    plot ({\x}, {8 * ((0.5 + \x/8) - 0.5)});

  \node[blue] at (1.2, 3.6) {DQI+BM};
  \node[red] at (2.95, 2.05) {Prange};

  \draw[thick, dotted] (0.4,0) -- (0.4,4);
  \draw[thick, dotted] (2,0) -- (2,4);
\end{tikzpicture}    \caption{\label{fig:OPI_comparison} Here we plot the expected fraction $\langle s\rangle/p$ of satisfied constraints achieved by DQI with the Berlekamp-Massey decoder and by Prange's algorithm for the OPI problem in the balanced case $r/p = 1/2$, as a function of the ratio of variables to constraints $n/p$. At $n/p=1/10$ Prange's algorithm satisfies a fraction $0.55$ of the clauses whereas DQI satisfies $\langle s \rangle/p = 1/2 + \sqrt{19}/20 \simeq 0.7179$. As a concrete challenge to the classical algorithms community we propose matching or exceeding this value in polynomial time. In our concrete resource estimation in \sect{sec:resources} we consider $n/p = 1/2$, where OPI achieves $\langle s \rangle/p = 1/2 + \sqrt{3}/4 \simeq 0.9330$ and Prange's algorithm achieves $0.75$.}
\end{center}
\end{figure}

\begin{table}
    \begin{center}
    \begin{tabular}[t]{|l|c|}
        \hline
        \textbf{Algorithm} & \textbf{SAT Fraction} \\
        \hline
        Tailored Heuristic (7 min $\times$ 1 core) & 0.880  \\
        \hline
        Long Anneal (73 hrs $\times$ 5 cores) & 0.832 \\
        \hline
        DQI + Belief Propagation & $\hspace{-12pt} \geq 0.831$ \\
        \hline
        Prange's algorithm & 0.812 \\
        \hline
        Short Anneal (8 sec $\times$ 1 core) & 0.764 \\
        \hline
        Greedy Algorithm & 0.666 \\
        \hline
        AdvRand & 0.554 \\
        \hline
    \end{tabular}
    \end{center}
    \caption{\label{tab:xorboard} Here, we compare DQI, using a standard belief propagation decoder, against classical algorithms for a randomly-generated max-XORSAT instance with irregular degree distribution specified in \sect{sec:wins}. We consider an example instance with $31,216$ variables and $50,000$ constraints. The classical algorithms above are defined in \sect{sec:classical}. For simulated annealing the satisfaction fraction grows with runtime, so we report two numbers. The first is the optimum reachable by limiting simulated annealing to the same runtime used by belief propagation to solve the problem to which the max-XORSAT instance is reduced by DQI (8 seconds $\times$ 1 core) and the second is for the shortest anneal that matched satisfaction fraction achieved by DQI+BP (73 hours $\times$ 5 cores).}
\end{table}

\section{Discussion}\label{sec:discussion}

The idea that quantum Fourier transforms could be used to achieve reductions between problems on lattices and their duals originates in the early 2000s in work of Regev, Aharonov, and Ta-Shma \cite{ATS03,R04,AR05,R09}. Linear codes, as considered here, are closely analogous to lattices but over finite fields. By considering lattices with only geometric structure no quantum speedups were found using these reductions until the 2021 breakthrough of Chen, Liu, and Zhandry \cite{CLZ22}, which obtains a superpolynomial speedup for a constraint satisfaction problem by combining these ideas with an intrinsically quantum decoding method. Other recent explorations of Regev-style reductions to general unstructured codes and lattices are given in \cite{chailloux2024quantum, DRT23, EH22}. Here, we restrict attention to codes defined by matrices that are either sparse or algebraically structured and in the latter case are able to obtain an apparent superpolynomial quantum speedup for an optimization problem.

Recently, Yamakawa and Zhandry have also considered the application of Regev-style reductions to a problem with extra structure and obtained quantum advantage \cite{YZ22}. They define an oracle problem that they prove can be solved using polynomially many quantum queries but requires exponentially many classical queries. Their problem is essentially equivalent to max-LINSAT over an exponentially large finite field $\mathbb{F}_{2^t}$, where the sets $F_1,\ldots,F_m$ are defined by random oracles and the matrix $B$ is obtained from a folded Reed-Solomon code. In \sect{sec:folded} we recount the exact definition of the Yamakawa-Zhandry problem and argue that DQI can be extended to the Yamakawa-Zhandry problem and in this case likely yields solutions satisfying all constraints. Although problems with exponentially large $F_1,\ldots,F_m$ defined by oracles are far removed from industrial optimization problems, this limiting case provides evidence against the possibility of efficiently simulating DQI with classical algorithms and thereby ``dequantizing'' it, as has happened with some prior quantum algorithms proposed as potential superpolynomial speedups \cite{CGP22}. More precisely, our argument suggests that DQI cannot be dequantized by any \emph{relativizing} techniques, in the sense of \cite{F94}.

We conclude by noting that the work reported here initiates the exploration of quantum speedups through DQI but is very far from completing it. In particular, we highlight three avenues for future work: multivariate OPI, custom decoders for solving max-XORSAT by DQI, and sampling problems. First, we note that the DQI algorithm can be straightforwardly adapted to solve the multivariate generalization of OPI. As shown in \sect{sec:multivariate}, multivariate OPI gets reduced by DQI to the decoding of Reed-Muller codes. Known polynomial-time classical algorithms can decode all Reed-Muller codes out to half their distance \cite{PW04}. (Reed-Solomon codes are the univariate special case.) Consequently, one expects a region of parameter space for which DQI achieves superpolynomial speedup on multivariate OPI, which includes the speedup on univariate OPI presented here as a special case. Mapping out this region of quantum advantage remains for future work.

Second, we note that our exploration of DQI applied to max-XORSAT is far from exhaustive. In particular, \eq{eq:semicircle_general1} enables a benchmark-driven approach to the development of tailored heuristics for decoding designed to achieve quantum speedup on some class of optimization problems via DQI. This search can be guided by upper bounds on the performance of DQI that, via the semicircle law, follow from information-theoretic limits on decoding. Such an analysis is given in \sect{sec:limits} and shows that for $D$-regular max-$k$-XORSAT instances, the upper bound on the possible performance of DQI with classical decoders is already exceeded by the empirical performance of simulated annealing when $k$ is too small relative to $D$. Additionally, we are able to compare the performance of DQI against the Quantum Approximate Optimization Algorithm (QAOA) for various ensembles of max-$k$-XORSAT instances at $k=2$ and $k=3$ and on all of these QAOA exceeds the upper bound on performance for DQI with classical decoders.

These limits show that, for DQI to achieve advantage on max-$k$-XORSAT, one must either go to large $k$ or move to quantum decoders that exploit the coherence of the bit flip errors. Large-$k$ problems are reduced by DQI to decoding problems in which the parity check matrix is denser than in typical LDPC codes. The increased density degrades the performance of belief propagation. This suggests future research developing decoders to tolerate denser parity check matrices than are typically used. Despite some progress along these lines \cite{FW05,FM07,DD08,TS08,YLB09,TRH09,TRH10} this remains an underexplored area compared to the decoding of codes with very sparse parity check matrices. With quantum decoders, it remains information-theoretically possible for DQI to achieve advantage over known polynomial-time classical and quantum algorithms, even for small $k$. Realizing this potential advantage depends on the development of polynomial-size quantum circuits for this \textit{quantum decoding problem}. Some exciting progress on this problem has been reported in \cite{chailloux2024quantum, PR22,CLZ22}.

Third, we note that DQI produces unbiased samples, in which the probability of obtaining a given solution to an optimization problem is constant across all solutions achieving a given objective value. This guarantee of fair sampling is absent for most classical optimization algorithms and has known applications to very hard problems of approximate counting \cite{SJ89}. \\

\noindent \textbf{Data availability:} The problem instances that we describe, the code used in our computer experiments and resource estimation, and the raw data from our plots are available at \url{https://doi.org/10.5281/zenodo.13327870}. \\

\noindent \textbf{Acknowledgments:} We thank Robin Kothari, Ryan O'Donnell, Edward Farhi, Hartmut Neven, Kostyantyn Kechedzhi, Sergio Boixo, Vadim Smelyanskiy, Yuri Lensky, Dorit Aharonov, Oded Regev, Jarrod McClean, Madhu Sudan, Umesh Vazirani, Yuval Ishai, Brett Hemenway Falk, Oscar Higgott, John Azariah, Ojas Parekh, Jon Machta, Helmut Katzgraber, Craig Gidney, Noureldin Yosri and Dmitri Maslov for useful discussions. MW's work on this project was funded by a grant from Google Quantum AI.\\

\noindent \textbf{Organization:}  In \sect{sec:OptNumSat}, we state Theorem~\ref{thm:semicircle}, which characterizes the performance of DQI on max-LINSAT problems in terms of the ability to solve the corresponding decoding problem. In \sect{sec:OPI} we formally define the OPI problem and apply Theorem~\ref{thm:semicircle} to predict DQI's performance on OPI. In \sect{sec:sparsinstances} we discuss how DQI performs on unstructured sparse instances of sparse max-XORSAT using belief propagation decoding. In \sect{sec:prior} we discuss the long line of prior work related to DQI. In \sect{sec:DQI} we explain the DQI algorithm in detail. In \sect{sec:semicircle} we prove Theorem~\ref{thm:semicircle}. In \sect{sec:beyond_2l} we state and prove ~\ref{thm:imperfect_decoding}, which is an analogue of Theorem~\ref{thm:semicircle} for the setting where $p = 2$ and $\ell$ exceeds half the distance of the code $C^\perp$.
In \sect{sec:classical} we discuss several existing classical and quantum optimization algorithms and compare their performance with DQI. We follow this in \sect{sec:wins} by showing how we construct an instance for which DQI, using belief propagation decoding, can achieve an approximate optimum that is very difficult to replicate using simulated annealing. In \sect{sec:limits} we derive information-theoretic upper bounds on the approximate optima achievable by DQI, which depend on whether one is considering classical or quantum decoders. In \sect{sec:folded} we generalize the max-LINSAT problem and the DQI algorithm to folded codes and extension fields in order to shed some light on DQI's potential applicability to the problem considered by Yamakawa and Zhandry in \cite{YZ22}. In \sect{sec:multivariate} we generalize the OPI problem to multivariate polynomials. Lastly, in \sect{sec:resources} we obtain concrete resource requirements (qubits, Clifford gates, and non-Clifford gates) to apply DQI, using the Berlekamp Massey decoder, to the OPI problem.

\section{Characterizing the Performance of DQI}\label{sec:OptNumSat}

DQI reduces the problem of approximating max-LINSAT to the problem of decoding the linear code $C^\perp$ over $\mathbb{F}_p$ whose parity check matrix is $B^T$. That is,
\begin{equation}
    C^\perp = \{ \mathbf{d} \in \mathbb{F}_p^m : B^T \mathbf{d} = \mathbf{0} \}.
\end{equation}
This decoding problem is to be solved in superposition, such as by a reversible implementation of any efficient classical decoding algorithm. If $C^\perp$ can be efficiently decoded out to $\ell$ errors then, given any appropriately normalized degree-$\ell$ polynomial $P$, DQI can efficiently produce the state
\begin{equation}
    \ket{P(f)} = \sum_{\mathbf{x} \in \mathbb{F}_p^n} P(f(\mathbf{x})) \ket{\mathbf{x}}.
\end{equation}
Upon measuring in the computational basis one obtains a given string $\mathbf{x}$ with probability $P(f(\mathbf{x}))^2$. One can choose $P$ to bias this distribution toward strings of large objective value. Larger $\ell$ allows this bias to be stronger, but requires the solution of a harder decoding problem. 

More quantitatively, in \sect{sec:semicircle}, we prove the following theorem.

\begin{theorem}
    \label{thm:semicircle}
        Given a prime $p$ and $B \in \mathbb{F}_p^{m \times n}$, let $f(\mathbf{x}) = \sum_{i=1}^m f_i (\sum_{j=1}^n B_{ij} x_j)$ be a max-LINSAT objective function. Suppose $|f_i^{-1}(+1)| = r$ for all $i=1,\ldots,m$ and some $r\in\{1,\dots,p-1\}$. Given a degree-$\ell$ polynomial $P$, let $\langle s \rangle$ be the expected number of satisfied constraints for the symbol string obtained upon measuring the corresponding DQI state $\ket{P(f)}$ in the computational basis. Suppose $2 \ell + 1 < d^\perp$ where $d^\perp$ is the minimum distance of the code $C^\perp = \{ \mathbf{d} \in \mathbb{F}_p^m : B^T \mathbf{d} = \mathbf{0} \}$, \textit{i.e.} the minimum Hamming weight of any nonzero codeword in $C^\perp$. In the limit $m \to \infty$, with $\ell/m$ fixed, the optimal choice of degree-$\ell$ polynomial $P$ to maximize $\langle s \rangle$ yields
        \begin{equation}
            \label{eq:semicircle_general}
            \frac{\langle s \rangle}{m} = \left( \sqrt{\frac{\ell}{m}\left(1 - \frac{r}{p}\right)} + \sqrt{\frac{r}{p}\left(1 - \frac{\ell}{m}\right)}\right)^2
        \end{equation}
        if $\frac{r}{p}\leq 1-\frac{\ell}{m}$ and $\frac{\langle s \rangle}{m}=1$ otherwise.
\end{theorem}

Theorem~\ref{thm:semicircle} assumes that $2\ell + 1 < d^\perp$, which is the same as requiring that $C^\perp$ can be in principle decoded from up to $\ell$ worst-case errors.  Further, if this decoding can be done \emph{efficiently}, then the DQI algorithm is also efficient.  In our analysis, we show how to relax these assumptions.  In particular, in  Theorem~\ref{thm:imperfect_decoding} we show that even when $2 \ell + 1 \geq d^\perp$ and it is not possible to decode $\ell$ worst-case errors, an efficient algorithm that succeeds with high probability over random errors can be used in the DQI algorithm to efficiently achieve a fraction of satisfied constraints close to the one given by \eqref{eq:semicircle_general}, at least for max-XORSAT problems $B \mathbf{x} \stackrel{\max}{=} \mathbf{v}$ with average-case $\mathbf{v}$.

For the balanced case $r \to p/2$, \eq{eq:semicircle_general} simplifies to
\begin{equation}
    \label{eq:dqifrac}
    \frac{\langle s \rangle}{m} =  \frac{1}{2} + \sqrt{\frac{\ell}{m} \left( 1 - \frac{\ell}{m} \right)},
\end{equation}
\textit{i.e.} the equation of a semicircle.

DQI reduces the problem of satisfying a large number of linear constraints to the problem of correcting a large number of errors in a linear code. Decoding linear codes is also an NP-hard problem in general \cite{BMT78}. So, one must ask whether this reduction is ever advantageous. We next present evidence that it can be.

\section{Optimal Polynomial Intersection}
\label{sec:OPI}

The problem which provides our clearest demonstration of the power of DQI is the OPI problem, as specified in Definition \ref{def:OPI} and illustrated in Fig. \ref{fig:curvefitting}. In this section, we explain how to apply DQI to OPI, and identify a parameter regime for OPI where DQI outperforms all classical algorithms known to us.

We first observe that OPI is equivalent to a special case of max-LINSAT. Let $q_0,\ldots,q_{n-1} \in \mathbb{F}_p$ be the coefficients in $Q$:
\begin{equation}
    Q(y) = \sum_{j=0}^{n-1} q_j \, y^j.
\end{equation}
Recall that a primitive element of a finite field is an element such that taking successive powers of it yields all nonzero elements of the field. Every finite field contains one or more primitive elements. Thus, we can choose $\gamma$ to be any primitive element of $\mathbb{F}_p$ and re-express the OPI objective function as
\begin{equation}
    f_\text{OPI}(Q) = |\{i \in \{0,1,\ldots,p-2\} : Q(\gamma^i) \in F_{\gamma^i}\}|.
\end{equation}
Next, let
\begin{equation}
    f_i(x) = \left\{ \begin{array}{ll} +1 & \textrm{if $x \in F_{\gamma^i}$} \\   -1 & \textrm{otherwise}
    \end{array}
    \right.
\end{equation}
for $i=0,\ldots,p-2$ and define the matrix $B$ by
\begin{equation}
    \label{eq:BRS}
    B_{ij} = \gamma^{i \times j} \quad i = 0,\ldots,p-2 \quad \quad j = 0,\ldots, n-1.
\end{equation}
Then the max-LINSAT objective function is $f(\mathbf{q})=\sum_{i=0}^{p-2} f_i(\mathbf{b}_i\cdot\mathbf{q})$ where $\mathbf{q}=(q_0,\ldots,q_{n-1})^T\in\mathbb{F}_p^n$ and $\mathbf{b}_i$ is the $i\th$ row of $B$. But $\mathbf{b}_i\cdot\mathbf{q} = Q(\gamma^i)$, so $f(\mathbf{q})=2\cdot f_\text{OPI}(Q)-(p-1)$ which means that the max-LINSAT objective function $f$ and the OPI objective function $f_\text{OPI}$ are equivalent. Thus we have re-expressed our OPI instance as an equivalent instance of max-LINSAT with $m=p-1$ constraints.

We will apply DQI to the case where
\begin{equation} 
    \label{eq:balf}
    |f_i^{-1}(+1)| = \lfloor p/2 \rfloor \quad \forall i=0,\ldots,p-2.
\end{equation}
By \eq{eq:balf}, in the limit of large $p$ we have $|f_i^{-1}(+1)|/p \to 1/2$ and $|f_i^{-1}(-1)|/p \to 1/2$ for all $i$. We call functions with this property ``balanced.'' 

When $B$ has the form \eq{eq:BRS}, then $C^\perp = \{\mathbf{d} \in \mathbb{F}_p^{p-1} : B^T \mathbf{d} = 0\}$ is a Reed-Solomon code with alphabet $\mathbb{F}_p$, block length $p-1$, dimension $p-n-1$, and distance $n + 1$. Note that our definition of $n$ is inherited from the parameters of the max-LINSAT instances that we start with and hence our notations for block length and dimension unfortunately do not conform to standard notations from coding theory.

Maximum likelihood syndrome decoding for Reed-Solomon codes can be solved in polynomial time out to half the distance of the code, \textit{e.g.} using the Berlekamp-Massey algorithm \cite{B15}. Consequently, in DQI we can take $\ell = \lfloor \frac{n+1}{2} \rfloor$. In \eq{eq:dqifrac} we can thus set the number of errors corrected $\ell \to \frac{n}{2}$ and the number of constraints $m \to p$, which shows that the asymptotic performance of DQI using Berlekamp-Massey is
\begin{equation}
\frac{\langle s \rangle_{\mathrm{DQI+BM}}}{p} = \frac{1}{2} + \sqrt{\frac{n}{2p} \left( 1 - \frac{n}{2p} \right)}.
\end{equation}
Here we have approximated $n+1$ by $n$ and $p-1$ by $p$ since this is an asymptotic formula anyway. For exact expressions at finite size see \sect{sec:semicircle}. The largest asymptotic fraction of satisfied clauses for OPI that we know how to obtain classically in polynomial time is
\begin{equation}
\frac{\langle s \rangle_{\mathrm{Prange}}}{m} = \frac{1}{2} + \frac{n}{2p},
\end{equation}
which is achieved by Prange's algorithm. These are plotted in Fig. \ref{fig:OPI_comparison}, where one sees that DQI+BP exceeds Prange's algorithm for all $n/p \in (0,1)$. (See \sect{sec:trunc} for a description of Prange's algorithm.)

Therefore, the Optimal Polynomial Intersection problem demonstrates the power of DQI. Assuming no polynomial-time classical algorithm for this problem is found that can match this fraction of satisfied constraints, this constitutes an example of an superpolynomial quantum speedup. It is noteworthy that our quantum algorithm is not based on a reduction to an Abelian Hidden Subgroup or Hidden Shift problem. The margin of victory for the approximation fraction (0.7179 vs. 0.55) is also satisfyingly large. Nevertheless, it is also of great interest to investigate whether such a quantum speedup can be obtained for more generic constraint satisfaction problems, with less underlying structure, as we do in the next section.

Before moving on to unstructured optimization problems, we make two remarks.

\begin{remark}[Relationship to the work of Yamakawa and Zhandry]\label{rem:YZ}  First, we note that the algorithm of Yamakawa and Zhandry \cite{YZ22}---which solves a version of OPI---does not apply in our setting. As discussed in \sect{sec:list-recovery}, the parameters of our OPI problem are such that solutions satisfying all constraints are statistically likely to exist but these exact optima seem to be computationally intractable to find using known classical algorithms. The quantum algorithm of Yamakawa and Zhandry, when it can be used, produces a solution satisfying all constraints. However, the quantum algorithm of Yamakawa and Zhandry has high requirements on the decodability of $C^\perp$. Specifically, for the ``balanced case'' in which $|f_i(+1)| \simeq |f_i(-1)|$ for all $i$, the requirement is that $C^\perp$ can be decoded from a $1/2$ fraction of random errors. For our OPI example, $C^\perp$ has rate $9/10$. Shannon's noisy-channel coding theorem implies that it is not possible to reliably decode $C^\perp$ in this setting. Thus, the quantum algorithm of Yamakawa and Zhandry is not applicable.
\end{remark}

\begin{remark}[Classical Complexity of OPI]\label{rem:OPIhard}
Proving rigorous classical hardness guarantees for OPI seems like a challenging problem. OPI, and OPI-like problems, have been proposed as cryptographically hard problems.  As discussed in \sect{sec:lattice}, a version of OPI in a different parameter regime was proposed as a hardness assumption for cryptographic applications by \cite{NP99}.  This conjecture was broken by \cite{BN00} using lattice attacks, but we demonstrate in \sect{sec:lattice} that these attacks do not apply in our parameter regime.  Later work \cite{NP06} proposed two updated hardness assumptions, which each would imply the hardness of a special case of OPI.\footnote{In more detail, the first problem assumed to be hard is related to bounded distance decoding for Reed-Solomon codes from random errors, which corresponds to OPI when $|f_i^{-1}(+1)| = 1$.  The second problem can be viewed as a generalization of OPI to \emph{randomly folded} Reed-Solomon codes with $|f_i^{-1}(+1)|$ larger than $1$; we show in Appendix~\ref{sec:folded} that DQI can apply to folded codes, but it does not yield useful attacks in the relevant parameter regime.}  These assumptions have yet to be broken to the best of our knowledge, and DQI does not seem to be an effective attack on them in the parameter regimes of interest.

There are other problems related to OPI that are known to be computationally hard, under standard assumptions. For example, the problem of maximum-likelihood decoding for Reed-Solomon Codes---which is the case of OPI when $|f_i^{-1}(+1)| = 1$ for all $i$---is known to be NP-hard~\cite{GV05,GGG18}.  List-decoding and bounded-distance decoding for Reed-Solomon codes to a large enough radius---also related to OPI when $|f_i^{-1}(+1)| = 1$---is known to be as hard as discrete log~\cite{CW07,CW08}. Theorem~\ref{thm:semicircle} does not provide strong performance guarantees for DQI applied to these problems. It would be very interesting to show that OPI (in a parameter regime $|f_i^{-1}(+1)| \propto p$ where Theorem~\ref{thm:semicircle} does give strong performance guarantees) is classically hard under standard cryptographic assumptions.
\end{remark}

\section{Random Sparse max-XORSAT}
\label{sec:sparsinstances}

In this section, we consider average-case instances from certain families of bounded degree max-$k$-XORSAT. In a max-$k$-XORSAT instance with degree bounded by $D$, each constraint contains at most $k$ variables and each variable is contained in at most $D$ constraints. In other words, the matrix $B \in \mathbb{F}_2^{m \times n}$ defining the instance has at most $k$ nonzero entries in any row and at most $D$ nonzero entries in any column. DQI reduces this to decoding the code $C^\perp$ whose parity check matrix is $B^T$. Codes with sparse parity check matrices are known as Low Density Parity Check (LDPC) codes. Randomly sampled LDPC codes are known to be correctable from a near-optimal number of random errors (asymptotically as $m$ grows)~\cite{Gal62}. Consequently, in the limit of large $m$ they can in principle be decoded up to a number of random errors that nearly saturates the information-theoretic limit dictated by the rate of the code. When $k$ and $D$ are very small, information-theoretically optimal decoding for random errors can be closely approached by polynomial-time decoders such as belief propagation \cite{RU01, MM09}. This makes sparse max-XORSAT a promising target for DQI.

In this section we focus on benchmarking DQI with standard belief propagation decoding (DQI+BP) against simulated annealing on max-$k$-XORSAT instances. We choose simulated annealing as our primary classical point of comparison because it is often very effective in practice on sparse constraint satisfaction problems and also because it serves as a representative example of local search heuristics. Local search heuristics are widely used in practice and include greedy algorithms, parallel tempering, TABU search, and many quantum-inspired optimization methods. As discussed in \sect{sec:local}, these should all be expected to have similar scaling behavior with $D$ on average-case max-$k$-XORSAT with bounded degree. Because of simulated annealing's simplicity, representativeness, and strong performance on average-case constraint satisfaction problems, beating simulated annealing on some class of instances is a good first test for any new classical or quantum optimization heuristic.

It is well-known that max-XORSAT instances become harder to approximate as the degree of the variables is increased \cite{H00, BM15}. Via DQI, a max-XORSAT instance of degree $D$ is reduced to a problem of decoding random errors for a code in which each parity check contains at most $D$ variables. As $D$ increases, with $m/n$ held fixed, the distance of the code and hence its performance under information-theoretically optimal decoding are not degraded at all. Thus, as $D$ grows, the fraction of constraints satisfied by DQI with information-theoretically optimal decoding would not degrade. In contrast, classical optimization algorithms based on local search yield satisfaction fractions converging toward $1/2$ in the limit $D \to \infty$ which is no better than random guessing. Thus as $D$ grows with $k/D$ fixed, DQI with information-theoretically optimal decoding will eventually surpass all classical local search heuristics. (See also Fig. \ref{fig:shannon_regions}.) However, for most ensembles of codes, the number of errors correctable by standard polynomial-time decoders such as belief propagation falls increasingly short of information-theoretic limits as the degree $D$ of the parity checks increases. Thus increasing $D$ generically makes the problem harder both for DQI+BP and for classical optimization heuristics.

Despite this challenge, we are able to find some unstructured families of sparse max-XORSAT instances for which DQI with standard belief propagation decoding finds solutions satisfying a fraction of constraints that is very difficult to replicate using simulated annealing. We do so by tuning the degree distribution of the instances. For example, in \sect{sec:wins}, we generate an example max-XORSAT instance from our specified degree distribution, which has $31,216$ variables and $50,000$ constraints, where each constraint contains an average of 53.973 variables and each variable is contained in an average of 86.451 constraints. We find that DQI with standard belief propagation decoding can find solutions in which the fraction of constraints satisfied is at least $0.831$. It does so by reducing this problem to a decoding problem that can be solved by our implementation of belief propagation in 8 seconds, excluding the time needed to load and parse the instance. In contrast, our implementation of simulated annealing requires approximately 73 hours to reach this, even when allowed five cores in parallel; restricted to 8 seconds of runtime on a single core it is only able to satisfy $0.764$.

Furthermore, as shown in Table \ref{tab:xorboard}, DQI+BP achieves higher satisfaction fraction than we are able to obtain in a comparable number of computational steps using any of the general-purpose classical optimization algorithms that we tried. However, unlike our OPI example, we do not put this forth as an example of quantum advantage. Rather, we are able to construct a tailored classical algorithm specialized to these instances which, within seven minutes of runtime, finds solutions where the fraction of constraints satisfied is $0.88$, thereby slightly beating DQI+BP.

\section{Relation to Other Work}
\label{sec:prior}

DQI is related to a family of quantum reductions that originate with the work of Aharonov, Ta-Shma, and Regev \cite{ATS03, AR05}. In this body of work the core idea is to use the Fourier convolution theorem to obtain reductions between nearest lattice vector problems and shortest lattice vector problems. In this section we summarize the other quantum algorithms using this idea and discuss their relationship to DQI.

In \cite{CLZ22}, Chen, Liu, and Zhandry introduce a novel and powerful intrinsically-quantum decoding method that they call filtering, which can in some cases solve quantum decoding problems for which the analogous classical decoding problems cannot be solved by any known efficient classical algorithm. By combining this decoding method with a Regev-style reduction, they are able to efficiently find approximate optima to certain shortest vector problems defined using the infinity norm for which no polynomial-time classical solution is known. Due to the use of the infinity norm, the problem solved by the quantum algorithm of \cite{CLZ22} is one of satisfying all constraints, rather than the more general problem of maximizing the number of constraints satisfied, as is addressed by DQI. Although conceptualized differently, and implemented over the space of codewords rather than the space of syndromes, the algorithm of \cite{CLZ22} is similar in spirit to DQI and can be regarded as a foundational prior work. It is particularly noteworthy that the apparent superpolynomial quantum speedup achieved in \cite{CLZ22}, through the use of quantum decoders, is obtained for a purely geometrical problem with no algebraic structure.

In \cite{YZ22}, Yamakawa and Zhandry define an oracle problem that they prove can be solved using polynomially many quantum queries but requires exponentially many classical queries. Their problem is essentially a class of instances of max-LINSAT over an exponentially large finite field $\mathbb{F}_{2^t}$, where the functions $f_1,\ldots,f_m$ are defined by random oracles and the matrix $B$ has algebraic structure. In \sect{sec:folded} we recount the exact definition of the Yamakawa-Zhandry problem and show how DQI can be extended to the Yamakawa-Zhandry problem. In the problem defined by Yamakawa and Zhandry, the truth tables for the constraints are defined by random oracles, and therefore, in the language of Theorem \ref{thm:semicircle}, $r/p = 1/2$. Furthermore, the Yamakawa-Zhandry instance is designed such that $C^\perp$ can be decoded with exponentially small failure probability if half of the symbols are corrupted by errors. Thus, in the language of Theorem \ref{thm:semicircle}, we can take $\ell/m = 1/2$. In this case, if we extrapolate \eq{eq:semicircle_general} to the Yamakawa-Zhandry regime, one obtains $\langle s \rangle/m = 1$, indicating that DQI should find a solution satisfying all constraints. The quantum algorithm given by Yamakawa and Zhandry for finding a solution satisfying all constraints is different from DQI, but similar in spirit.

Extrapolation of the semicircle law is necessary because the Yamakawa-Zhandry example does not satisfy $2 \ell + 1 < d^\perp$, and therefore the conditions of Theorem \ref{thm:semicircle} are not met. In Theorem \ref{thm:imperfect_decoding} we prove a variant of the semicircle law for $2 \ell + 1 > d^\perp$ for max-XORSAT with average case $\mathbf{v}$. The Yamakawa-Zhandry problem uses a random oracle, which is the direct generalization to $\mathbb{F}_q$ of average case $\mathbf{v}$. It seems likely that the proof of Theorem \ref{thm:imperfect_decoding} can be generalized to from $\mathbb{F}_2$ to $\mathbb{F}_q$ with arbitrary $q$ at the cost of additional technical complications but without the need for conceptual novelty. In this case the claim that DQI encompasses the Yamakawa-Zhandry upper bound on quantum query complexity could be made rigorous.

An important difference between our OPI results and the exponential quantum query complexity speedup of Yamakawa and Zhandry is that the latter depends on $\mathbb{F}_q$ being an exponentially large finite field. This allows Yamakawa and Zhandry to obtain an information-theoretic exponential classical query-complexity lower bound, thus making their separation rigorous. In contrast, our apparent superpolynomial speedup for OPI uses a finite field of polynomial size, in which case the truth tables are polynomial size and known explicitly. This regime is more relevant to real-world optimization problems. The price we pay is that the classical query complexity in this regime is necessarily polynomial. This means that the information-theoretic argument of \cite{YZ22} establishing an exponential improvement in query complexity does not apply in our setting. Instead, we argue heuristically for an superpolynomial quantum speedup, by comparing to known classical algorithms.

In \cite{DRT23}, Debris-Alazard, Remaud, and Tillich construct a quantum reduction from the approximate shortest codeword problem on a code $C$ to the bounded distance decoding problem on its dual $C^\perp$. The shortest codeword problem is closely related to max-XORSAT: the max-XORSAT problem is to find the codeword in $C = \{B \mathbf{x} \,|\, \mathbf{x} \in \mathbb{F}_2^n\}$ that has smallest Hamming distance from a given bit string $\mathbf{v}$, whereas the shortest codeword problem is to find the codeword in $C\setminus\{\mathbf{0}\}$ with smallest Hamming weight. Roughly speaking, then, the shortest codeword problem is the special case of max-XORSAT where $\mathbf{v}=\mathbf{0}$ except that the trivial solution $\mathbf{0}$ is excluded. Although conceptualized quite differently, the reduction of \cite{DRT23} is similar to the $p=2$ special case of DQI in that it is achieved via a quantum Hadamard transform.

Another interesting connection between DQI and prior work appears in the context of planted inference problems, such as planted $k$XOR, where the task is to recover a secret good assignment planted in an otherwise random $k$XOR instance. It has been recently shown that quantum algorithms can achieve polynomial speedups for planted $k$XOR by efficiently constructing a \emph{guiding state} that has improved overlap with the planted solution \cite{schmidhuber2025quartic}. Curiously, the $\ell\th$ order guiding state studied in \cite{schmidhuber2025quartic} seems related to the $\ell\th$ order DQI state presented here. The key conceptual difference is that, to obtain an assignment that satisfies a large number of constraints, the DQI state in this work is measured in the computational basis, whereas the guiding state in \cite{schmidhuber2025quartic} is measured in the eigenbasis of the so-called Kikuchi matrix, and subsequently rounded. It is interesting that for the large values of $\ell$ studied in this work, the (Fourier-transformed) Kikuchi matrix asymptotically approaches a diagonal matrix such that its eigenbasis is close to the computational basis. 

Next, we would like to highlight \cite{chailloux2024quantum} in which Chailloux and Tillich introduce a method for decoding superpositions of errors, which they use in the context of a Regev-style reduction. Their method to solve the quantum decoding problem is based on a technique called \emph{unambiguous state discrimination} (USD). Roughly speaking, USD is a coherent quantum rotation which allows us to convert bit-flip error into erasure error. If the resulting erasure error rate is small enough, one can decode perfectly using Gaussian elimination. Interestingly, this is efficient regardless of the sparsity of the code. Curiously, however, Chailloux and Tillich find that the resulting performance of this combination of quantum algorithms yields performance for the shortest codeword problem that exactly matches that of Prange's algorithm \cite{P62}.

Lastly, we note that since we posted the first version of this manuscript as an arXiv preprint in August, 2024 there have already been some works that nicely build upon DQI. In particular, we would like to highlight \cite{CT24} in which Chailloux and Tillich observe that for $p > 2$ the distribution over errors on a given corrupted symbol in the decoding problem faced by DQI is non-uniform. Information-theoretically this is more advantageous than uniform errors. Furthermore, they show that this advantage can be exploited by polynomial-time decoders, thereby improving the approximation ratio efficiently achievable by DQI. Additionally, in \cite{PBH25}, a concrete resource analysis is carried out for DQI in which the decoding of $C^\perp$ is carried out using a reversible circuit implementing information set decoding.

\section{Decoded Quantum Interferometry}
\label{sec:DQI}

Here we describe in full detail the Decoded Quantum Interferometry algorithm, illustrated in Fig. \ref{fig:dqi_algo}. We start with DQI for max-XORSAT, which is the special case of max-LINSAT with $\mathbb{F} = \mathbb{F}_2$. Then we describe DQI with $\mathbb{F} = \mathbb{F}_p$ for any prime $p$. In each case, our explanation consists of a discussion of the quantum state we intend to create followed by a description of the quantum algorithm used to create it. The generalization of DQI to extension fields is given in \sect{sec:folded}.

\begin{figure}
    \centering
    \begin{tikzpicture}[domain=0:25, scale=.7]
    \draw[gray] (17.5, 17.05) rectangle ++(6, 9.35);
    \node[gray, anchor=north west] at (17.75, 26.25) {Legend:};
    \draw[thick, brown, fill=brown!10, rounded corners] (18, 23.5) rectangle ++(5, 1.75) node[pos=.5, align=center] {quantum register \\ allocations};
    \draw[thick, violet, fill=violet!10, rounded corners] (18, 21.5) rectangle ++(5, 1.75) node[pos=.5, align=center] {unitary \\ operations};
    \draw[thick, teal, fill=teal!10] (18, 19.5) rectangle ++(5, 1.75) node[pos=.5, align=center] {pure quantum \\ states};
    \draw[thick, darkgray, fill=darkgray!10, rounded corners] (18, 17.5) rectangle ++(5, 1.75) node[pos=.5, align=center] {non-unitary \\ operations};
    
    \draw[->, gray, thick] (1.5, 6) -- node[right] {time} (1.5, 2);
    
    \draw[thick, brown!30] (2.5, 14.5) -- (2.5, 26);
    \draw[thick, brown, fill=brown!10, rounded corners] (0, 24.75) rectangle ++(5, 1.75) node[pos=.5, align=center] {weight register \\ $\lceil\log_2\ell\rceil$ qubits};
    
    \draw[thick, brown!30] (10.25, 5) -- (10.25, 22);
    \draw[thick, brown, fill=brown!10, rounded corners] (5.5, 21) rectangle ++(9.5, 1.75) node[pos=.5, align=center] {error register \\ $m$ qubits};
    
    \draw[thick, brown!30] (19.5, 0) -- (19.5, 12);
    \draw[thick, brown, fill=brown!10, rounded corners] (15.5, 11.25) rectangle ++(8, 1.75) node[pos=0.5, align=center] {syndrome register \\ $n$ qubits};
    
    \draw[thick, violet, fill=violet!10, rounded corners] (0, 22.5) rectangle ++(5, 1.75) node[pos=.5, align=center] {embed $w_k$ \\ into amplitudes};
    \draw[thick, teal, fill=teal!10] (0, 21) rectangle ++(5, 1) node[pos=.5, align=center] {$\sum_k w_k |k\rangle$};
    
    \draw[thick, violet, fill=violet!10, rounded corners] (0, 19.5) rectangle ++(15, 1) node[pos=.5, align=center] {prepare dependent Dicke state};
    \draw[thick, teal, fill=teal!10] (0, 18) rectangle ++(15, 1) node[pos=.5, align=center] {$\sum_k w_k {m \choose k}^{-1/2} \sum_{\mathbf{y}} |k\rangle  |\mathbf{y}\rangle$};
    
    \draw[thick, violet, fill=violet!10, rounded corners] (0, 15.75) rectangle ++(15, 1.75) node[pos=.5, align=center] {calculate Hamming weight \\ to uncompute weight register};
    \draw[thick, teal, fill=teal!10] (5.5, 14.25) rectangle ++(9.5, 1) node[pos=.5, align=center] {$\sum_k w_k {m \choose k}^{-1/2} \sum_{\mathbf{y}} |\mathbf{y}\rangle$};
    
    \draw[thick, violet, fill=violet!10, rounded corners] (5.5, 12.75) rectangle ++(9.5, 1) node[pos=.5, align=center] {$Z^{v_1}\otimes\ldots\otimes Z^{v_m}$};
    \draw[thick, teal, fill=teal!10] (5.5, 11.25) rectangle ++(9.5, 1) node[pos=.5, align=center] {$\sum_k w_k {m \choose k}^{-1/2} \sum_{\mathbf{y}} (-1)^{\mathbf{v}\cdot\mathbf{y}} |\mathbf{y}\rangle$};
    
    \draw[thick, violet, fill=violet!10, rounded corners] (5.5, 9.75) rectangle ++(18, 1) node[pos=.5, align=center] {multiply by matrix $B^T$};
    \draw[thick, teal, fill=teal!10] (5.5, 8.25) rectangle ++(18, 1) node[pos=.5, align=center] {$\sum_k w_k {m \choose k}^{-1/2} \sum_{\mathbf{y}} (-1)^{\mathbf{v}\cdot\mathbf{y}} |\mathbf{y}\rangle|B^T\mathbf{y}\rangle$};
    
    \draw[thick, violet, fill=violet!10, rounded corners] (5.5, 6) rectangle ++(18, 1.75) node[pos=.5, align=center] {solve bounded distance decoding \\ problem to uncompute error register};
    \draw[thick, teal, fill=teal!10] (15.5, 4.5) rectangle ++(8, 1) node[pos=.5, align=center] {$\sum_k w_k {m \choose k}^{-1/2} \sum_{\mathbf{y}} (-1)^{\mathbf{v}\cdot\mathbf{y}} |B^T\mathbf{y}\rangle$};
    
    \draw[thick, violet, fill=violet!10, rounded corners] (15.5, 3) rectangle ++(8, 1) node[pos=.5, align=center] {$H\otimes\ldots\otimes H$};
    \draw[thick, teal, fill=teal!10] (15.5, 1.5) rectangle ++(8, 1) node[pos=.5, align=center] {$\sum_{\mathbf{x} \in \mathbb{F}_2^n} P(f(\mathbf{x})) |\mathbf{x}\rangle$};
    
    \draw[thick, darkgray, fill=darkgray!10, rounded corners] (0, 14.25) rectangle ++(5, 1) node[pos=.5, align=center] {discard $|0\rangle$};
    \draw[thick, darkgray, fill=darkgray!10, rounded corners] (5.5, 4.5) rectangle ++(9.5, 1) node[pos=.5, align=center] {postselect on $|0\rangle$};
    \draw[thick, darkgray, fill=darkgray!10, rounded corners] (15.5, 0) rectangle ++(8, 1) node[pos=.5, align=center] {measure};
    
    \end{tikzpicture}
    
    \caption{Decoded Quantum Interferometry over $\mathbb{F}_2$. The algorithm begins with a computation, on a classical computer, of the principal eigenvector $(w_0,\ldots,w_\ell)^T$ of a certain matrix $A^{(m,\ell,0)}$. $P(f)$ is a degree-$\ell$ polynomial that enhances the probability of sampling $\mathbf{x}\in\mathbb{F}_2^n$ with high value of the objective $f$. Index $k$ ranges over $\{0,\ldots,\ell\}$ and $\mathbf{y}$ over the set $\{\mathbf{y} \in \mathbb{F}_2^m:|\mathbf{y}| = k\}$ of $m$-bit strings of Hamming weight $k$. Weight register qubits may be reused for the syndrome register. If $2 \ell + 1 < d^\perp$, the postselection succeeds with probability $\geq 1 - \varepsilon_\ell$, where $\varepsilon_\ell$ is the decoding failure rate on random weight-$\ell$ errors.}
    \label{fig:dqi_algo}
    \end{figure}

In this section, we assume throughout that $2 \ell + 1 < d^\perp$. In some cases we find that it is possible to decode more than $d^\perp/2$ errors with high probability, in which case we can use $\ell$ exceeding this bound. This contributes some technical complications to the description and analysis of the DQI algorithm, which we defer to \sect{sec:beyond_2l}.

\subsection{DQI for max-XORSAT}
\label{sec:genl}

\subsubsection{DQI Quantum State for max-XORSAT}
\label{sec:genl_dqi_state}

Recall that the objective function for max-XORSAT can be written in the form $f = f_1 + \ldots + f_m$, where $f_i$ takes the value $+1$ if the $i\th$ constraint is satisfied and $-1$ otherwise. More concretely, the function defined by
\begin{eqnarray}
    f(\mathbf{x}) & = & \sum_{i=1}^m f_i(\mathbf{b}_i \cdot \mathbf{x}) \label{eq:max1} \\
    f_i(\mathbf{b}_i\cdot\mathbf{x}) & = & (-1)^{v_i} (-1)^{\mathbf{b}_i \cdot \mathbf{x}}, \label{eq:max2}
\end{eqnarray}
expresses the number of constraints satisfied minus the the number of constraints unsatisfied by the bit string $\mathbf{x}$.

Given an objective of the form $f = f_1 + \ldots + f_m$ and a degree-$\ell$ univariate polynomial
\begin{equation}
    P(f) = \sum_{k = 0}^\ell \alpha_k f^k,
\end{equation}
we can regard $P(f)$ as a degree-$\ell$ multivariate polynomial in $f_1,\ldots,f_m$. Since $f_1,\ldots,f_m$ are $\pm1$-valued we have $f_i^2 = 1$ for all $i$. Consequently, we can express $P(f)$ as
\begin{equation}
    \label{eq:Psym}
    P(f) = \sum_{k = 0}^\ell u_k P^{(k)}(f_1,\ldots,f_m),
\end{equation}
where $P^{(k)}(f_1,\ldots,f_m)$ is the degree-$k$ elementary symmetric polynomial, \textit{i.e.} the unweighted sum of all $\binom{m}{k}$ products of $k$ distinct factors from $f_1,\ldots,f_m$.

For simplicity, we will henceforth always assume that the overall normalization of $P$ has been chosen so that the state $\ket{P(f)} = \sum_{\mathbf{x} \in \mathbb{F}_2^n} P(f(\mathbf{x})) \ket{\mathbf{x}}$ has unit norm. In analogy with equation \eqref{eq:Psym}, we will write the DQI state $\ket{P(f)}$ as a linear combination of $\ket{P^{(0)}},\ldots,\ket{P^{(\ell)}}$, where
\begin{equation}
    \label{eq:definition_of_ket_for_elementary_symmetric_polynomial}
    \ket{P^{(k)}} := \frac{1}{\sqrt{2^n \binom{m}{k}}} \sum_{\mathbf{x} \in \mathbb{F}_2^n} P^{(k)}(f_1(\mathbf{b}_1 \cdot \mathbf{x}),\ldots,f_m(\mathbf{b}_m \cdot \mathbf{x})) \ket{\mathbf{x}}.
\end{equation}
Thus, by \eq{eq:max2},
\begin{eqnarray}
    P^{(k)}(f_1,\ldots,f_m) & = & \sum_{\substack{i_1,\ldots,i_k \\ \mathrm{distinct}}} f_{i_1} \times \ldots \times f_{i_k} \\
    & = & \sum_{\substack{i_1,\ldots,i_k \\ \mathrm{distinct}}} (-1)^{v_{i_1} + \ldots + v_{i_k}} (-1)^{(\mathbf{b}_{i_1} + \ldots + \mathbf{b}_{i_k}) \cdot \mathbf{x}} \\
    & = & \sum_{\substack{\mathbf{y} \in \mathbb{F}_2^m \\ | \mathbf{y}| = k}} (-1)^{\mathbf{v} \cdot \mathbf{y}} (-1)^{(B^T \mathbf{y}) \cdot \mathbf{x}},
\end{eqnarray}
where $|\mathbf{y}|$ indicates the Hamming weight of $\mathbf{y}$. From this we see that the Hadamard transform of $\ket{P^{(k)}}$ is
\begin{equation}
    \ket{\widetilde{P}^{(k)}} := H^{\otimes n} \ket{P^{(k)}} = \frac{1}{\sqrt{\binom{m}{k}}} \sum_{\substack{\mathbf{y} \in \mathbb{F}_2^m \\ | \mathbf{y}| = k}} (-1)^{\mathbf{v} \cdot \mathbf{y}} \ket{B^T\mathbf{y}}. \label{eq:had2}
\end{equation}
For the set of $\mathbf{y} \in \mathbb{F}_2^n$ with $|\mathbf{y}| < d^\perp/2$ the corresponding bit strings $B^T \mathbf{y}$ are all distinct. Therefore, $\ket{\widetilde{P}^{(0)}}, \ldots, \ket{\widetilde{P}^{(\ell)}}$ form an orthonormal set provided $\ell < d^\perp/2$, and so do $\ket{P^{(0)}}, \ldots, \ket{P^{(\ell)}}$. In this case, the DQI state $\ket{P(f)} = \sum_{\mathbf{x} \in \mathbb{F}_2^n} P(f(\mathbf{x})) \ket{\mathbf{x}}$ can be expressed as
\begin{equation}
    \ket{P(f)} = \sum_{k=0}^\ell w_k \ket{P^{(k)}},
\end{equation}
where
\begin{equation}
    w_k = u_k \sqrt{2^n \binom{m}{k}},
\end{equation}
and $\langle P(f)|P(f)\rangle = \|\mathbf{w}\|^2$ where $\mathbf{w}:=(w_0,\ldots,w_\ell)^T$.

\subsubsection{DQI Algorithm for max-XORSAT}
\label{sec:genl_dqi_algo}

With these facts in hand, we now present the DQI algorithm for max-XORSAT. The algorithm utilizes three quantum registers: a \textit{weight register} comprising $\lceil\log_2\ell\rceil$ qubits, an \textit{error register} with $m$ qubits, and a \textit{syndrome register} with $n$ qubits.

As the first step in DQI we initialize the weight register in the state
\begin{equation}
    \sum_{k=0}^\ell w_k \ket{k}.
\end{equation}
The choice of $\mathbf{w} \in \mathbb{R}^{\ell+1}$ that maximizes the expected number of satisfied constraints can be obtained by solving for the principal eigenvector of an $(\ell+1) \times (\ell+1)$ matrix, as described in \sect{sec:semicircle}. Given $\mathbf{w}$, this state preparation can be done efficiently because it is a state of only $\lceil\log_2 \ell\rceil$ qubits. One way to do this is to use the method from \cite{LKS18} that prepares an arbitrary superposition over $\ell$ computational basis states using $\widetilde{O}(\ell)$ quantum gates.

Next, conditioned on the value $k$ in the weight register, we prepare the error register into the uniform superposition over all bit strings of Hamming weight $k$
\begin{equation}
    \to \sum_{k=0}^\ell w_k \ket{k} \frac{1}{\sqrt{\binom{m}{k}}} \sum_{\substack{\mathbf{y} \in \mathbb{F}_2^m \\ |\mathbf{y}| = k}} \ket{\mathbf{y}}.
\end{equation}
Efficient methods for preparing such superpositions using $O(\ell m)$ quantum gates have been devised due to applications in physics, where they are known as Dicke states \cite{BE22, WT24}. Next, we uncompute $\ket{k}$ which can be done easily since $k$ is simply the Hamming weight of $\mathbf{y}$. After doing so and discarding the weight register, we are left with the state
\begin{equation}
    \to \sum_{k=0}^\ell w_k \frac{1}{\sqrt{\binom{m}{k}}} \sum_{\substack{\mathbf{y} \in \mathbb{F}_2^m \\ |\mathbf{y}| = k}} \ket{\mathbf{y}}.
\end{equation}
Next, we apply the phases $(-1)^{\mathbf{v} \cdot \mathbf{y}}$ by performing a Pauli-$Z_i$ on each qubit for which $v_i = 1$, at the cost of $O(m)$ quantum gates
\begin{equation}
    \to \sum_{k=0}^\ell w_k \frac{1}{\sqrt{\binom{m}{k}}} \sum_{|\mathbf{y}| = k} (-1)^{\mathbf{v} \cdot \mathbf{y}} \ket{\mathbf{y}}.
\end{equation}
Then, we reversibly compute $B^T \mathbf{y}$ into the syndrome register by standard matrix-vector multiplication over $\mathbb{F}_2$ at the cost $O(mn)$ quantum gates
\begin{equation}
    \to \sum_{k=0}^\ell w_k \frac{1}{\sqrt{\binom{m}{k}}} \sum_{|\mathbf{y}| = k} (-1)^{\mathbf{v} \cdot \mathbf{y}} \ket{\mathbf{y}} \ket{B^T \mathbf{y}}.
\end{equation}
From this, to obtain $\ket{\widetilde{P}^{(k)}}$ we need to use the content $\mathbf{s} = B^T \mathbf{y}$ of the syndrome register to reversibly uncompute the content $\mathbf{y}$ of the  error register, in superposition. This is the hardest part. We next discuss how to compute $\mathbf{y}$ from $\mathbf{s}$.

Recall that $B \in \mathbb{F}_2^{m \times n}$ with $m > n$. Thus, given $\mathbf{s}$, solving $\mathbf{s} = B^T \mathbf{y}$, for $\mathbf{y}$ is an underdetermined linear algebra problem over $\mathbb{F}_2$. It is only the constraint $|\mathbf{y}| \leq \ell$ that renders the solution unique, provided $\ell$ is not too large. This linear algebra problem with a Hamming weight constraint is recognizable as syndrome decoding. The kernel of $B^T$ defines an error correcting code $C^\perp$. The string $\mathbf{s}$ is interpreted as the error syndrome, and $\mathbf{y}$ is interpreted as the error. If $\ell$ is less than half the minimum distance of $C^\perp$ then the problem of inferring the error from the syndrome has a unique solution. When this solution can be found efficiently, we can efficiently uncompute the content of the error register, which can then be discarded. This leaves
\begin{equation}
    \to \sum_{k=0}^\ell w_k \frac{1}{\sqrt{\binom{m}{k}}} \sum_{|\mathbf{y}| = k} (-1)^{\mathbf{v} \cdot \mathbf{y}} \ket{B^T \mathbf{y}}
\end{equation}
in the syndrome register which we recognize as
\begin{equation}
    = \sum_{k=0}^\ell w_k \ket{\widetilde{P}^{(k)}}.
\end{equation}
By taking the Hadamard transform we then obtain
\begin{equation}
    \to \sum_{k=0}^\ell w_k \ket{P^{(k)}}
\end{equation}
which is the desired DQI state $\ket{P(f)}$.

\subsection{DQI for General max-LINSAT}
\label{sec:genp}

\subsubsection{DQI Quantum State for General max-LINSAT}
\label{sec:genp_dqi_state}

Recall that the max-LINSAT objective takes the form $f(\mathbf{x}) = \sum_{i=1}^m f_i(\mathbf{b}_i \cdot \mathbf{x})$ with $f_i:\mathbb{F}_p \to \{+1,-1\}$. In order to keep the presentation as simple as possible, we restrict attention to the situation where the preimages $F_i:=f_i^{-1}(+1)$ for $i=1,\ldots,m$ have the same cardinality $r:=|F_i|\in\{1,\ldots,p-1\}$. We will find it convenient to work in terms of $g_i$ which we define as $f_i$ shifted and rescaled so that its Fourier transform
\begin{equation}
\label{eq:tildeg}
\tilde{g}_i(y) = \frac{1}{\sqrt{p}} \sum_{x \in \mathbb{F}_p} \omega_p^{y x} g_i(x)
\end{equation}
vanishes at $y=0$ and is normalized, \textit{i.e.}  $\sum_{x\in\mathbb{F}_p}|g_i(x)|^2=\sum_{y\in\mathbb{F}_p}|\tilde{g}_i(y)|^2=1$. (Here and throughout, $\omega_p = e^{i 2 \pi /p}$.) In other words, rather than using $f_i$ directly, we will use
\begin{equation}
    \label{eq:gf}
    g_i(x) := \frac{f_i(x) - \overline{f}}{\varphi}
\end{equation}
where $\overline{f} := \frac{1}{p}\sum_{x\in\mathbb{F}_p} f_i(x)$ and $\varphi := \left(\sum_{y\in\mathbb{F}_p}|f_i(y)-\overline{f}|^2\right)^{1/2}$. Explicitly, one finds
\begin{eqnarray}
\bar{f} & = & \frac{2 r}{p} - 1 \label{eq:fbar}\\
\varphi & = & \sqrt{4 r \left( 1 - \frac{r}{p} \right)}. \label{eq:varphi}
\end{eqnarray}

By \eq{eq:gf}, the sums
\begin{eqnarray}
f(\mathbf{x}) & = & f_1(\mathbf{b}_1 \cdot \mathbf{x}) + \ldots + f_m(\mathbf{b}_m \cdot \mathbf{x}) \\
g(\mathbf{x}) & = & g_1(\mathbf{b}_1 \cdot \mathbf{x}) + \ldots + g_m(\mathbf{b}_m \cdot \mathbf{x})
\end{eqnarray}
are related according to
\begin{equation}
f(\mathbf{x}) = g(\mathbf{x}) \varphi + m \overline{f}.
\end{equation}
We transform the polynomial $P(f(\mathbf{x}))$ into an equivalent polynomial $Q(g(\mathbf{x}))$ of the same degree by substituting in this relation and absorbing the relevant powers of $\varphi$ and $m \overline{f}$ into the coefficients. That is, $Q(g(\mathbf{x})) = P(f(\mathbf{x}))$. As shown in Appendix \ref{app:esp}, $Q(g(\mathbf{x}))$ can always be expressed as a linear combination of elementary symmetric polynomials. That is,
\begin{equation}
    \label{eq:Qsym}
    Q(g(\mathbf{x})) := \sum_{k=0}^\ell u_k P^{(k)}(g_1(\mathbf{b}_1 \cdot \mathbf{x}), \dots, g_m(\mathbf{b}_m \cdot \mathbf{x})).
\end{equation}

As in the case of DQI for max-XORSAT above, we write the DQI state
\begin{equation}
    \ket{P(f)} = \sum_{\mathbf{x} \in \mathbb{F}_p^n} P(f(\mathbf{x})) \ket{\mathbf{x}} = \sum_{\mathbf{x} \in \mathbb{F}_p^n} Q(g(\mathbf{x})) \ket{\mathbf{x}} = \ket{Q(g)}
\end{equation}
as a linear combination of $\ket{P^{(0)}},\ldots,\ket{P^{(\ell)}}$ defined as
\begin{equation}
    \label{eq:normalized}
    \ket{P^{(k)}} := \frac{1}{\sqrt{p^{n-k} \binom{m}{k}}} \sum_{\mathbf{x} \in \mathbb{F}_p^n} P^{(k)}(g_1(\mathbf{b}_1 \cdot \mathbf{x}),\ldots,g_m(\mathbf{b}_m \cdot \mathbf{x})) \ket{\mathbf{x}}.
\end{equation}
By definition,
\begin{equation}
    \label{eq:Pk}
    P^{(k)}(g_1(\mathbf{b}_1 \cdot \mathbf{x}),\ldots,g_m(\mathbf{b}_m \cdot \mathbf{x})) = \sum_{\substack{i_1,\ldots,i_k \\ \mathrm{distinct}}} \prod_{i \in \{i_1,\ldots,i_k\}} g_i(\mathbf{b}_i \cdot \mathbf{x}).
\end{equation}
Substituting \eq{eq:tildeg} into \eq{eq:Pk} yields
\begin{eqnarray}
    P^{(k)}(g_1(\mathbf{b}_1 \cdot \mathbf{x}),\ldots,g_m(\mathbf{b}_m \cdot \mathbf{x})) & = & \sum_{\substack{i_1,\ldots,i_k \\ \mathrm{distinct}}} \prod_{i \in \{i_1,\ldots,i_k\}} \left( \frac{1}{\sqrt{p}} \sum_{y_i \in \mathbb{F}_p} \omega_p^{ -y_i \mathbf{b}_i \cdot \mathbf{x}} \ \tilde{g}_i(y_i) \right) \\
    & = &  \sum_{\substack{\mathbf{y} \in \mathbb{F}_p^m \\ |\mathbf{y}| = k }} \frac{1}{\sqrt{p^k}} \  \omega_p^{-(B^T \mathbf{y}) \cdot \mathbf{x}} \ \prod_{\substack{i=1 \\ y_i \neq 0}}^m \tilde{g}_i(y_i).
\end{eqnarray}
From this we see that the Quantum Fourier Transform of $\ket{P^{(k)}}$ is
\begin{equation}
    \label{eq:FTgenp}
    |\widetilde{P}^{(k)}\rangle := F^{\otimes n}|P^{(k)}\rangle = \frac{1}{\sqrt{\binom{m}{k}}} \sum_{\substack{\mathbf{y}\in\mathbb{F}_p^m\\|\mathbf{y}|=k}} \left( \prod_{\substack{i=1\\y_i \neq 0}}^m \tilde{g}_i(y_i) \right)
    |B^T \mathbf{y}\rangle
\end{equation}
where $F_{ij}=\omega_p^{ij}/\sqrt{p}$ with $i,j=0,\dots,p-1$. As in the case of max-XORSAT earlier, if $|\mathbf{y}| < d^\perp/2$, then $B^T \mathbf{y}$ are all distinct. Therefore, if $\ell < d^\perp/2$ then $\ket{\widetilde{P}^{(0)}},\ldots,\ket{\widetilde{P}^{(\ell)}}$ form an orthonormal set and so do $\ket{P^{(0)}},\ldots,\ket{P^{(\ell)}}$. Thus, the DQI state $\ket{P(f)} = \sum_{\mathbf{x} \in \mathbb{F}_p^n} P(f(\mathbf{x})) \ket{\mathbf{x}}$ can be expressed as
\begin{equation}
    \ket{P(f)} = \sum_{k=0}^\ell w_k \ket{P^{(k)}}
\end{equation}
where
\begin{equation}
    w_k = u_k\sqrt{p^{n-k} \binom{m}{k}},
\end{equation}
and $\langle P(f)|P(f)\rangle = \|\mathbf{w}\|^2$.

\subsubsection{DQI Algorithm for General max-LINSAT}
\label{sec:genp_dqi_algo}

DQI for general max-LINSAT employs three quantum registers: a \textit{weight register} comprising $\lceil\log_2\ell\rceil$ qubits, an \textit{error register} with $m\lceil\log_2 p\rceil$ qubits, and a \textit{syndrome register} with $n\lceil\log_2 p\rceil$ qubits. We will consider the error and syndrome registers as consisting of $m$ and $n$ subregisters, respectively, where each subregister consists of $\lceil\log_2 p\rceil$ qubits\footnote{One can also regard each subregister as a logical $p$-level quantum system, \textit{i.e.} a qudit, encoded in $\lceil\log_2 p\rceil$ qubits.}. We will also regard the ordered collection of least significant\footnote{Least significant qubit is the one which is in the state $\ket{1}$ when the subregister stores the value $1\in\mathbb{F}_p$.} qubits from the $m$ subregisters of the error register as the fourth register. We will refer to this $m$ qubit register as the \textit{mask register}.

The task of DQI is to produce the state $\ket{Q(g)}$, which by construction is equal to $\ket{P(f)}$. We proceed as follows. First, as in the $p=2$ case, we initialize the weight register in the normalized state
\begin{equation}
    \label{eq:start}
    \sum_{k=0}^\ell w_k \ket{k}.
\end{equation}
Next, conditioned on the value $k$ in the weight register, we prepare the mask register in the corresponding Dicke state
\begin{equation}
    \to \sum_{k=0}^\ell w_k \ket{k} \frac{1}{\sqrt{\binom{m}{k}}} \sum_{\substack{\boldsymbol{\mu} \in \{0,1\}^m \\ |\boldsymbol{\mu}| = k}} \ket{\mathbf{\boldsymbol{\mu}}}.
\end{equation}
Then, we uncompute $|k\rangle$ using the fact that $k=|\boldsymbol{\mu}|$, which yields
\begin{equation}
    \to \sum_{k=0}^\ell w_k \frac{1}{\sqrt{\binom{m}{k}}} \sum_{\substack{\boldsymbol{\mu} \in \{0,1\}^m \\ |\boldsymbol{\mu}| = k}} \ket{\mathbf{\boldsymbol{\mu}}}.
\end{equation}
Subsequently, we will turn the content of the error register from a superposition of \textit{bit} strings $\boldsymbol{\mu}\in\{0,1\}^m$ of Hamming weight $k$ into a superposition of \textit{symbol} strings $\mathbf{y}\in\mathbb{F}_p^m$ of Hamming weight $k$. Let $G_i$ denote an operator acting on $\lceil\log_2 p\rceil$ qubits such that
\begin{equation}\label{eq:g_i}
    G_i\ket{0} = \ket{0},\quad\quad G_i\ket{1} = \sum_{y\in\mathbb{F}_p} \tilde{g}_i(y) \ket{y}.
\end{equation}
By reparametrizing from $f_1,\ldots,f_m$ to $g_1,\ldots,g_m$ we ensured that $\tilde{g}_i(0) = 0$ for all $i$, so
\begin{equation}\label{eq:g_i_1}
    G_i\ket{1} = \sum_{y\in\mathbb{F}^*_p} \tilde{g}_i(y) \ket{y}
\end{equation}
where $\mathbb{F}_p^* = \mathbb{F}_p\setminus\{0\}$. This guarantees that $G_i\ket{0}=\ket{0}$ is orthogonal to $G_i\ket{1}$, so we may assume that $G_i$ is unitary. It may be realized, \textit{e.g.} using the techniques from \cite{LKS18}. The fact that $\tilde{g}_i(0) = 0$ also implies that the parallel application $G:=\prod_{i=1}^mG_i$ of each $G_i$ to the respective subregister of the error register preserves the Hamming weight. More precisely, the symbol string on every term in the expansion of $G\ket{\boldsymbol{\mu}}$ in the computational basis has the same Hamming weight as $\boldsymbol{\mu}$. Indeed, using \eqref{eq:g_i} and \eqref{eq:g_i_1}, we find
\begin{equation}
    \sum_{\substack{\boldsymbol{\mu} \in \{0,1\}^m \\ |\boldsymbol{\mu}| = k}} G \ket{\boldsymbol{\mu}}
    = \sum_{\substack{\mathbf{y} \in \mathbb{F}_p^m \\ |\mathbf{y}|=k}} \tilde{g}_{y(1)}\left(y_{y(1)}\right) \ldots \tilde{g}_{y(k)}\left(y_{y(k)}\right) \ket{\mathbf{y}}
\end{equation}
where $y_i$ denotes the $i\th$ entry of $\mathbf{y}$, and $y(j)$ denotes the index of the $j\th$ nonzero entry of $\mathbf{y}$. Thus, by applying $G$ to the error register, we obtain
\begin{gather}
    \to \sum_{k=0}^\ell w_k \frac{1}{\sqrt{\binom{m}{k}}} \sum_{\substack{\mathbf{y} \in \mathbb{F}_p^m \\ |\mathbf{y}|=k}} \tilde{g}_{y(1)}\left(y_{y(1)}\right) \ldots \tilde{g}_{y(k)}\left(y_{y(k)}\right) \ket{\mathbf{y}}.
\end{gather}
Next, we reversibly compute $B^T\mathbf{y}$ into the syndrome register, obtaining the state
\begin{equation}
    \to \sum_{k=0}^\ell w_k \frac{1}{\sqrt{\binom{m}{k}}} \sum_{\substack{\mathbf{y} \in \mathbb{F}_p^m \\ |\mathbf{y}|=k}} \tilde{g}_{y(1)}\left(y_{y(1)}\right) \ldots \tilde{g}_{y(k)}\left(y_{y(k)}\right) \ket{\mathbf{y}} \ket{B^T \mathbf{y}}.
\end{equation}
The task of finding $\mathbf{y}$ from $B^T \mathbf{y}$ is the bounded distance decoding problem on $C^\perp = \{ \mathbf{y} \in \mathbb{F}_p^m : B^T \mathbf{y} = \mathbf{0} \}$. Thus uncomputing the content of the error register can be done efficiently whenever the bounded distance decoding problem on $C^\perp$ can be solved efficiently out to distance $\ell$.

Uncomputing disentangles the syndrome register from the error register leaving it in the state
\begin{gather}
    \to \sum_{k=0}^\ell w_k \frac{1}{\sqrt{\binom{m}{k}}} \sum_{\substack{\mathbf{y} \in \mathbb{F}_p^m \\ |\mathbf{y}|=k}} \tilde{g}_{y(1)}\left(y_{y(1)}\right) \ldots \tilde{g}_{y(k)}\left(y_{y(k)}\right) \ket{B^T \mathbf{y} } \\
    = \sum_{k=0}^\ell w_k \frac{1}{\sqrt{\binom{m}{k}}} \sum_{\substack{\mathbf{y} \in \mathbb{F}_p^m \\ |\mathbf{y}|=k}} \left( \prod_{\substack{i=1\\y_i \neq 0}}^m \tilde{g}_i(y_i) \right) \ket{B^T \mathbf{y}}\label{eq:beforeFT}.
\end{gather}
Comparing \eq{eq:beforeFT} with \eq{eq:FTgenp} shows that this is equal to
\begin{equation}
= \sum_{k=0}^\ell w_k \ket{\widetilde{P}^{(k)}}.
\end{equation}
Thus, by applying the inverse Quantum Fourier Transform on $\mathbb{F}_p^n$ to the syndrome register we obtain
\begin{equation}
    \to \sum_{k=0}^\ell w_k \ket{P^{(k)}}. \label{eq:Fourier_transformed_DQI_state}
\end{equation}
By the definition of $w_k$ this is $\ket{Q(g)}$, which equals $\ket{P(f)}$.

\section{Optimal Expected Fraction of Satisfied Constraints}
\label{sec:semicircle}

In this section we prove Theorem \ref{thm:semicircle} for the asymptotic performance of DQI as well as quantify the finite-size corrections. Before doing so, we make two remarks. First, we note that the condition $2 \ell + 1 < d^\perp$ is equivalent to saying that $\ell$ must be less than half the distance of the code, which is the same condition needed to guarantee that the decoding problem used in the uncomputation step of DQI has a unique solution. This condition is met by our OPI example. It is not met in our irregular max-XORSAT example. In \sect{sec:beyond_2l} we show that the semicircle law \eq{eq:semicircle_general} remains a good approximation even beyond this limit, for average case $\mathbf{v}$. Second, we note that Lemma \ref{thm:genw} gives an \emph{exact} expression for the expected number of constraints satisfied by DQI at any finite size and for any choice of polynomial $P$, in terms of an $(\ell+1) \times (\ell+1)$ quadratic form. By numerically evaluating optimum of this quadratic form we find that the finite size behavior converges fairly rapidly to the asymptotic behavior of Theorem \ref{thm:semicircle}, as illustrated in Fig. \ref{fig:semicircle}.

\begin{figure}
\begin{center}
    \includegraphics[width=0.6\textwidth]{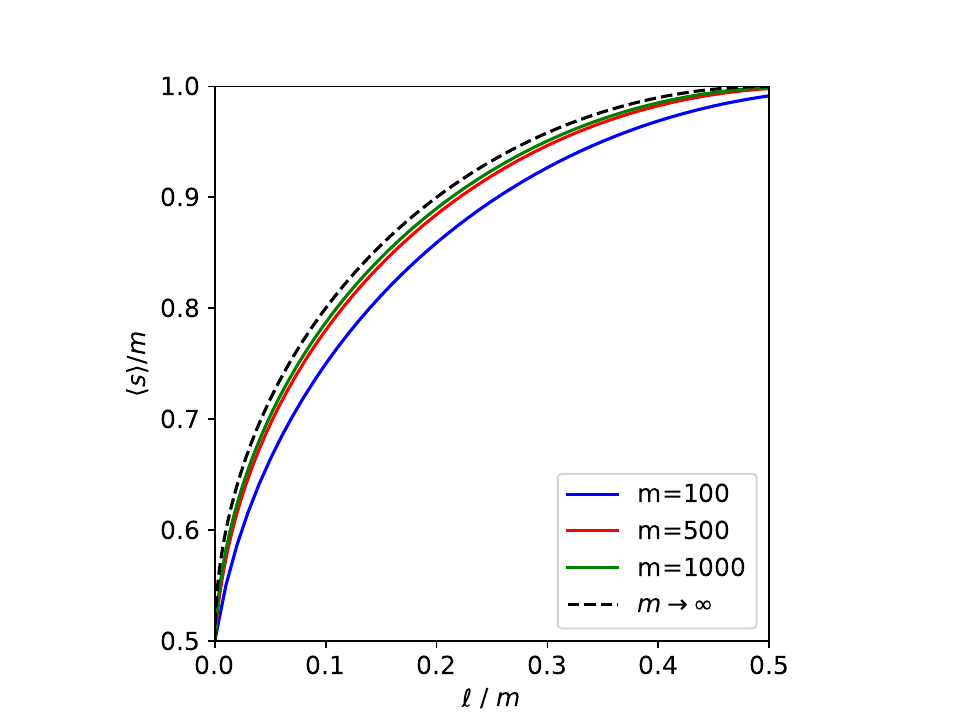}
    \caption{\label{fig:semicircle} Here we plot the expected fraction $\langle s\rangle/m$ of satisfied constraints, as dictated by Lemma \ref{thm:genw}, upon measuring $\ket{P(f)}$ when $P$ is the optimal degree-$\ell$ polynomial. We show the balanced case where $|f_i^{-1}(+1)| \simeq p/2$ for all $i$. Accordingly, the dashed black line corresponds to the asymptotic formula \eq{eq:semicircle_general} with $r/p = 1/2$.}
    \end{center}
\end{figure}

\subsection{Preliminaries}

Recall from \sect{sec:genp_dqi_state} that the quantum state $|\widetilde{P}(f)\rangle=F^{\otimes n}|P(f)\rangle$ obtained by applying the Quantum Fourier Transform to a DQI state $|P(f)\rangle = \sum_{\mathbf{x}\in\mathbb{F}_p^n}P(f(\mathbf{x}))|\mathbf{x}\rangle$
is
\begin{equation}
\ket{\widetilde{P}(f)} = \sum_{k=0}^\ell \frac{w_k}{\sqrt{\binom{m}{k}}} \sum_{\substack{\mathbf{y} \in \mathbb{F}_p^m\\
|\mathbf{y}|=k}} \left( \prod_{\substack{i=1\\y_i \neq 0}}^m \tilde{g}_i(y_i) \right) \ket{B^T \mathbf{y}}
\end{equation}
This can be written more succinctly as
\begin{align}
    |\widetilde{P}(f)\rangle&=\sum_{k=0}^\ell\frac{w_k}{\sqrt{{m \choose k}}}\sum_{\substack{\mathbf{y}\in\mathbb{F}_p^m\\|\mathbf{y}|=k}}\tilde{g}(\mathbf{y}) \ket{B^T \mathbf{y}} \label{eq:tilde_p_f}
\end{align}
by defining for any $\mathbf{y}\in\mathbb{F}_p^m$
\begin{align}\label{eq:tilde_h_for_codewords}
    \tilde{g}(\mathbf{y})=\prod_{\substack{i=1\\y_i\ne 0}}^m\tilde{g}_i(y_i)
\end{align}
where we regard the product of zero factors as $1$, so that $\tilde{g}_i(\mathbf{0})=1$. Next, we note that
\begin{align}\label{eq:sum_of_squared_magnitudes_of_dqi_amplitudes_is_binomial_coefficient}
    \sum_{\substack{\mathbf{y}\in\mathbb{F}_p^m\\|\mathbf{y}|=k}}|\tilde{g}(\mathbf{y})|^2 &= \sum_{\substack{i_1,\dots,i_k \in \{1,\dots,m\} \\ \text{distinct}}} \left(\sum_{y_1 \in \mathbb{F}_p} |\tilde{g}_{i_1}(y_1)|^2\right) \dots \left(\sum_{y_k \in \mathbb{F}_p} |\tilde{g}_{i_k}(y_k)|^2\right) = {m \choose k}
\end{align}
where we used the fact that $\tilde{g}_i(0)=0$ for all $i$. Lastly, we formally restate the following fact, which we derived in \sect{sec:genp_dqi_state}.
\begin{lemma}\label{lem:dqi_state_normalization_condition}
Let $d^\perp$ denote the minimum distance of the code $C^\perp = \{ \mathbf{d} \in \mathbb{F}_p^m : B^T \mathbf{d} = \mathbf{0} \}$. If $2 \ell < d^\perp$, then $\langle\widetilde{P}(f)|\widetilde{P}(f)\rangle=\|\mathbf{w}\|^2$.
\end{lemma}

The proof of Theorem \ref{thm:semicircle} consists of two parts. In the next subsection, we express the expected fraction $\langle s^{(m,\ell)}\rangle/m$ of satisfied constraints as a quadratic form involving a certain tridiagonal matrix. Then, we derive a closed-form asymptotic formula for the maximum eigenvalue of the matrix.   We end this section by combining the two parts into a proof of Theorem \ref{thm:semicircle}.  We note that \cite{BN06} encountered a similar matrix in the context of proving bounds on the size of codes with good distance, and derived a similar expression for the maximum eigenvalue, using similar methods.

\subsection{Expected Number of Satisfied Constraints}

\begin{lemma}
    \label{thm:genw}
    Let $f(\mathbf{x})=\sum_{i=1}^m f_i\left(\sum_{j=1}^n B_{ij} x_j \right)$ be a max-LINSAT objective function with matrix $B\in\mathbb{F}_p^{m\times n}$ for a prime $p$ and positive integers $m$ and $n$ such that $m>n$. Suppose that $|f_i^{-1}(+1)| = r$ for some $r\in\{1,\dots,p-1\}$. Let $P$ be a degree-$\ell$ polynomial normalized so that $\langle P(f)|P(f)\rangle=1$ with $P(f) = \sum_{k=0}^\ell w_k P^{(k)}(f_1,\ldots,f_m)/\sqrt{p^{n-k} {m \choose k}}$ its decomposition as a linear combination of elementary symmetric polynomials. Let $\langle s^{(m,\ell)} \rangle$ be the expected number of satisfied constraints for the symbol string obtained upon measuring the DQI state $\ket{P(f)}$ in the computational basis. If $2 \ell + 1 < d^\perp$ where $d^\perp$ is the minimum distance of the code $C^\perp = \{ \mathbf{d} \in \mathbb{F}_p^m : B^T \mathbf{d} = \mathbf{0} \}$, then
    \begin{equation}\label{eq:s_quadratic}
        \langle s^{(m,\ell)} \rangle = \frac{mr}{p} + \frac{\sqrt{r(p-r)}}{p} \mathbf{w}^\dag A^{(m,\ell,d)} \mathbf{w}
    \end{equation}
    where $\mathbf{w} = (w_0,\ldots,w_\ell)^T$ and $A^{(m,\ell,d)}$ is the $(\ell+1)\times(\ell+1)$ symmetric tridiagonal matrix
    \begin{equation}
    \label{eq:A_def}
        A^{(m,\ell,d)}=\begin{bmatrix}
        0   & a_1 & \\
        a_1 & d   & a_2   & \\
            & a_2 & 2d    & \ddots  & \\
            &     & \ddots &        & a_\ell \\
            &     &        & a_\ell & \ell d
        \end{bmatrix}
    \end{equation}
    with $a_k=\sqrt{k(m-k+1)}$ and $d = \frac{p-2r}{\sqrt{r(p-r)}}$.
\end{lemma}

\begin{proof}
The number of constraints satisfied by $\mathbf{x}\in\mathbb{F}_p^n$ is
\begin{align}\label{eq:def_s}
    s(\mathbf{x}) = \sum_{i=1}^m \mathbb{1}_{F_i}(\mathbf{b}_i\cdot\mathbf{x})
\end{align}
where
\begin{align}
    \mathbb{1}_A(x)=\begin{cases}
    1 & \text{if} \quad x \in A\\
    0 & \text{otherwise}
    \end{cases}
\end{align}
is the \textit{indicator function} of the set $A$. For any $v,x\in\mathbb{F}_p$, we can write
\begin{align}
    \mathbb{1}_{\{v\}}(x) = \frac{1}{p} \sum_{a\in\mathbb{F}_p} \omega_p^{a(x-v)}
\end{align}
so
\begin{align}
    \mathbb{1}_{F_i}(x) = \sum_{v\in F_i} \mathbb{1}_{\{v\}}(x) = \frac{1}{p}\sum_{v\in F_i} \sum_{a\in\mathbb{F}_p} \omega_p^{a(x-v)}.
\end{align}
Substituting into \eqref{eq:def_s}, we have
\begin{align}\label{eq:s_sum}
    s(\mathbf{x}) = \frac{1}{p} \sum_{i=1}^m \sum_{v\in F_i} \sum_{a\in\mathbb{F}_p} \omega_p^{a(\mathbf{b}_i\cdot\mathbf{x}-v)}.
\end{align}
The expected number of constraints satisfied by a symbol string sampled from the output distribution of a DQI state $|P(f)\rangle$ is
\begin{align}
    \langle s^{(m,\ell)}\rangle = \langle P(f)|S_f|P(f)\rangle
\end{align}
where $S_f=\sum_{\mathbf{x}\in\mathbb{F}_p^n}s(\mathbf{x})|\mathbf{x}\rangle\langle\mathbf{x}|$. We can express the observable $S_f$ in terms of the \textit{clock operator} $Z=\sum_{b\in\mathbb{F}_p}\omega_p^b|b\rangle\langle b|$ on $\mathbb{C}^p$ as
\begin{align}
    S_f &= \sum_{\mathbf{x}\in\mathbb{F}_p^n} s(\mathbf{x}) |\mathbf{x}\rangle \langle\mathbf{x}| \\
    &= \frac{1}{p} \sum_{i=1}^m \sum_{v\in F_i} \sum_{a\in\mathbb{F}_p} \sum_{\mathbf{x}\in\mathbb{F}_p^n} \omega_p^{a(\mathbf{b}_i\cdot\mathbf{x}-v)} |\mathbf{x}\rangle \langle\mathbf{x}| \\
    &= \frac{1}{p} \sum_{i=1}^m \sum_{v\in F_i} \sum_{a\in\mathbb{F}_p} \omega_p^{-av} \prod_{j=1}^n \sum_{x_j\in\mathbb{F}_p} \omega_p^{aB_{ij}x_j} |x_j\rangle \langle x_j| \\
    &= \frac{1}{p} \sum_{i=1}^m \sum_{v\in F_i} \sum_{a\in\mathbb{F}_p} \omega_p^{-av} \prod_{j=1}^n Z^{aB_{ij}}.
\end{align}
Next, we apply the Fourier transform $F_{ij}=\omega_p^{ij}/\sqrt{p}$ with $i,j=0,\dots,p-1$ to write $\langle s^{(m,\ell)}\rangle$ in terms of the \textit{shift operator} $X=\sum_{b\in\mathbb{F}_p}|b+1\rangle\langle b|$ on $\mathbb{C}^p$ as
\begin{align}
    \langle s^{(m,\ell)}\rangle&=\frac{1}{p}\sum_{i=1}^m\sum_{v\in F_i}\sum_{a\in\mathbb{F}_p}\omega_p^{-av}\langle P(f)|\prod_{j=1}^nZ_j^{aB_{ij}}|P(f)\rangle\\
    &=\frac{1}{p}\sum_{i=1}^m\sum_{v\in F_i}\sum_{a\in\mathbb{F}_p}\omega_p^{-av}\langle \widetilde{P}(f)|\prod_{j=1}^nX_j^{-aB_{ij}}|\widetilde{P}(f)\rangle.\label{eq:expected_satisfied_count_as_pauli_x}
\end{align}
where we used $FZF^\dagger=X^{-1}$ and $|\widetilde{P}(f)\rangle=F^{\otimes n}|P(f)\rangle$.

Substituting \eqref{eq:tilde_p_f} into \eqref{eq:expected_satisfied_count_as_pauli_x}, we get
\begin{align}
    \langle s^{(m,\ell)}\rangle&=\frac{1}{p}\sum_{k_1,k_2=0}^\ell\frac{w_{k_1}^* w^{\phantom{*}}_{k_2}}{\sqrt{{m \choose k_1}{m \choose k_2}}}\sum_{\substack{\mathbf{y}_1,\mathbf{y}_2\in\mathbb{F}_p^m\\|\mathbf{y}_1|=k_1\\|\mathbf{y}_2|=k_2}} \tilde{g}^*(\mathbf{y}_1) \tilde{g}(\mathbf{y}_2)\sum_{i=1}^m\sum_{v\in F_i}\sum_{a\in\mathbb{F}_p}\omega_p^{-av}\langle B^T\mathbf{y}_1|\prod_{j=1}^nX_j^{-aB_{ij}}|B^T\mathbf{y}_2\rangle
\end{align}

Let $\mathbf{e}_1,\dots,\mathbf{e}_m\in\mathbb{F}_p^m$ denote the standard basis of one-hot vectors. Then
\begin{align}
    \prod_{j=1}^nX_j^{-aB_{ij}}|B^T\mathbf{y}_2\rangle = |B^T\mathbf{y}_2 -a\mathbf{b}_i^T\rangle = |B^T(\mathbf{y}_2-a\mathbf{e}_i)\rangle
\end{align}
so
\begin{align}
    \langle s^{(m,\ell)}\rangle&=\frac{1}{p}\sum_{k_1,k_2=0}^\ell\frac{w^*_{k_1} w^{\phantom{*}}_{k_2}}{\sqrt{{m \choose k_1}{m \choose k_2}}}\sum_{\substack{\mathbf{y}_1,\mathbf{y}_2\in\mathbb{F}_p^m\\|\mathbf{y}_1|=k_1\\|\mathbf{y}_2|=k_2}} \tilde{g}^*(\mathbf{y}_1) \tilde{g}(\mathbf{y}_2)\sum_{i=1}^m\sum_{v\in F_i}\sum_{a\in\mathbb{F}_p}\omega_p^{-av}\langle B^T\mathbf{y}_1|B^T(\mathbf{y}_2-a\mathbf{e}_i)\rangle.\label{eq:expected_satisfied_fraction_in_terms_of_offset}
\end{align}
But $|B^T\mathbf{y}_1\rangle$ and $|B^T(\mathbf{y}_2-a\mathbf{e}_i)\rangle$ are computational basis states, so
\begin{align}
    \langle B^T\mathbf{y}_1|B^T(\mathbf{y}_2-a\mathbf{e}_i)\rangle = \begin{cases}
        1 & \text{if} \quad B^T\mathbf{y}_1 = B^T(\mathbf{y}_2 - a\mathbf{e}_i)\\
        0 & \text{otherwise}.
    \end{cases}
\end{align}
Moreover,
\begin{align}
    B^T\mathbf{y}_1 = B^T(\mathbf{y}_2 - a\mathbf{e}_i) \iff \mathbf{y}_1 - \mathbf{y}_2 + a\mathbf{e}_i\in C^\perp \iff \mathbf{y}_1 = \mathbf{y}_2 - a\mathbf{e}_i
\end{align}
where we used the assumption that the smallest Hamming weight of a non-zero symbol string in $C^\perp$ is $d^\perp > 2\ell +1 \geq k_1+k_2+1$.

There are three possibilities to consider: $|\mathbf{y}_1|=|\mathbf{y}_2|-1$, $|\mathbf{y}_2|=|\mathbf{y}_1|-1$, and $|\mathbf{y}_1|=|\mathbf{y}_2|$. We further break up the last case into $\mathbf{y}_1\ne\mathbf{y}_2$ and $\mathbf{y}_1=\mathbf{y}_2$. Before simplifying \eqref{eq:expected_satisfied_fraction_in_terms_of_offset}, we examine the values of $i\in\{1,\dots,m\}$ and $a\in\mathbb{F}_p$ for which $\langle B^T\mathbf{y}_1|B^T(\mathbf{y}_2-a\mathbf{e}_i)\rangle=1$ in each of the four cases. We also compute the value of the product $\tilde{g}^*(\mathbf{y}_1) \tilde{g}(\mathbf{y}_2)$.

Consider first the case $|\mathbf{y}_1|=|\mathbf{y}_2|-1$. Here, $a \ne 0$ and $i\in\{1,\dots,m\}$ is the position which is zero in $\mathbf{y}:=\mathbf{y}_1$ and $a$ in $\mathbf{y}_2$. Therefore, by definition \eqref{eq:tilde_h_for_codewords}
\begin{align}
    \tilde{g}^*(\mathbf{y}_1) \tilde{g}(\mathbf{y}_2)=|\tilde{g}(\mathbf{y})|^2 \tilde{g}_i(a).
\end{align}
Next, suppose $|\mathbf{y}_2|=|\mathbf{y}_1|-1$. Then $a \ne 0$ and $i\in\{1,\dots,m\}$ is the position which is zero in $\mathbf{y}:=\mathbf{y}_2$ and $-a$ in $\mathbf{y}_1$. Thus,
\begin{align}
    \tilde{g}^*(\mathbf{y}_1) \tilde{g}(\mathbf{y}_2)=|\tilde{g}(\mathbf{y}_2)|^2 \tilde{g}_i^*(-a) = |\tilde{g}(\mathbf{y})|^2 \tilde{g}_i(a).
\end{align}
Consider next the case $|\mathbf{y}_1|=|\mathbf{y}_2|$ and $\mathbf{y}_1\ne\mathbf{y}_2$. Here, $a \ne 0$ and $i\in\{1,\dots,m\}$ is the position which is $z-a$ in $\mathbf{y}_1$ and $z$ in $\mathbf{y}_2$ for some $z\in\mathbb{F}_p\setminus\{0,a\}$. Let $\mathbf{y}\in\mathbb{F}_p^m$ denote the vector which is zero at position $i$ and agrees with $\mathbf{y}_1$, and hence with $\mathbf{y}_2$, on all other positions. Then
\begin{align}
    \tilde{g}^*(\mathbf{y}_1) \tilde{g}(\mathbf{y}_2)=|\tilde{g}(\mathbf{y})|^2\tilde{g}_i^*(z-a)\tilde{g}_i(z)=|\tilde{g}(\mathbf{y})|^2\tilde{g}_i(a-z)\tilde{g}_i(z).
\end{align}
Finally, when $\mathbf{y}:=\mathbf{y}_1=\mathbf{y}_2$, we have
\begin{align}
    \tilde{g}^*(\mathbf{y}_1) \tilde{g}(\mathbf{y}_2)=|\tilde{g}(\mathbf{y})|^2.
\end{align}
In this case, $a=0$ and $i\in\{1,\dots,m\}$.

Putting it all together we can rewrite \eqref{eq:expected_satisfied_fraction_in_terms_of_offset} as
\begin{align}
    \langle s^{(m,\ell)}\rangle&=\frac{1}{p}\sum_{k=0}^{\ell-1}\frac{w_k^* w^{\phantom{*}}_{k+1}}{\sqrt{{m \choose k}{m \choose k+1}}}\sum_{\substack{\mathbf{y}\in\mathbb{F}_p^m\\|\mathbf{y}|=k}}|\tilde{g}(\mathbf{y})|^2\sum_{\substack{i=1\\y_i=0}}^m\sum_{v\in F_i}\sum_{a\in\mathbb{F}_p^*}\omega_p^{-av}\tilde{g}_i(a)\label{eq:s_tilde_h_a_super}\\
    &+\frac{1}{p}\sum_{k=0}^{\ell-1}\frac{w_{k+1}^* w^{\phantom{*}}_k}{\sqrt{{m \choose k+1}{m \choose k}}}\sum_{\substack{\mathbf{y}\in\mathbb{F}_p^m\\|\mathbf{y}|=k}}|\tilde{g}(\mathbf{y})|^2\sum_{\substack{i=1\\y_i=0}}^m\sum_{v\in F_i}\sum_{a\in\mathbb{F}_p^*}\omega_p^{-av}\tilde{g}_i(a)\label{eq:s_tilde_h_a_sub}\\
    &+\frac{1}{p}\sum_{k=1}^{\ell}\frac{|w_{k}|^2}{{m \choose k}}\sum_{\substack{\mathbf{y}\in\mathbb{F}_p^m\\|\mathbf{y}|=k-1}}|\tilde{g}(\mathbf{y})|^2\sum_{\substack{i=1\\y_i=0}}^m\sum_{v\in F_i}\sum_{a\in\mathbb{F}_p^*}\sum_{z\in\mathbb{F}_p\setminus\{0,a\}}\omega_p^{-av}\tilde{g}_i(a-z)\tilde{g}_i(z)\label{eq:s_tilde_h_a_diag}\\
    &+\frac{1}{p}\sum_{k=0}^{\ell}\frac{|w_{k}|^2}{{m \choose k}}\sum_{\substack{\mathbf{y}\in\mathbb{F}_p^m\\|\mathbf{y}|=k}}|\tilde{g}(\mathbf{y})|^2\sum_{i=1}^m\sum_{v\in F_i}\sum_{a\in\{0\}}\omega_p^{-av}.
\end{align}
Remembering that $\tilde{g}_i(0)=0$, we recognize the innermost sums in \eqref{eq:s_tilde_h_a_super} and \eqref{eq:s_tilde_h_a_sub} as the inverse Fourier transform so, for $\mathbf{y}$ with $|\mathbf{y}|=k$, we have
\begin{align}
    \sum_{\substack{i=1\\y_i=0}}^m\sum_{v\in F_i}\sum_{a\in\mathbb{F}_p^*}\omega_p^{-av}\tilde{g}_i(a)=\sqrt{p}\sum_{\substack{i=1\\y_i=0}}^m\sum_{v\in F_i}g_i(v)
    =(m-k)\sqrt{r(p-r)}
\end{align}
where we used \eq{eq:fbar} and \eq{eq:varphi} to calculate
\begin{align}\label{eq:gi_on_Fi}
    g_i(v) = \frac{1-\bar{f}}{\varphi} = \sqrt{\frac{p - r}{pr}}
\end{align}
for $v \in F_i$. Similarly, with $\tilde{g}_i(0)=0$, we recognize the sums indexed by $a$ and $z$ in \eqref{eq:s_tilde_h_a_diag} as the inverse Fourier transform of the convolution of $\tilde{g}_i$ with itself, so
\begin{align}
    \sum_{a\in\mathbb{F}_p}\omega_p^{-av}\sum_{z\in\mathbb{F}_p}\tilde{g}_i(a-z)\tilde{g}_i(z) &= \frac{1}{p}\sum_{a\in\mathbb{F}_p}\omega_p^{-av}\sum_{z\in\mathbb{F}_p}\sum_{x\in\mathbb{F}_p}\omega_p^{x(a-z)}g_i(x)\sum_{y\in\mathbb{F}_p}\omega_p^{yz}g_i(y) \\
    &=\sum_{a,x,y\in\mathbb{F}_p}\omega_p^{a(x-v)}g_i(x)g_i(y)\frac{1}{p}\sum_{z\in\mathbb{F}_p}\omega_p^{(y-x)z} \\
    &=\sum_{x\in\mathbb{F}_p}g_i(x)^2\sum_{a\in\mathbb{F}_p}\omega_p^{a(x-v)}=p \ g_i(v)^2
\end{align}
and, for $\mathbf{y}$ with $|\mathbf{y}|=k-1$, we have
\begin{align}
    \sum_{\substack{i=1\\y_i=0}}^m \sum_{v\in F_i} \sum_{a\in\mathbb{F}_p^*} \sum_{z\in\mathbb{F}_p\setminus\{0,a\}} \omega_p^{-av}\tilde{g}_i(a-z)\tilde{g}_i(z) &= \sum_{\substack{i=1\\y_i=0}}^m \sum_{v\in F_i} \left(\sum_{a,z\in\mathbb{F}_p}\omega_p^{-av}\tilde{g}_i(a-z)\tilde{g}_i(z) - \sum_{z\in\mathbb{F}_p}|\tilde{g}_i(z)|^2\right) \nonumber \\
    &=\sum_{\substack{i=1\\y_i=0}}^m\sum_{v\in F_i}\left(p \ g_i(v)^2-1\right) \\
    &=\sum_{\substack{i=1\\y_i=0}}^m\sum_{v\in F_i}\left(\frac{p - r}{r} -1\right) \\
    &=(m-k+1)(p-2r)
\end{align}
where we used \eqref{eq:gi_on_Fi}. Substituting back into \eqref{eq:s_tilde_h_a_super}, \eqref{eq:s_tilde_h_a_sub}, and \eqref{eq:s_tilde_h_a_diag}, we obtain
\begin{align}
    \langle s^{(m,\ell)}\rangle&=\frac{\sqrt{r(p-r)}}{p} \sum_{k=0}^{\ell-1}w^*_k w^{\phantom{*}}_{k+1}\frac{m-k}{\sqrt{{m \choose k}{m \choose k+1}}}\sum_{\substack{\mathbf{y}\in\mathbb{F}_p^m\\|y|=k}}|\tilde{g}(\mathbf{y})|^2 \\
    &+\frac{\sqrt{r(p-r)}}{p} \sum_{k=0}^{\ell-1}w^*_{k+1} w^{\phantom{*}}_k \frac{m-k}{\sqrt{{m \choose k+1}{m \choose k}}}\sum_{\substack{\mathbf{y}\in\mathbb{F}_p^m\\|y|=k}}|\tilde{g}(\mathbf{y})|^2 \\
    &+\frac{p-2r}{p}\sum_{k=1}^{\ell}|w_{k}|^2\frac{m-k+1}{{m \choose k}}\sum_{\substack{\mathbf{y}\in\mathbb{F}_p^m\\|y|=k-1}}|\tilde{g}(\mathbf{y})|^2 \\
    &+\frac{mr}{p}\sum_{k=0}^{\ell}|w_{k}|^2\frac{1}{{m \choose k}}\sum_{\substack{\mathbf{y}\in\mathbb{F}_p^m\\|y|=k}}|\tilde{g}(\mathbf{y})|^2
\end{align}
and using equation \eqref{eq:sum_of_squared_magnitudes_of_dqi_amplitudes_is_binomial_coefficient}, we get
\begin{align}
    \langle s^{(m,\ell)}\rangle&=\frac{\sqrt{r(p-r)}}{p} \sum_{k=0}^{\ell-1} w_k^* w_{k+1}\sqrt{(k+1)(m-k)} \\
    &+\frac{\sqrt{r(p-r)}}{p} \sum_{k=0}^{\ell-1}w_{k+1}^* w_k\sqrt{(k+1)(m-k)} \\
    &+\frac{p-2r}{p}\sum_{k=0}^\ell |w_k|^2 k \,+\, \frac{mr}{p}.
\end{align}
Defining
\begin{align}\label{eq:def_A_matrix}
    A^{(m,\ell,d)}=\begin{bmatrix}
    0   & a_1 & \\
    a_1 & d   & a_2   & \\
        & a_2 & 2d    & \ddots  & \\
        &      & \ddots &        & a_\ell \\
        &      &        & a_\ell & \ell d
    \end{bmatrix}
\end{align}
where $a_k=\sqrt{k(m-k+1)}$ for $k=1,\dots,\ell$ and $d\in\mathbb{R}$, we can write
\begin{equation}
    \langle s^{(m,\ell)}\rangle = \frac{mr}{p} + \frac{\sqrt{r(p-r)}}{p} \mathbf{w}^\dag A^{(m,\ell,d)} \mathbf{w}
\end{equation}
where $\mathbf{w} = (w_0,\ldots,w_\ell)^T$ and $d=\frac{p-2r}{\sqrt{r(p-r)}}$.
\end{proof}

\subsection{Asymptotic Formula for Maximum Eigenvalue of Matrix \texorpdfstring{$A^{(m,\ell,d)}/m$}{A}}

\begin{lemma}\label{lem:asymptotic_formula_for_max_eigenvalue}
Let $\lambda_\text{max}^{(m,\ell,d)}$ denote the maximum eigenvalue of the symmetric tridiagonal matrix $A^{(m,\ell,d)}$ defined in \eqref{eq:def_A_matrix}. If $\ell\leq m/2$ and $d\geq -\frac{m-2\ell}{\sqrt{\ell(m-\ell)}}$, then
\begin{align}
    \lim_{\substack{m,\ell\to\infty\\\ell/m=\mu}}\frac{\lambda_\text{max}^{(m,\ell,d)}}{m}=\mu d + 2\sqrt{\mu(1-\mu)}
\end{align}
where the limit is taken as both $m$ and $\ell$ tend to infinity with the ratio $\mu=\ell/m$ fixed.
\end{lemma}

\begin{proof}
First, we show that if $\ell\leq m/2$, then
\begin{align}
    \lim_{\substack{m,\ell\to\infty\\\ell/m=\mu}} \frac{\lambda_\text{max}^{(m,\ell,d)}}{m}\geq \mu d + 2\sqrt{\mu(1-\mu)}.
\end{align}
Define vector $v^{(m,\ell)}\in\mathbb{R}^{\ell+1}$ as
\begin{align}
    v_i^{(m,\ell)}=\begin{cases}
    0&\text{for}\quad i=0,1,\dots,\ell-\lceil\sqrt{\ell}\rceil\\
    \lceil\sqrt{\ell}\rceil^{-1/2}&\text{for}\quad i=\ell-\lceil\sqrt{\ell}\rceil+1,\dots,\ell.
    \end{cases}
\end{align}
Then, $\|v^{(m,\ell)}\|_2^2=\frac{\lceil\sqrt{\ell}\rceil}{\lceil\sqrt{\ell}\rceil}=1$. But
\begin{align}
    \lim_{\substack{m,\ell\to\infty\\\ell/m=\mu}} \frac{\lambda_\text{max}^{(m,\ell,d)}}{m} &\geq \lim_{\substack{m,\ell\to\infty\\\ell/m=\mu}}\frac{\left(v^{(m,\ell)}\right)^T A^{(m,\ell,d)} v^{(m,\ell)}}{m} \\
    \label{eq:two_term_lower_bound_on_max_eigenvalue_of_A}
    &= \lim_{\substack{m,\ell\to\infty\\\ell/m=\mu}} \left(\frac{1}{m\lceil\sqrt{\ell}\rceil}\sum_{k=\ell-\lceil\sqrt{\ell}\rceil+1}^\ell kd\right) + \lim_{\substack{m,\ell\to\infty\\\ell/m=\mu}} \left(\frac{1}{m\lceil\sqrt{\ell}\rceil}\sum_{k=\ell-\lceil\sqrt{\ell}\rceil+2}^\ell 2a_k\right).
\end{align}
The second term in \eqref{eq:two_term_lower_bound_on_max_eigenvalue_of_A} is bounded below by
\begin{align}
    \lim_{\substack{m,\ell\to\infty\\\ell/m=\mu}} \left(\frac{1}{m\lceil\sqrt{\ell}\rceil}\sum_{k=\ell-\lceil\sqrt{\ell}\rceil+2}^\ell 2a_k\right) &\geq\lim_{\substack{m,\ell\to\infty\\\ell/m=\mu}} \left(\frac{2a_{\ell-\lceil\sqrt{\ell}\rceil+2}(\lceil\sqrt{\ell}\rceil-1)}{m\lceil\sqrt{\ell}\rceil}\right) \\
    &= \lim_{\substack{m,\ell\to\infty\\\ell/m=\mu}} \left(\frac{2\sqrt{(\ell-\lceil\sqrt{\ell}\rceil+2)(m-\ell+\lceil\sqrt{\ell}\rceil-1)}}{m} \cdot \frac{\lceil\sqrt{\ell}\rceil-1}{\lceil\sqrt{\ell}\rceil}\right) \nonumber \\
    &= 2\sqrt{\mu(1-\mu)}
\end{align}
where we used the fact that $a_k$ increases as a function of $k$ for $k\leq m/2$. If $d\geq 0$, then the first term in \eqref{eq:two_term_lower_bound_on_max_eigenvalue_of_A} is bounded below by
\begin{align}
    \lim_{\substack{m,\ell\to\infty\\\ell/m=\mu}} \left(\frac{1}{m\lceil\sqrt{\ell}\rceil}\sum_{k=\ell-\lceil\sqrt{\ell}\rceil+1}^\ell kd\right)
    &\geq \lim_{\substack{m,\ell\to\infty\\\ell/m=\mu}} \quad \frac{\lceil\sqrt{\ell}\rceil\cdot(\ell-\lceil\sqrt{\ell}\rceil+1)\cdot d}{m\lceil\sqrt{\ell}\rceil} = \mu d
\end{align}
and if $d<0$, then it is bounded below by
\begin{align}
    \lim_{\substack{m,\ell\to\infty\\\ell/m=\mu}} \left(\frac{1}{m\lceil\sqrt{\ell}\rceil}\sum_{k=\ell-\lceil\sqrt{\ell}\rceil+1}^\ell kd\right)
    &\geq \lim_{\substack{m,\ell\to\infty\\\ell/m=\mu}} \quad \frac{\lceil\sqrt{\ell}\rceil\cdot\ell\cdot d}{m\lceil\sqrt{\ell}\rceil} = \mu d.
\end{align}
Putting it all together, we get
\begin{align}
    \lim_{\substack{m,\ell\to\infty\\\ell/m=\mu}} \frac{\lambda_\text{max}^{(m,\ell,d)}}{m}\geq \mu d + 2\sqrt{\mu(1-\mu)}.
\end{align}

Next, we establish a matching upper bound on $\lambda_\text{max}^{(m,\ell,d)}$. For $k = 0, 1, \ldots, \ell$, let $R_k^{(m, \ell,d)}$ denote the sum of the off-diagonal entries in the $k\th$ row of $A^{(m, \ell,d)}$ and set $a_0 = a_{\ell+1} = 0$, so that $R_k = a_k + a_{k+1}$. By Gershgorin's circle theorem, for every eigenvalue $\lambda^{(m, \ell,d)}$ of $A^{(m, \ell,d)}$, we have
\begin{align}
    \lambda^{(m,\ell,d)} &\leq \max_{k\in\{0,\dots,\ell\}} (k d + R_k) \\
    &= \max_{k\in\{0,\dots,\ell\}} \left(k d + \sqrt{k(m-k+1)} + \sqrt{(k+1)(m-k)}\right). \label{eq:int}
\end{align}
By assumption, $\ell < m/2$, so $\sqrt{(k+1)(m-k)} \geq \sqrt{k(m-k+1)}$ for all $k$ in the sum. Thus \eq{eq:int} yields
\begin{align}
\lambda^{(m,\ell,d)}
    &\leq \max_{k\in\{0,\dots,\ell\}} \left(k d + 2\sqrt{(k+1)(m-k)}\right)  \\
    &\leq \max_{k\in\{0,\dots,\ell\}} \left(k d + 2\sqrt{k(m-k)}+2\sqrt{m-k}\right) \\
    &\leq 2\sqrt{m} + \max_{k\in\{0,\dots,\ell\}} \left(k d + 2\sqrt{k(m-k)}\right) \\
    &= 2\sqrt{m} + \max_{k\in\{0,\dots,\ell\}} \xi(k)
\end{align}
where we define $\xi(x)=xd + 2\sqrt{x(m-x)}$. Note that $\xi^{''}(x)=-\frac{m^2}{2[x(m-x)]^{3/2}}<0$ for all $x\in(0,m)$, so the derivative $\xi^{'}(x)=d+\frac{m-2x}{\sqrt{x(m-x)}}$ is decreasing in this interval. However, by assumption $d\geq-\frac{m-2\ell}{\sqrt{\ell(m-\ell)}}$, so $\xi^{'}(x)\geq\xi^{'}(\ell)\geq 0$ for $x\in(0,\ell]$, since $\ell < m/2$. Therefore, $\xi(x)$ is increasing on this interval and we have
\begin{align}
    \lambda^{(m,\ell,d)} \leq 2\sqrt{m} + \ell d + 2\sqrt{\ell(m-\ell)}.
\end{align}
But then
\begin{align}
\lim_{\substack{m,\ell\to\infty\\\ell/m=\mu}}\frac{\lambda_\text{max}^{(m,\ell,d)}}{m} &\leq \lim_{\substack{m, \ell \to \infty \\ \ell / m = \mu}} \frac{2\sqrt{m} + \ell d + 2\sqrt{\ell(m-\ell)}}{m}=\mu d + 2\sqrt{\mu (1 - \mu)}
\end{align}
which establishes the matching upper bound and completes the proof of the lemma.
\end{proof}

\subsection{Optimal Asymptotic Expected Fraction of Satisfied Constraints}

In this subsection we use Lemmas \ref{thm:genw} and \ref{lem:asymptotic_formula_for_max_eigenvalue} to prove Theorem \ref{thm:semicircle}.

\begin{proof}
Recall from Lemma \ref{lem:dqi_state_normalization_condition} that $\|\mathbf{w}\|_2 = 1$. Therefore, the expected number of satisfied constraints is maximized by choosing $\mathbf{w}$ in \eqref{eq:s_quadratic} to be the normalized eigenvector of $A^{(m,\ell,d)}$ corresponding to its maximal eigenvalue. This leads to
\begin{equation}\label{eq:satisfied_fraction_bounded_by_eigenvalue}
    \frac{\langle s^{(m,\ell)}\rangle_{\mathrm{opt}}}{m} = \rho + \sqrt{\rho(1-\rho)}\frac{\lambda_{\max}^{(m,\ell,d)}}{m}
\end{equation}
where $\rho=\frac{r}{p}$ and $d=\frac{p-2r}{\sqrt{r(p-r)}}$. Consider first the case of $\frac{r}{p}\leq 1-\frac{\ell}{m}$. Then $\frac{p-r}{r}\geq\frac{\ell}{m-\ell}$ and
\begin{align}
    d = \frac{p-2r}{\sqrt{r(p-r)}} = \sqrt{\frac{p-r}{r}} - \sqrt{\frac{r}{p-r}} \geq \sqrt{\frac{\ell}{m-\ell}} - \sqrt{\frac{m-\ell}{\ell}} = -\frac{m-2\ell}{\sqrt{\ell(m-\ell)}}.
\end{align}
Moreover, $\ell < (d^\perp-1)/2 \leq (m-1)/2$. Therefore, Lemma \ref{lem:asymptotic_formula_for_max_eigenvalue} applies and we have
\begin{align}
    \lim_{\substack{m,\ell \to \infty \\ \ell/m = \mu}} \frac{\langle s^{(m,\ell)}\rangle_{\mathrm{opt}}}{m} &= \rho + \sqrt{\rho(1-\rho)} \lim_{\substack{m,\ell \to \infty \\ \ell/m = \mu}} \frac{\lambda_{\max}^{(m,\ell,d)}}{m}. \\
    &= \rho + \sqrt{\rho(1-\rho)} \left(\mu d+2\sqrt{\mu(1-\mu)}\right) \\
\end{align}
Recalling $d = \frac{1-2 \rho}{\sqrt{\rho(1-\rho)}}$ this yields
\begin{align}
    \lim_{\substack{m,\ell \to \infty \\ \ell/m = \mu}} \frac{\langle s^{(m,\ell)}\rangle_{\mathrm{opt}}}{m}
    &= \mu+\rho-2\mu\rho+2\sqrt{\mu(1-\mu)\rho(1-\rho)} \\
    &= \left(\sqrt{\mu(1-\rho)} + \sqrt{\rho(1-\mu)}\right)^2
\end{align}
for $\rho \leq 1-\mu$. In particular, if $\rho=1-\mu$, then $\left(\sqrt{\mu(1-\rho)} + \sqrt{\rho(1-\mu)}\right)^2=1$. But $\langle s^{(m,\ell)}\rangle_{\mathrm{opt}}$ is an increasing function of $r$ that cannot exceed $1$. Consequently,
\begin{align}
    \lim_{\substack{m,\ell \to \infty \\ \ell/m = \mu}} \frac{\langle s^{(m,\ell)}\rangle_{\mathrm{opt}}}{m} = 1
\end{align}
for $\rho > 1-\mu$. Putting it all together, we obtain
\begin{align}
    \lim_{\substack{m,\ell \to \infty \\ \ell/m = \mu}} \frac{\langle s^{(m,\ell)}\rangle_{\mathrm{opt}}}{m} = \begin{cases}
        \left(\sqrt{\mu(1-\rho)} + \sqrt{\rho(1-\mu)}\right)^2 & \text{if}\,\,\rho\leq 1-\mu\\
        1 & \text{otherwise}
    \end{cases}
\end{align}
which completes the proof of Theorem \ref{thm:semicircle}.
\end{proof}

\section{Removing the Minimum Distance Assumption}
\label{sec:beyond_2l}

So far, in our description and analysis of DQI we have always assumed that $2 \ell + 1 < d^\perp$. This condition buys us several advantages. First, it ensures that the states $\ket{\widetilde{P}^{(0)}},\ldots,\ket{\widetilde{P}^{(\ell)}}$, from which we construct
\begin{equation}
\label{eq:ketptilde}
\ket{\widetilde{P}(f)} = \sum_{k=0}^\ell w_k \frac{1}{\sqrt{\binom{m}{k}}} \sum_{\substack{\mathbf{y} \in \mathbb{F}_2^m \\ |\mathbf{y}| = k}} (-1)^{\mathbf{v} \cdot \mathbf{y}} \ket{B^T \mathbf{y}},
\end{equation}
form an orthonormal set which implies that
\begin{equation}
    \label{eq:normdenom}
    \langle \widetilde{P}(f) | \widetilde{P}(f) \rangle = \mathbf{w}^\dag \mathbf{w}.
\end{equation}
Second, it allows us to obtain an \emph{exact} expression for the expected number of constraints satisfied without needing to know the weight distributions of either the code $C$ or $C^\perp$, namely for $p=2$
\begin{equation}
    \label{eq:quadnum}
    \langle s \rangle = \frac{1}{2} m + \frac{1}{2} \mathbf{w}^\dag A^{(m,\ell, 0)} \mathbf{w}
\end{equation}
where $A^{(m,\ell, d)}$ is the $(\ell+1) \times (\ell + 1)$ matrix defined in \eq{eq:A_def}. These facts in turn allow us to prove Theorem \ref{thm:semicircle}.

For the irregular LDPC code described in \sect{sec:wins}, we can estimate $d^\perp$ using the Gilbert-Varshamov bound\footnote{By \cite{Gal62} it is known that the distance of random LDPC codes drawn from various standard ensembles is well-approximated asymptotically by the Gilbert-Varshamov bound. We extrapolate that this is also a good approximation for our ensemble of random LDPC codes with irregular degree distribution.}. By experimentally testing belief propagation, we find that for some codes it is able to correct slightly more than $(d^\perp-1)/2$ errors with high reliability. Under this circumstance, equations (\ref{eq:normdenom}) and (\ref{eq:quadnum}) no longer hold exactly. Here, we prove that they continue to hold in expectation for max-XORSAT with uniformly average $\mathbf{v}$, up to small corrections due to decoding failures. The precise statement of our result is given in Theorem \ref{thm:imperfect_decoding}.

In the remainder of this section we first describe precisely what we mean by the DQI algorithm in the case of $2 \ell + 1 \geq d^\perp$. (For simplicity, we consider the case $p=2$.)

As the initial step of DQI, we perform classical preprocessing to choose $\mathbf{w} \in \mathbb{R}^{\ell+1}$, which is equivalent to making a choice of degree-$\ell$ polynomial $P$. In the case $2 \ell + 1 < d^\perp$ we can exactly compute the choice of $\mathbf{w}$ that maximizes $\langle s \rangle /m$. Specifically, it is the principal eigenvector of $A^{(m,\ell, d)}$ defined in \eq{eq:A_def}. Once we reach or exceed $2 \ell + 1 = d^\perp$ the principal eigenvector of this matrix is not necessarily the optimal choice. But we can still use it as our choice of $\mathbf{w}$, and as we will show below, it remains a good choice.

After choosing $\mathbf{w}$, the next step in the DQI algorithm, as discussed in \sect{sec:genl_dqi_algo} is to prepare the state
\begin{equation}
    \label{eq:dqi_state_def}
    \sum_{k=0}^\ell w_k \frac{1}{\sqrt{{m \choose k}}}\sum_{\substack{\mathbf{y}\in\mathbb{F}_2^m\\|\mathbf{y}|=k}} (-1)^{\mathbf{v} \cdot \mathbf{y}} \ket{\mathbf{y}} \ket{B^T \mathbf{y}}.
\end{equation}
The following step is to uncompute $\ket{\mathbf{y}}$. When $2 \ell + 1 \geq d^\perp$ this uncomputation will not succeed with $100\%$ certainty and one cannot produce exactly the state 
$\ket{\widetilde{P}(f)} = \sum_{k=0}^\ell w_k \ket{\widetilde{P}^{(k)}}$. Instead, the goal is to produce a good approximation to $\ket{\widetilde{P}(f)}$ with high probability. This is because the number of errors $\ell$ is large enough so that, by starting from a codeword in $C^\perp$ and then flipping $\ell$ bits, the nearest codeword (in Hamming distance) to the resulting string may be a codeword other than the starting codeword. (This is the same reason that $\ket{\widetilde{P}^{(0)}},\ldots,\ket{\widetilde{P}^{(\ell)}}$ are no longer orthogonal and hence the norm of $\ket{\widetilde{P}(f)}$ is no longer exactly equal to $\|\mathbf{w}\|$.)

If the decoder succeeds on a large fraction of the errors that are in superposition, one simply postselects on success of the decoder and obtains a normalized state which is a good approximation to the (unnormalized) ideal state $\ket{\widetilde{P}(f)}$. The last steps of DQI are to perform a Hadamard transform and then measure in the computational basis, just as in the case of $2 \ell + 1 < d^\perp$.

\subsection{General Expressions for the Expected Number of Satisfied Constraints}

Here we derive generalizations of lemmas \ref{lem:dqi_state_normalization_condition} and \ref{thm:genw} for arbitrary code distances. We restrict our attention to the max-XORSAT case where we are given $B \in \mathbb{F}_2^{m \times n}$ and $\mathbf{v} \in \mathbb{F}_2^m$ and we seek to maximize the objective function $f(\mathbf{x}) = \sum_{i = 1}^m (-1)^{\mathbf{b}_i \cdot \mathbf{x} + v_i}$, where $\mathbf{b}_i$ is the $i\th$ row of $B$ and $v_i$ is the $i\th$ entry of $\mathbf{v}$.

By examining \eq{eq:ketptilde} one sees that if $d^\perp > 2\ell$ then the $B^T \mathbf{y}$ are all distinct and the norm of $|\widetilde{P}(f)\rangle$ is the norm of $\mathbf{w}$, as shown in Lemma \ref{lem:dqi_state_normalization_condition}. More generally, we have the following lemma.  
\begin{lemma}
    \label{lem:Mlemma} The squared norm of $|\widetilde{P}(f)\rangle$ is 
    \begin{equation}
    \langle \widetilde{P}(f) | \widetilde{P}(f) \rangle = \mathbf{w}^\dag M^{(m,\ell)} \mathbf{w},
    \end{equation}
    where $M^{(m,\ell)}$ is the $(\ell + 1) \times (\ell + 1)$ symmetric matrix defined by
    \begin{equation}
        M_{k,k'}^{(m,\ell)} = \frac{1}{\sqrt{\binom{m}{k} \binom{m}{k'}}} \sum_{\substack{\mathbf{y} \in \mathbb{F}_2^m \\ |\mathbf{y}| = k}} \sum_{\substack{\mathbf{y}' \in \mathbb{F}_2^m \\ |\mathbf{y}'| = k'}} (-1)^{(\mathbf{y} + \mathbf{y}') \cdot \mathbf{v}} \delta_{B^T \mathbf{y}, B^T \mathbf{y}'}, \label{eq:Mdef}
    \end{equation}
    for $0 \leq k,k' \leq \ell$. 
\end{lemma}
\begin{proof}
This is immediate from \eqref{eq:ketptilde} and the fact that the $\ket{B^T \mathbf{y}}$ are computational basis states. 
\end{proof}
Note that if $d^\perp > 2\ell$, then $M^{(m,\ell)}_{k,k'} = \delta_{k,k'}$, in agreement with Lemma \ref{lem:dqi_state_normalization_condition}. 

\begin{lemma}
\label{lem:Alemma_beyond_2l}
Let $\ket{P(f)}$ be the Hadamard transform of the state $\ket{\widetilde{P}(f)}$ defined in \eq{eq:ketptilde}. Let $\langle f \rangle$ be the expected objective value for the symbol string obtained upon measuring the DQI state $\ket{P(f)}$ in the computational basis. If the weights $w_k$ are such that $|P(f)\rangle$ is normalized, then
\begin{equation}
    \langle f \rangle = \mathbf{w}^\dag \bar{A}^{(m,\ell)} \mathbf{w}
\end{equation}
where $\bar{A}^{(m,\ell)}$ is the $(\ell+1)\times(\ell+1)$ symmetric matrix defined by
\begin{equation}
    \label{eq:barAdef}
    \bar{A}^{(m,\ell)}_{k,k'}  =  \frac{1}{\sqrt{\binom{m}{k} \binom{m}{k'}}} \sum_{i=1}^m \sum_{(\mathbf{y},\mathbf{y'}) \in S^{(i)}_{k,k'}} (-1)^{v_i + \mathbf{v} \cdot (\mathbf{y} + \mathbf{y}')}
\end{equation}
for $0 \leq k,k' \leq \ell$, and 
\begin{equation}
    \label{eq:def_Si}
    S^{(i)}_{k,k'} =  \{(\mathbf{y},\mathbf{y}') \in \mathbb{F}_2^m \times \mathbb{F}_2^m : |\mathbf{y}| = k, |\mathbf{y}'| = k', B^T(\mathbf{y} + \mathbf{y}' + \mathbf{e}_i) = \mathbf{0} \}.
\end{equation}
for one-hot vectors $\mathbf{e}_1,\ldots,\mathbf{e}_m$ in $\mathbb{F}_2^m$.
\end{lemma}
\begin{proof}
    The expected value of $f(\mathbf{x})$ in state $\ket{P(f)}$ is
\begin{equation}
\langle f \rangle = \bra{P(f)} H_f \ket{P(f)},
\end{equation}
where
\begin{equation}
    H_f = \sum_{i=1}^m (-1)^{v_i} \prod_{j:B_{ij}=1} Z_j
\end{equation}
and $Z_j$ denotes the Pauli $Z$ operator acting on the $j\th$ qubit. Recalling that conjugation by Hadamard interchanges Pauli $X$ with $Z$, we have
\begin{eqnarray}
    \langle  f \rangle & = & \sum_{i=1}^m (-1)^{v_i} \bra{P(f)} \prod_{j: B_{ij} = 1} Z_j \ket{P(f)} \\
    & = & \sum_{i=1}^m (-1)^{v_i} \bra{\widetilde{P}(f)} \prod_{j : B_{ij} = 1} X_j \ket{\widetilde{P}(f)} \\ & = & \sum_{k,k'=0}^\ell \frac{w_k w_{k'}}{\sqrt{\binom{m}{k} \binom{m}{k'}}} \sum_{|\mathbf{y}|=k} \sum_{|\mathbf{y}'|=k'} (-1)^{\mathbf{v} \cdot (\mathbf{y} + \mathbf{y}')} \sum_{i=1}^m (-1)^{v_i} \bra{B^T \mathbf{y}} \prod_{j:B_{ij}=1} X_j \ket{B^T \mathbf{y}'}, \label{eq:hopping1}
\end{eqnarray}
where in the last line we have plugged in the definition of $\ket{\widetilde{P}(f)}$. 
We can rewrite this quadratic form as
\begin{equation}
    \langle f
    \rangle = \sum_{k,k'=0}^\ell w_k \ w_{k'} \ \bar{A}^{(m,\ell)}_{k,k'},
\end{equation}
where
\begin{equation}
    \bar{A}^{(m,\ell)}_{k,k'} = \frac{1}{\sqrt{\binom{m}{k} \binom{m}{k'}}} \sum_{|\mathbf{y}|=k} \sum_{|\mathbf{y}'|=k'} (-1)^{\mathbf{v} \cdot (\mathbf{y} + \mathbf{y}')} \sum_{i=1}^m (-1)^{v_i} \bra{B^T \mathbf{y}} \prod_{j:B_{ij}=1} X_j \ket{B^T \mathbf{y}'}. \label{eq:hopping2}
\end{equation}
For the one-hot vectors $\mathbf{e}_1,\ldots,\mathbf{e}_m \in \mathbb{F}_2^m$ we have
\begin{equation}
    \prod_{j:B_{ij}=1} X_j \ket{B^T\mathbf{y}'} = \ket{B^T(\mathbf{y}'+\mathbf{e}_i)}.
\end{equation}
Substituting this into \eq{eq:hopping2} yields
\begin{equation}
   \bar{A}^{(m,\ell)}_{k,k'} = \frac{1}{\sqrt{\binom{m}{k} \binom{m}{k'}}} \sum_{|\mathbf{y}|=k} \sum_{|\mathbf{y}'|=k'} (-1)^{\mathbf{v} \cdot (\mathbf{y} + \mathbf{y}')} \sum_{i=1}^m (-1)^{v_i} \langle B^T \mathbf{y} | B^T (\mathbf{y}'+ \mathbf{e}_i) \rangle. \label{eq:hopping3}
\end{equation}
We next note that $\langle B^T \mathbf{y} | B^T (\mathbf{y}'+ \mathbf{e}_i) \rangle$ equals one when $B^T \mathbf{y} = B^T (\mathbf{y}'+\mathbf{e}_i)$ and zero otherwise. The condition $B^T \mathbf{y} = B^T (\mathbf{y}'+\mathbf{e}_i)$ is equivalent to $B^T (\mathbf{y} + \mathbf{y}'+\mathbf{e}_i) = \mathbf{0}$. Hence,
\begin{equation}
    \bar{A}^{(m,\ell)}_{k,k'} = \frac{1}{\sqrt{\binom{m}{k} \binom{m}{k'}}} \sum_{i=1}^m \sum_{(\mathbf{y},\mathbf{y'}) \in S^{(i)}_{k,k'}} (-1)^{v_i + \mathbf{v} \cdot (\mathbf{y} + \mathbf{y}')}
\end{equation}
where
\begin{equation}
    S^{(i)}_{k,k'} = \{(\mathbf{y},\mathbf{y}') \in \mathbb{F}_2^m \times \mathbb{F}_2^m : |\mathbf{y}| = k, |\mathbf{y}'| = k', B^T (\mathbf{y} + \mathbf{y}' + \mathbf{e}_i) = \mathbf{0} \}.
\end{equation}
\end{proof}
Note once again that if $d^\perp > 2\ell + 1$ then $\bar{A}^{(m,\ell)}$ simplifies to $A^{(m,\ell,0)}$ defined in \eqref{eq:A_def}. It is also easy to verify that $\bar{A}^{(m,\ell)}$ can be rewritten in terms of $M$ as \begin{equation}
    \bar{A}^{(m,\ell)}_{k,k'} = \sqrt{k' (m-k'+1)}M^{(m,\ell)}_{k,k'-1} + \sqrt{(k'+1) (m-k')}M^{(m,\ell)}_{k,k'+1},
\end{equation}
where we define $M^{(m,\ell)}_{k,k'}$ according to formula \eqref{eq:Mdef}, even if $k$ or $k'$ exceed $\ell+1$.

\subsection{Average-case \texorpdfstring{$\mathbf{v}$}{v}}

Any choice of $\mathbf{w} \in \mathbb{C}^{\ell+1}$ defines a corresponding state $\ket{P(f)} = H^{\otimes n} \ket{\widetilde{P}(f)}$ via \eq{eq:ketptilde}. (In general this will not be normalized.)

As shown in the previous subsection the general expression for the expected value of $f$ achieved by measuring the ideal normalized DQI state $\ket{P(f)}/\| \ket{P(f)} \|$ in the computational basis, which holds even when $2 \ell + 1 \geq d^\perp$, is
\begin{equation}
\langle f \rangle = \frac{\mathbf{w}^\dag \bar{A}^{(m,\ell)} \mathbf{w}}{\mathbf{w}^\dag M^{(m,\ell)} \mathbf{w}}.
\end{equation}

In the case $2 \ell + 1 \geq d^\perp$ it is impossible to precisely obtain the ideal DQI state. This is because the number of errors $\ell$ is large enough so that, by starting from a codeword in $C^\perp$ and then flipping $\ell$ bits, the nearest codeword (in Hamming distance) to the resulting string may be a codeword other than the starting codeword. Therefore, in this section, we analyze $\langle f \rangle$ in the presence of a nonzero rate of decoding failures. Specifically, we consider the case that $\mathbf{v}$ is chosen uniformly at random and calculate $\mathbb{E}_{\mathbf{v}} \langle f \rangle$, as this averaging simplifies the analysis substantially, mainly due to the following fact.

\begin{lemma}
\label{lem:expectation_of_barA}
Let $\bar{A}^{(m,\ell)}$ be as defined in \eq{eq:barAdef}. Suppose $\mathbf{v}$ is chosen uniformly at random from $\mathbb{F}_2^m$. Then $\mathbb{E}_\mathbf{v} \bar{A}^{(m,\ell)} = A^{(m,\ell,0)}$, where $A^{(m,\ell,0)}$ is as defined in \eqref{eq:A_def}.
\end{lemma}
\begin{proof}
Recall that
\begin{equation}
    \bar{A}^{(m,\ell)}_{k,k'}  =  \frac{1}{\sqrt{\binom{m}{k} \binom{m}{k'}}} \sum_{i=1}^m \sum_{(\mathbf{y},\mathbf{y'}) \in S^{(i)}_{k,k'}} (-1)^{v_i + \mathbf{v} \cdot (\mathbf{y} + \mathbf{y}')} \label{eq:abar_recall}
\end{equation}
for $0 \leq k,k' \leq \ell$, and 
\begin{equation}
    S^{(i)}_{k,k'} =  \{(\mathbf{y},\mathbf{y}') \in \mathbb{F}_2^m \times \mathbb{F}_2^m : |\mathbf{y}| = k, |\mathbf{y}'| = k', B^T(\mathbf{y} + \mathbf{y}' + \mathbf{e}_i) = \mathbf{0} \}.
\end{equation}
We can express $S^{(i)}_{k,k'}$ as the union of two disjoint pieces
    \begin{eqnarray}
        S_{k,k'}^{(i)} & = & S_{k,k'}^{(i,0)} \cup S_{k,k'}^{(i,1)} \\
        \label{eq:def_Si0}
        S_{k,k'}^{(i,0)} & = & \{ (\mathbf{y},\mathbf{y}') \in S_{k,k'}^{(i)} : \mathbf{y} + \mathbf{y}' + \mathbf{e}_i = \mathbf{0} \} \\
        S_{k,k'}^{(i,1)} & = & \{ (\mathbf{y},\mathbf{y}') \in S_{k,k'}^{(i)} : \mathbf{y} + \mathbf{y}' + \mathbf{e}_i \neq \mathbf{0} \}.
    \end{eqnarray}
We can then write $\bar{A}^{(m,\ell)}$ as the corresponding sum of two contributions
\begin{eqnarray}
    \label{eq:twoterms}
    \bar{A}^{(m,\ell)}_{k,k'} & = & \frac{1}{\sqrt{\binom{m}{k} \binom{m}{k'}}} \sum_{i=1}^m \sum_{(\mathbf{y},\mathbf{y'}) \in S^{(i,0)}_{k,k'}} (-1)^{\mathbf{v} \cdot (\mathbf{y} + \mathbf{y}' + \mathbf{e}_i)} \nonumber \\
    & + & \frac{1}{\sqrt{\binom{m}{k} \binom{m}{k'}}} \sum_{i=1}^m \sum_{(\mathbf{y},\mathbf{y'}) \in S^{(i,1)}_{k,k'}} (-1)^{\mathbf{v} \cdot (\mathbf{y} + \mathbf{y}' + \mathbf{e}_i)}.
\end{eqnarray}
We next average over $\mathbf{v}$. By the definition of $S_{k,k'}^{(i,0)}$, all terms in the sum have $\mathbf{y} + \mathbf{y}' + \mathbf{e}_i=\mathbf{0}$ and therefore the first term is independent of $\mathbf{v}$. This renders the averaging over $\mathbf{v}$ trivial for the first term. By the definition of $S_{k,k'}^{(i,1)}$, the second term contains exclusively contributions where $\mathbf{y} + \mathbf{y}' + \mathbf{e}_i \neq \mathbf{0}$. By the identity $\frac{1}{2^m} \sum_{\mathbf{v} \in \mathbb{F}_2^m} (-1)^{\mathbf{v} \cdot \mathbf{z}} = \delta_{\mathbf{z},\mathbf{0}}$ we see that the second term makes zero contribution to the average. Thus we obtain
\begin{equation}
    \label{eq:EvbarA}
    \mathbb{E}_{\mathbf{v}} \bar{A}^{(m,\ell)}_{k,k'} = \frac{1}{\sqrt{\binom{m}{k} \binom{m}{k'}}} \sum_{i=1}^m \sum_{(\mathbf{y},\mathbf{y'}) \in S^{(i,0)}_{k,k'}} (-1)^{\mathbf{v} \cdot (\mathbf{y} + \mathbf{y}' + \mathbf{e}_i)}.
\end{equation}

We next observe that $S^{(i,0)}_{k,k'}$ contains the terms where $\mathbf{y} + \mathbf{y}' + \mathbf{e}_i = \mathbf{0}$ which are exactly the terms that contribute at $2 \ell + 1 < d^\perp$, whereas $S^{(i,1)}_{k,k'}$ contains the terms where $\mathbf{y} + \mathbf{y}' + \mathbf{e}_i \in C^\perp\setminus\{\mathbf{0}\}$ which, prior to averaging, are the new contribution arising when $\ell$ exceeds this bound. Consequently, $\mathbb{E}_{\mathbf{v}} \bar{A}^{(m,\ell)}_{k,k'}$ is  exactly equal to $A_{k,k'}^{(m,\ell,0)}$, as defined in \eq{eq:A_def}. (One can also verify this by direct calculation.) That is, we have obtained
\begin{equation}
    \mathbb{E}_\mathbf{v} \bar{A}^{(m,\ell)}_{k,k'} = A^{(m,\ell,0)}_{k,k'},
\end{equation}
which completes the proof of the lemma.
\end{proof}

\subsection{Imperfect decoding}

A deterministic decoder partitions the set of errors $\mathbb{F}_2^m = \mathcal{D} \cup \mathcal{F}$ into the set $\mathcal{D}$ of errors $\mathbf{y}$ correctly identified by the decoder based on the syndrome $B^T\mathbf{y}$ and the set $\mathcal{F}$ of errors misidentified. The Hamming shell $\mathcal{E}_k$ of radius $k$ is analogously partitioned $\mathcal{E}_k = \mathcal{D}_k \cup \mathcal{F}_k$. We will quantify decoder's failure rate using $\varepsilon_k := |\mathcal{F}_k| / {m \choose k}$ and $\varepsilon := \max_{0 \leq k \leq \ell} \varepsilon_k$.

The quantum state of the error and syndrome registers after the error uncomputation step of the DQI algorithm using an imperfect decoder is
\begin{align}
    \sum_{k=0}^\ell \frac{w_k}{\sqrt{{m \choose k}}} \left( \sum_{\substack{\mathbf{y}\in\mathcal{D}_k \\ |\mathbf{y}|=k}} (-1)^{\mathbf{v}\cdot\mathbf{y}}|\mathbf{0}\rangle|B^T\mathbf{y}\rangle + \sum_{\substack{\mathbf{y}\in\mathcal{F}_k \\ |\mathbf{y}|=k}} (-1)^{\mathbf{v}\cdot\mathbf{y}}|\mathbf{y}\oplus\mathbf{y}'\rangle|B^T\mathbf{y}\rangle \right)
\end{align}
where $\mathbf{y} \ne \mathbf{y}'$. After uncomputing the error register, we postselect on the register being $|\mathbf{0}\rangle$. If the postselection is successful, then the syndrome register is in the quantum state proportional to the following unnormalized state vector
\begin{align}\label{eq:imperfect_dqi_state_def}
    |\widetilde{P}_\mathcal{D}(f)\rangle := \sum_{k=0}^\ell \frac{w_k}{\sqrt{{m \choose k}}} \sum_{\substack{\mathbf{y}\in\mathcal{D}_k \\ |\mathbf{y}|=k}} (-1)^{\mathbf{v}\cdot\mathbf{y}}|B^T\mathbf{y}\rangle.
\end{align}

The following theorem describes the effect that decoding failure rate $\varepsilon$ has on the approximation ratio achieved by DQI. We do not assume that $2\ell + 1 < d^\perp$.

\begin{theorem}
\label{thm:imperfect_decoding} 
 Given $B \in \mathbb{F}_2^{m \times n}$ and $\mathbf{v} \in \mathbb{F}_2^m$, let $f$ be the objective function $f(\mathbf{x}) = \sum_{i=1}^m (-1)^{v_i + \mathbf{b}_i \cdot \mathbf{x}}$. Let $P$ be any degree-$\ell$ polynomial and let $P(f)=\sum_{k=0}^{\ell} w_k P^{(k)}\left(f_1, \ldots, f_m\right) / \sqrt{2^n \binom{m}{k}}$ be the decomposition of $P(f)$ as a linear combination of elementary symmetric polynomials. Let $|P_\mathcal{D}(f)\rangle$ denote a DQI state prepared using an imperfect decoder that misidentifies $\varepsilon_k{m \choose k}$ errors of Hamming weight $k$ and let $\langle f \rangle$ be the expected objective value for the symbol string resulting from the measurement of this state in the computational basis. If $\mathbf{v} \in \mathbb{F}_2^m$ is chosen uniformly at random, then
\begin{equation}
    \label{eq:imperfect_avg_beyond}
    \mathbb{E}_{\mathbf{v}} \langle f\rangle \geq \frac{\mathbf{w}^\dagger \left[ A^{(m,\ell,0)} - 2 \varepsilon (m + 1) \right] \mathbf{w}}{\sum_{k = 0}^\ell w_k^2 (1-\varepsilon_k)} \geq \left( \frac{\mathbf{w}^\dagger A^{(m,\ell,0)} \mathbf{w}}{\mathbf{w}^\dagger \mathbf{w}} - 2 \varepsilon (m + 1) \right),
\end{equation}
where $\varepsilon = \max_{0 \le k \le \ell} \varepsilon_k$ and $A^{(m,\ell,0)}$ is the tridiagonal matrix defined in equation \eqref{eq:A_def}. Moreover, if $\ell \leq m/2$ and one chooses $\mathbf{w}$ to be the principal eigenvector of $A^{(m,\ell,0)}$ then \eq{eq:imperfect_avg_beyond} yields the following lower bound in the limit of large $\ell$ and $m$ with the ratio $\mu = \ell/m$ fixed:
For a random $\mathbf{v} \in \mathbb{F}_2^m$,
\begin{align}
    \label{eq:imperfect_asymptotic_beyond_avg}
    \lim_{\substack{m \to \infty \\ \ell/m = \mu}} \frac{\mathbb{E}_{\mathbf{v}}\langle f\rangle}{m}  & \geq 2 \left(\sqrt{ \frac{\ell}{m} \left( 1 - \frac{\ell}{m} \right)} - \varepsilon\right).
\end{align}
\end{theorem}

Before proving Theorem \ref{thm:imperfect_decoding}, we generalize lemma~\ref{lem:Mlemma} and lemma~\ref{lem:Alemma_beyond_2l} to the state $|P_\mathcal{D}(f)\rangle$ prepared by DQI with an imperfect decoder upon successful postselection.

\begin{lemma}
    \label{lem:Mlemma_for_imperfect_decoding}
    The squared norm of $|\widetilde{P}_\mathcal{D}(f)\rangle$ is 
    \begin{equation}
    \langle \widetilde{P}_\mathcal{D}(f) | \widetilde{P}_\mathcal{D}(f) \rangle = \sum_{k = 0}^\ell w_k^2 (1-\varepsilon_k) \leq \mathbf{w}^\dagger \mathbf{w}.
    \end{equation}
\end{lemma}
\begin{proof}
We first observe that all syndromes $\ket{B^T \mathbf{y}}$ for $\mathbf{y} \in \mathcal{D}$ are necessarily distinct bit strings and thus orthogonal quantum states. This follows from the fact that if a deterministic decoder correctly recovers $\mathbf{y}$ from the syndrome $B^T \mathbf{y}$, then it must fail on all $\mathbf{y}' \not = \mathbf{y}$ with $B^T \mathbf{y}' = B^T \mathbf{y}$. The squared norm of $|\widetilde{P}_\mathcal{D}(f)\rangle$ is thus \begin{equation}
     \langle \widetilde{P}_\mathcal{D}(f) | \widetilde{P}_\mathcal{D}(f) \rangle = \sum_{k = 0}^\ell \frac{w_k^2}{\binom{m}{k}} |\mathcal{D}_k| = \sum_{k = 0}^\ell w_k^2 (1-\varepsilon_k),
\end{equation} as claimed. 
\end{proof}

\begin{lemma}
\label{lem:Alemma_for_imperfect_decoding}
For $B \in \mathbb{F}_2^{m \times n}$ and $\mathbf{v} \in \mathbb{F}_2^m$, let $f$ be the objective function $f(\mathbf{x}) = \sum_{i=1}^m (-1)^{v_i + \mathbf{b}_i \cdot  \mathbf{x}}$. Let $P$ be any degree-$\ell$ polynomial and let $P(f)=\sum_{k=0}^{\ell} w_k P^{(k)}\left(f_1, \ldots, f_m\right) / \sqrt{2^n \binom{m}{k}}$ be the decomposition of $P(f)$ as a linear combination of elementary symmetric polynomials. Let $\langle f \rangle$ be the expected objective value for the symbol string obtained upon measuring the imperfect DQI state $|P_\mathcal{D}(f)\rangle$ in the computational basis. If the weights $w_k$ are such that $|P_\mathcal{D}(f)\rangle$ is normalized, then
\begin{equation}
    \langle f \rangle = \mathbf{w}^\dag \bar{A}^{(m,\ell,\mathcal{D})} \mathbf{w}
\end{equation}
where $\bar{A}^{(m,\ell,\mathcal{D})}$ is the $(\ell+1)\times(\ell+1)$ symmetric matrix defined by
\begin{equation}
    \label{eq:imperfect_barAdef}
    \bar{A}^{(m,\ell,\mathcal{D})}_{k,k'}  =  \frac{1}{\sqrt{\binom{m}{k} \binom{m}{k'}}} \sum_{i=1}^m \sum_{(\mathbf{y},\mathbf{y'}) \in S^{(i,\mathcal{D})}_{k,k'}} (-1)^{v_i + \mathbf{v} \cdot (\mathbf{y} + \mathbf{y}')}
\end{equation}
for $0 \leq k,k' \leq \ell$, and 
\begin{equation}
    S^{(i,\mathcal{D})}_{k,k'} =  \{(\mathbf{y},\mathbf{y}') \in \mathcal{D}_k \times \mathcal{D}_{k'} : B^T(\mathbf{y} + \mathbf{y}' + \mathbf{e}_i) = \mathbf{0} \}
\end{equation}
for one-hot vectors $\mathbf{e}_1,\ldots,\mathbf{e}_m$ in $\mathbb{F}_2^m$.
\end{lemma}
\begin{proof}
The proof is obtained from the proof of Lemma \ref{lem:Alemma_beyond_2l} by replacing each sum ranging over $\mathcal{E}_k$ with a sum ranging over $\mathcal{D}_k$ and the ideal DQI state $|P(f)\rangle$ with the imperfect DQI state $|P_\mathcal{D}(f)\rangle$.
\end{proof}

\subsection{Average-case \texorpdfstring{$\mathbf{v}$} with imperfect decoding}

The expected objective value achieved by sampling from a normalized DQI state $|P_\mathcal{D}(f)\rangle$ is
\begin{equation}
    \label{eq:expected_f_as_ratio}
    \langle f \rangle = \frac{\mathbf{w}^\dag \bar{A}^{(m,\ell,\mathcal{D})} \mathbf{w}}{\sum_{k = 0}^\ell w_k^2 (1-\varepsilon_k)}  \geq \frac{\mathbf{w}^\dag \bar{A}^{(m,\ell,\mathcal{D})} \mathbf{w}}{\mathbf{w}^\dagger \mathbf{w}}. 
\end{equation}

Next, we find the expectation of $\bar{A}^{(m,\ell,\mathcal{D})}$.
\begin{lemma}
\label{lem:expected_barA_for_imperfect_decoding}
Let $\bar{A}^{(m,\ell,\mathcal{D})}$ be defined as in \eq{eq:imperfect_barAdef}. Suppose $\mathbf{v}$ is chosen uniformly at random from $\mathbb{F}_2^m$. Then
\begin{equation}
    \mathbb{E}_\mathbf{v} \bar{A}^{(m,\ell,\mathcal{D})} = A^{(m,\ell,0)} -  E^{(m,\ell,\mathcal{F})} \label{eq:goodstuff}
\end{equation}
where $A^{(m,\ell,0)}$ is defined as in \eqref{eq:A_def} and
\begin{equation}
    \label{eq:imperfect_Edef}
    E^{(m,\ell,\mathcal{F})}_{k,k'}  =  \frac{1}{\sqrt{\binom{m}{k} \binom{m}{k'}}} \sum_{i=1}^m  |T^{(i,\mathcal{F})}_{k,k'}|
\end{equation}
for $0 \leq k,k' \leq \ell$ with
\begin{equation}
    T^{(i,\mathcal{F})}_{k,k'} =  \{(\mathbf{y},\mathbf{y}') \in \mathcal{E}_k \times \mathcal{F}_{k'} \cup \mathcal{F}_k \times \mathcal{E}_{k'}: \mathbf{y} + \mathbf{y}' + \mathbf{e}_i = \mathbf{0} \}.
\end{equation}
\end{lemma}
\begin{proof}
We can partition $S^{(i,\mathcal{D})}_{k,k'}$ into two disjoint subsets
    \begin{eqnarray}
        S_{k,k'}^{(i,\mathcal{D})} & = & S_{k,k'}^{(i,0,\mathcal{D})} \cup S_{k,k'}^{(i,1,\mathcal{D})} \\
        S_{k,k'}^{(i,0,\mathcal{D})} & = & \{ (\mathbf{y},\mathbf{y}') \in S_{k,k'}^{(i,\mathcal{D})} : \mathbf{y} + \mathbf{y}' + \mathbf{e}_i = \mathbf{0} \} \\
        S_{k,k'}^{(i,1,\mathcal{D})} & = & \{ (\mathbf{y},\mathbf{y}') \in S_{k,k'}^{(i,\mathcal{D})} : \mathbf{y} + \mathbf{y}' + \mathbf{e}_i \neq \mathbf{0} \}
    \end{eqnarray}
so that
\begin{eqnarray}
    \bar{A}^{(m,\ell,\mathcal{D})}_{k,k'} & = & \frac{1}{\sqrt{\binom{m}{k} \binom{m}{k'}}} \sum_{i=1}^m \sum_{(\mathbf{y},\mathbf{y'}) \in S^{(i,0,\mathcal{D})}_{k,k'}} (-1)^{v_i + \mathbf{v} \cdot (\mathbf{y} + \mathbf{y}')} \nonumber \\
    & + & \frac{1}{\sqrt{\binom{m}{k} \binom{m}{k'}}} \sum_{i=1}^m \sum_{(\mathbf{y},\mathbf{y'}) \in S^{(i,1,\mathcal{D})}_{k,k'}} (-1)^{v_i + \mathbf{v} \cdot (\mathbf{y} + \mathbf{y}')}.
    \label{eq:barA_is_sum_over_S0D_S1D}
\end{eqnarray}
We next observe that when we average over $\mathbf{v}$, the second term in \eq{eq:barA_is_sum_over_S0D_S1D} vanishes, because
\begin{eqnarray}
\mathbb{E}_{\mathbf{v}} \left( \sum_{(\mathbf{y},\mathbf{y}') \in S_{k,k'}^{(i,1,\mathcal{D})}} (-1)^{v_i + \mathbf{v} \cdot (\mathbf{y} + \mathbf{y}')} \right) & = & \sum_{(\mathbf{y},\mathbf{y}') \in S_{k,k'}^{(i,1,\mathcal{D})}} \mathbb{E}_{\mathbf{v}} \left( (-1)^{v_i + \mathbf{v} \cdot (\mathbf{y} + \mathbf{y}')} \right)\\
& = & \sum_{(\mathbf{y},\mathbf{y}') \in S_{k,k'}^{(i,1,\mathcal{D})}} \frac{1}{2^m} \sum_{\mathbf{v} \in \mathbb{F}_2^m} (-1)^{\mathbf{v} \cdot (\mathbf{y} + \mathbf{y}' + \mathbf{e}_i)} \\
& = & \sum_{(\mathbf{y},\mathbf{y}') \in S_{k,k'}^{(i,1,\mathcal{D})}} \delta_{\mathbf{y} + \mathbf{y}' + \mathbf{e}_i,\mathbf{0}} \\
& = & 0,
\end{eqnarray}
where the last equality follows from the definition of $S_{k,k'}^{(i,1,\mathcal{D})}$. Hence,
\begin{equation}
    \mathbb{E}_{\mathbf{v}} \bar{A}^{(m,\ell,\mathcal{D})}_{k,k'} = \mathbb{E}_{\mathbf{v}} \frac{1}{\sqrt{\binom{m}{k} \binom{m}{k'}}} \sum_{i=1}^m \sum_{(\mathbf{y},\mathbf{y'}) \in S^{(i,0,\mathcal{D})}_{k,k'}} (-1)^{v_i + \mathbf{v} \cdot (\mathbf{y} + \mathbf{y}')}.
\end{equation}
We next examine \eq{eq:EvbarA} and observe that $\mathbb{E}_{\mathbf{v}} \bar{A}^{(m,\ell,\mathcal{D})}_{k,k'}$ and $\mathbb{E}_{\mathbf{v}} \bar{A}^{(m,\ell)}_{k,k'}$ differ only by replacing the summation over $S^{(i,0,\mathcal{D})}_{k,k'}$ with a summation over $S^{(i,0)}_{k,k'}$.
Recalling the definition of $S^{(i,0)}_{k,k'}$ from equation \eqref{eq:def_Si0}, we observe that $S^{(i,0)}_{k,k'} = S^{(i,0,\mathcal{D})}_{k,k'} \cup T^{(i,\mathcal{F})}_{k,k'}$. Hence,
\begin{equation}
\mathbb{E}_{\mathbf{v}} \bar{A}^{(m,\ell)}_{k,k'} = \mathbb{E}_{\mathbf{v}} \bar{A}^{(m,\ell,\mathcal{D})}_{k,k'} + \mathbb{E}_{\mathbf{v}} \frac{1}{\sqrt{\binom{m}{k} \binom{m}{k'}}} \sum_{i=1}^m \sum_{(\mathbf{y},\mathbf{y}') \in T^{(i,\mathcal{F})}} (-1)^{\mathbf{v} \cdot (\mathbf{y} + \mathbf{y}'+\mathbf{e}_i)}.
\end{equation}
By the definition of $T^{(i,\mathcal{F})}$ one sees that all terms in the sum have $\mathbf{v} \cdot (\mathbf{y} + \mathbf{y}'+\mathbf{e}_i) = \mathbf{0}$ and hence
\begin{equation}
\mathbb{E}_{\mathbf{v}} \bar{A}^{(m,\ell)}_{k,k'} = \mathbb{E}_{\mathbf{v}} \bar{A}^{(m,\ell,\mathcal{D})}_{k,k'} + \frac{1}{\sqrt{\binom{m}{k} \binom{m}{k'}}} \sum_{i=1}^m |T^{(i,\mathcal{F})}|.
\end{equation}
By Lemma \ref{lem:expectation_of_barA}, this simplifies to
\begin{equation}
\bar{A}^{(m,\ell,0)}_{k,k'} = \mathbb{E}_{\mathbf{v}} A^{(m,\ell,\mathcal{D})}_{k,k'} + \frac{1}{\sqrt{\binom{m}{k} \binom{m}{k'}}} \sum_{i=1}^m |T^{(i,\mathcal{F})}|.
\end{equation}
By \eq{eq:imperfect_Edef} we can rewrite this as
\begin{equation}
\bar{A}^{(m,\ell,0)}_{k,k'} = \mathbb{E}_{\mathbf{v}} A^{(m,\ell,\mathcal{D})}_{k,k'} + E^{(m,\ell,\mathcal{F})}_{k,k'},
\end{equation}
which rearranges to \eq{eq:goodstuff} completing the proof.
\end{proof}

The last remaining ingredient for the proof of theorem \ref{thm:imperfect_decoding} is the following upper bound on the error term identified in Lemma \ref{lem:expected_barA_for_imperfect_decoding}.

\begin{lemma}
\label{lem:bound_on_norm_E}
Let $\varepsilon_k = |\mathcal{F}_k|/{m \choose k}$ and $\varepsilon = \max_{0 \leq k \leq \ell} \varepsilon_k$. Then,
\begin{equation}
    \|E^{(m,\ell,\mathcal{F})}\| \leq 2 \varepsilon (m + 1).
\end{equation}
\end{lemma}
\begin{proof}
We note that $E^{(m,\ell,\mathcal{F})}_{k,k'}=0$ unless $k=k'\pm 1$, so $E^{(m,\ell,\mathcal{F})}$ is tridiagonal with zeros on the diagonal. By \eq{eq:imperfect_Edef} we have
\begin{align}
    E^{(m,\ell,\mathcal{F})}_{k,k+1} = \frac{1}{\sqrt{\binom{m}{k} \binom{m}{k+1}}} \sum_{i=1}^m |T^{(i,\mathcal{F})}_{k,k+1}|.
\end{align}
Note that if $(\mathbf{y}, \mathbf{y}')\in T^{(i,\mathcal{F})}_{k,k+1}$, then $\mathbf{y}_i=0$ and $\mathbf{y}'_i=1$. Therefore, every pair $(\mathbf{y}, \mathbf{y}')$ with $\mathbf{y}\in\mathcal{F}_k$ contributes to at most $m-k$ out of the $m$ terms in the sum above. Similarly, every pair $(\mathbf{y}, \mathbf{y}')$ with $\mathbf{y}'\in\mathcal{F}_{k+1}$ contributes to at most $k+1$ terms. Consequently,
\begin{align}
    E^{(m,\ell,\mathcal{F})}_{k,k+1}
    & \leq \frac{1}{\sqrt{\binom{m}{k} \binom{m}{k+1}}} \big[ (m-k)|\mathcal{F}_k|+(k+1)|\mathcal{F}_{k+1}| \big] \\
    & = \frac{1}{\sqrt{\binom{m}{k} \binom{m}{k+1}}} \left[ (m-k) \varepsilon_k \binom{m}{k} +(k+1) \varepsilon_{k+1} \binom{m}{k+1} \right] \\
    & \leq (\varepsilon_k + \varepsilon_{k+1}) \sqrt{(k+1)(m-k)} \\
    & \leq \varepsilon (m + 1)
\end{align}
and by Gershgorin's circle theorem
\begin{equation}
    \|E^{(m,\ell,\mathcal{F})}\| \leq 2 \varepsilon (m + 1)
\end{equation}
completing the proof of the Lemma.
\end{proof}

\subsection{Proof of Theorem \ref{thm:imperfect_decoding}}

Finally, we prove Theorem \ref{thm:imperfect_decoding}.

\begin{proof}
Equation \eqref{eq:expected_f_as_ratio}  and Lemma \ref{lem:expected_barA_for_imperfect_decoding} imply that
\begin{equation}
    \mathbb{E}_\mathbf{v} \langle f\rangle \geq \frac{\mathbf{w}^{\dagger} \left[ A^{(m,\ell,0)} -  E^{(m,\ell,\mathcal{F})} \right] \mathbf{w}}{\sum_{k = 0}^\ell w_k^2 (1-\varepsilon_k)} \geq \left( \frac{\mathbf{w}^{\dagger} A^{(m,\ell,0)}\mathbf{w}}{\mathbf{w}^{\dagger} \mathbf{w}} - \frac{\mathbf{w}^{\dagger} E^{(m,\ell,\mathcal{F})}\mathbf{w}}{\mathbf{w}^{\dagger} \mathbf{w}} \right)
\end{equation}
which in light of Lemma \ref{lem:bound_on_norm_E} becomes
\begin{equation}
    \mathbb{E}_\mathbf{v} \langle f\rangle \geq \frac{\mathbf{w}^\dagger \left[ A^{(m,\ell,0)} - 2 \varepsilon (m + 1) \right] \mathbf{w}}{\sum_{k = 0}^\ell w_k^2 (1-\varepsilon_k)} \geq \left( \frac{\mathbf{w}^{\dagger} A^{(m,\ell,0)}\mathbf{w}}{\mathbf{w}^{\dagger} \mathbf{w}} - 2 \varepsilon (m + 1) \right).
\end{equation}

This proves the first part of Theorem \ref{thm:imperfect_decoding}. The second part, equation \eqref{eq:imperfect_asymptotic_beyond_avg}, follows by substituting the asymptotic formula for the leading eigenvalue of $A^{(m,\ell,0)}$ derived in Lemma \ref{lem:asymptotic_formula_for_max_eigenvalue}.
\end{proof}

\begin{remark} In practice, we find that decoding success probability decreases as the number of errors increases. Therefore, we can use the empirical failure probability of a classical decoder on uniformly random errors of Hamming weight $\ell$ as an upper bound on $\varepsilon$, which would be the failure probability of a decoder applied to a distribution over error weights zero to $\ell$ determined by the vector $\mathbf{w}$.
\end{remark}

\section{Other Optimization Algorithms}
\label{sec:classical}

In this section we survey algorithms against which DQI can be compared. In \sect{sec:local}, we consider local search heuristics such as simulated annealing and greedy optimization. In \sect{sec:sa_convergence} we comment in more detail on the convergence of simulated annealing when one adds more sweeps. In \sect{sec:trunc} we analyze Prange's algorithm in which one discards all but $n$ of the $m$ constraints on the $n$ variables and then solves the resulting linear system. In \sect{sec:advrand}, we summarize the AdvRand algorithm of \cite{BM15}. In \sect{sec:QAOA}, we discuss the Quantum Approximate Optimization Algorithm (QAOA). Though necessarily not exhaustive, we believe these algorithms constitute a thorough set of general-purpose optimization strategies to benchmark DQI against on random sparse max-XORSAT. Lastly, we analyze the two classes of algebraic algorithms that pose the most plausible challenge to DQI on our OPI problem and find that they are not successful in our parameter regime. Specifically, in \sect{sec:list-recovery} we consider list-recovery algorithms, and in \sect{sec:lattice}, we summarize the lattice-based heuristic of \cite{BN00}.

\subsection{General Local Search Heuristics}
\label{sec:local}

In local search methods, one makes a sequence of local moves in the search space, such as by flipping an individual bit, and preferentially accepts moves that improve the objective function. This class of heuristics includes simulated annealing, parallel tempering, TABU search, greedy algorithms, and some quantum-inspired optimization methods. For simplicity, we will restrict our analysis to max-XORSAT and consider only local search algorithms in which at each move a single variable among $x_1,\ldots,x_n$ is flipped between 1 and 0. The generalization to single-symbol-flip updates applied to max-LINSAT is straightforward.

Let $C_i$ be the set of constraints containing the variable $x_i$. In Gallager's ensemble $|C_i| = D$ for all $i$. For assignment $\mathbf{x} \in \mathbb{F}_2^n$ let $S_i(\mathbf{x})$ be the number of constraints in $C_i$ that are satisfied. Consider an assignment $\mathbf{x}$ such that a fraction $\phi$ of the $m$ constraints are satisfied. Then, modeling $C_1,\ldots,C_m$ as random subsets, we have
\begin{equation}
    \label{eq:pmodel}
    \mathrm{Pr} \left[ S_i(\mathbf{x}) = s \right] = \binom{D}{s} \phi^s (1-\phi)^{D-s}.
\end{equation}
Next, consider the move $\mathbf{x} \to \mathbf{x}'$ induced by flipping bit $i$. This causes all satisfied constraints in $C_i$ to become unsatisfied and vice-versa. Hence, such a move induces the change
\begin{equation}
    S_i(\mathbf{x}') = D - S_i(\mathbf{x}).
\end{equation}
This will be an improvement in the objective value if and only if $S_i(\mathbf{x}) < D/2$. According to \eq{eq:pmodel} the probability $P_i^{(+)}$ that the flip is an improvement is
\begin{equation}
    P_i^{(+)} = \sum_{s=0}^{\lfloor \frac{D-1}{2} \rfloor} \binom{D}{s} \phi^s (1-\phi)^{D-s}.
\end{equation}
By Hoeffding's inequality, we have that
\begin{equation}\label{eq:HoeffdingPPlus}
    P_i^{(+)} \leq \exp \left( - 2(\phi-1/2)^2 D \right).
\end{equation}

For large $\phi$ the probability $P_i^{(+)}$ becomes very small. When the probability of making an upward move becomes extremely small, the local optimization algorithm will no longer be able to achieve further improvement in any reasonable number of steps. If the algorithm takes a total of $N$ steps then, by the union bound and \eq{eq:HoeffdingPPlus}, the probability of finding such a move in any of the steps is upper bounded by
\begin{equation}
    P_{\mathrm{success}}^{\mathrm{bound}}(\phi) = N \exp \left( - 2(\phi -1/2)^2 D \right).
\end{equation}
From this we can see what is the highest value of $\phi$ for which the success probability remains significant. For example, if we set $P_{\mathrm{success}}^{\mathrm{bound}}(\phi_{\max}) = 1/2$ and solve for $\phi_{\max}$ the result is
\begin{equation}
    \label{eq:logfactor}
    \phi_{\max} = \frac{1}{2} + \sqrt{\frac{\log N + \log 2}{2D}}
\end{equation}

This analysis is only approximate. Indeed, one can see that the model cannot hold indefinitely as $N$ becomes larger, because eventually $\phi_{\max}$ becomes limited by the true optimum. Let us now compare it with computer experiment. In Fig. \ref{fig:SAfits} we show best fits to the satisfaction fraction versus $D$ with $k/D$ fixed at $1/10$, $1/2$, and $9/10$ for simulated annealing and greedy descent (which is equivalent to zero-temperature simulated annealing). We find that for each choice of $k/D$, $\phi_{\max}-1/2$ fits well to $c D^{-\nu}$ for some constant $c$ and some power $\nu$ but that the power $\nu$ is slightly smaller than $1/2$. We believe this to be a finite-size effect, as it has been shown \cite{MH22} that for an ensemble of random degree-$D$ max-$k$-XORSAT instances differing only slightly from the Gallager ensemble, the exact optimum scales like $\frac{1}{2}+\frac{P_k}{\sqrt{D}}$. Thus we find that the functional form
\begin{equation}
    \phi_{\max} = \frac{1}{2} + \frac{c}{\sqrt{D}}
\end{equation}
predicted by this argument is in reasonably good agreement with our experimental observations.

\begin{figure}
    \[
        \hspace{-15pt} \begin{array}{cc} \includegraphics[width=0.5\textwidth]{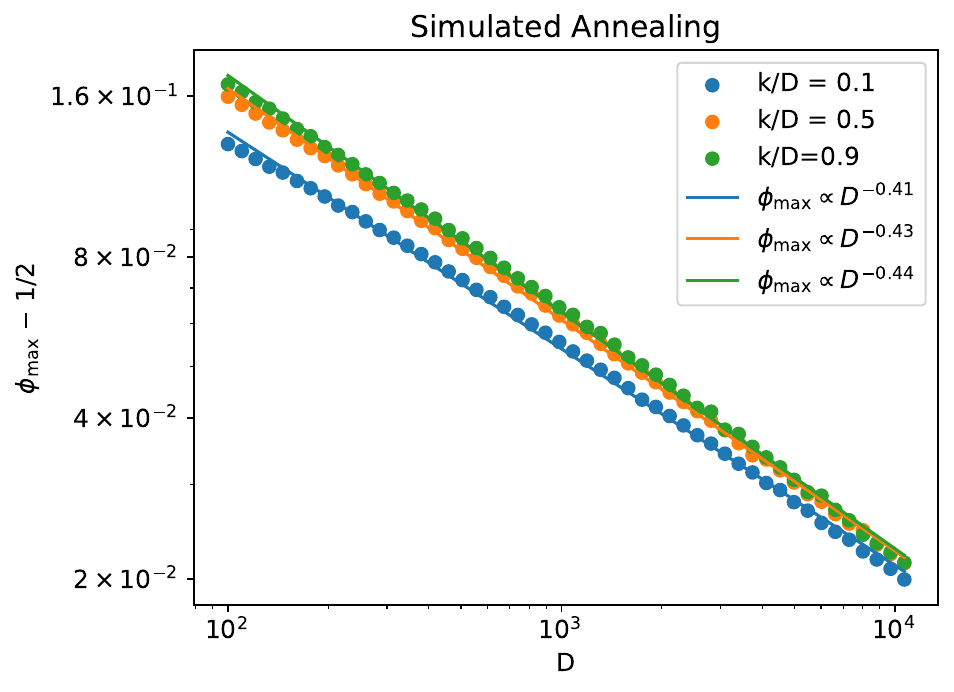} & \includegraphics[width=0.5\textwidth]{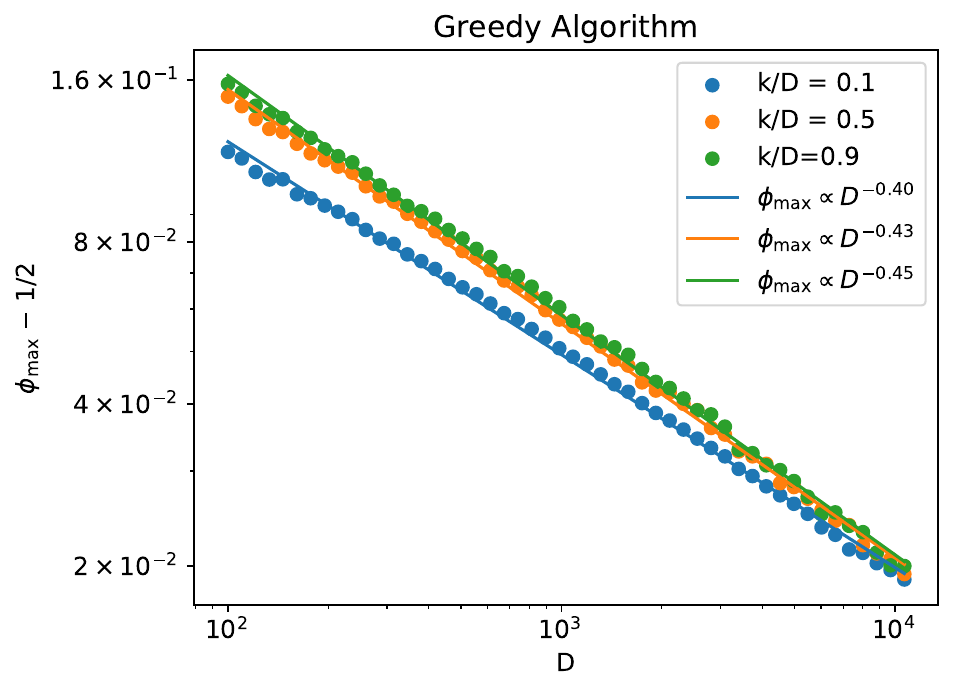} \end{array}
    \]
    \caption{The approximation achieved by simulated annealing (left) and greedy optimization (right) in our computer experiments at $n=10^5$, on a log-log scale, where each constraint contains $k$ variables and each variable is contained in $D$ constraints. The lines illustrates linear least-squares curve fits to the log-log data. The ensemble of max-$k$-XORSAT instances is formally defined in Appendix \ref{app:random_regular}. In these anneals we use $5,000$ sweeps with single-bit updates, and linearly increasing inverse temperature $\beta$.}
    \label{fig:SAfits}
\end{figure}

\subsection{Convergence of Simulated Annealing}
\label{sec:sa_convergence}

We next investigate whether the $N$-scaling predicted by the argument in \sect{sec:local} yields a good model of the behavior of simulated annealing. This scaling cannot persist indefinitely because eventually $\phi_{\max}$ is limited by the true optimum. In fact, we find that, in contrast to the $D$-scaling, the $N$-scaling suggested by the above argument is not corroborated by experimental evidence. Instead, empirical $N$-scaling fits much better to power-law convergence, as illustrated in Fig. \ref{fig:power_law}.

\begin{figure}
    \[
        \begin{array}{c}
            \includegraphics[width=0.8\textwidth]{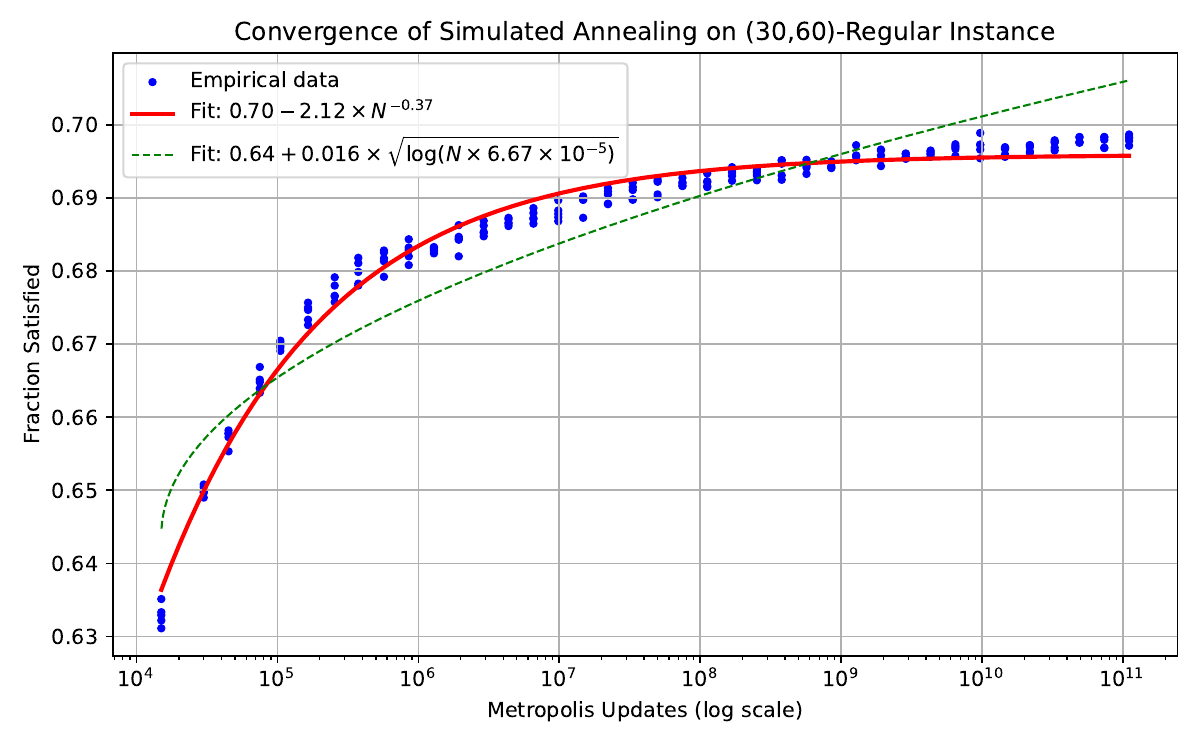} \\
            \includegraphics[width=0.8\textwidth]{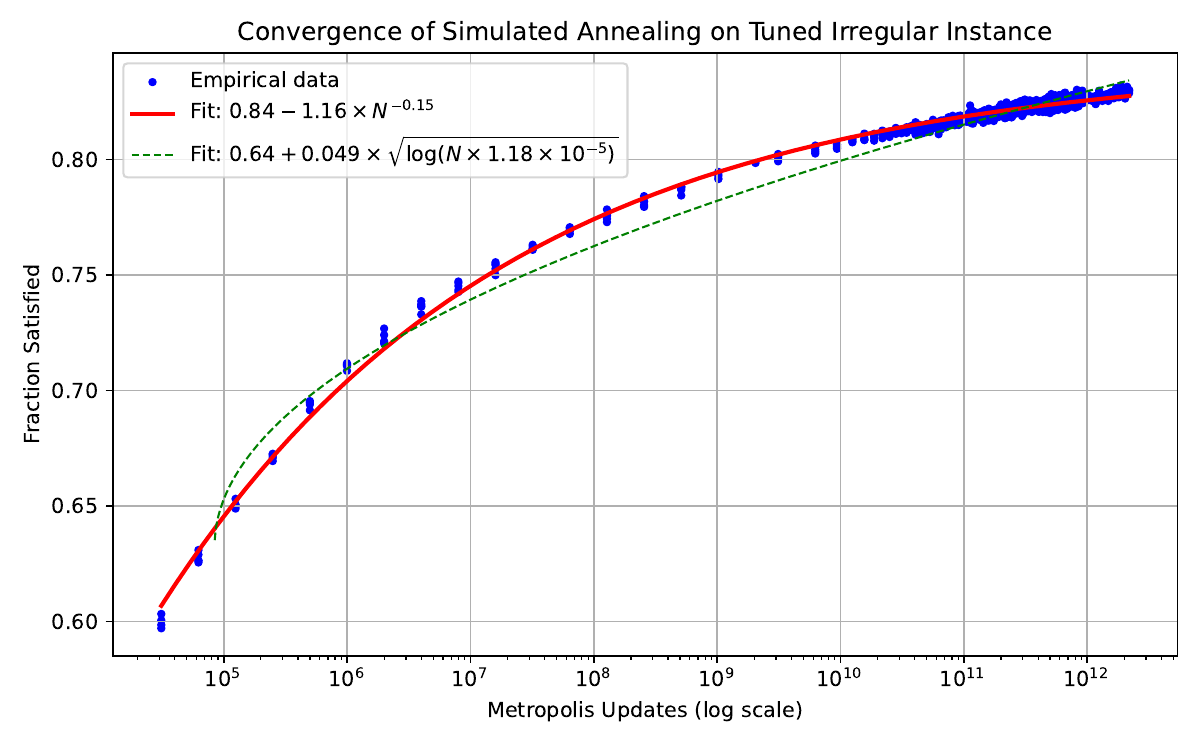}
        \end{array}
    \]
    \caption{\label{fig:power_law} Here we show the dependence of the fraction of constraints satisfied $\phi$ on the number of Metropolis updates $N$ used in simulated annealing. That is, $N$ is the number of variables times the number of sweeps. We find that the functional form $\phi = a + b \sqrt{c \log N}$ suggested by the heuristic argument of \sect{sec:local} fits poorly, but the form $\phi = a - b N^{-c}$ fits reasonably well. On the top panel we consider an instance from Gallager's ensemble with $k=30$, $D=60$, and $n=15,000$. On the bottom we use the irregular instance defined in \sect{sec:wins}, which has $n=31,216$. Each data point represents the final outcome of an independent anneal, and in each anneal we vary $\beta$ linearly from zero to five.}
\end{figure}

Although we do not have a strong theoretical handle on the $N$-scaling, we can nevertheless exploit the simple observation that increasing $N$ has diminishing returns. Thus, in all empirical analysis relating to simulated annealing, by choosing $N$ large one can ensure data points are on the relatively flat tail of the convergence curve and thus relatively insensitive to the specific choice of $N$. This renders the trends noted in figures \ref{fig:SAfits}, \ref{fig:advrand_etc}, \ref{fig:shannon_regions}, and \ref{fig:RSSA} relatively robust to choice of sweep count. In each of these plots we also kept the number of sweeps fixed for all data points plotted in order to minimize the effect of this as a confounding factor.

We next consider the problem of comparing the performance of DQI against simulated annealing. This task is rendered complicated by the fact that the number of clauses that simulated annealing is able to satisfy depends on how long one is willing to run the anneal. In Fig. \ref{fig:mega_convergence}, we plot the convergence of simulated annealing as a function of number of sweeps for the instance defined in \sect{sec:wins} and compare against the satisfaction fraction achieved by DQI+BP. We use the same data as in the lower panel of Fig. \ref{fig:power_law} but using a linear instead of logarithmic scale on the horizontal axis and zooming in on the region where the number of sweeps is at least $10^5$.

In all anneals, we start with inverse temperature $\beta = 0$ and then linearly increase with each sweep to a final value of $\beta=5$. Our results suggests this is an effective annealing schedule, although we do not claim it is precisely optimal. In each sweep, we cycle through the $n=31,216$ bits, and for each one consider a move in which the bit is flipped, accepting the move according to the Metropolis criterion. For this instance, our implementation of simulated annealing, which is optimized C++ code, is able to execute approximately 16 sweeps ($5 \times 10^5$ Metropolis moves) per second, though this varied slightly from run to run likely due to the fact that accepted Metropolis moves incur a larger computational cost than rejected moves in our implementation. Among our 1,079 annealing experiments, the largest number of sweeps we carried out was $7 \times 10^7$, and the longest anneal completed after 118 hours of runtime.

The first anneal to exceed the satisfaction fraction of $0.831$ guaranteed by Theorem \ref{thm:imperfect_decoding} for BP+DQI used $6 \times 10^7$ sweeps and ran for 73 hours. Since we ran five anneals at each number of sweeps we estimate that running five independent anneals in parallel at $6 \times 10^7$ sweeps would yield a collection of approximate optima, the best of which has a nontrivial chance of exceeding $0.831$.

\begin{figure}[ht]
    \begin{center}
        \includegraphics[width=0.6\textwidth]{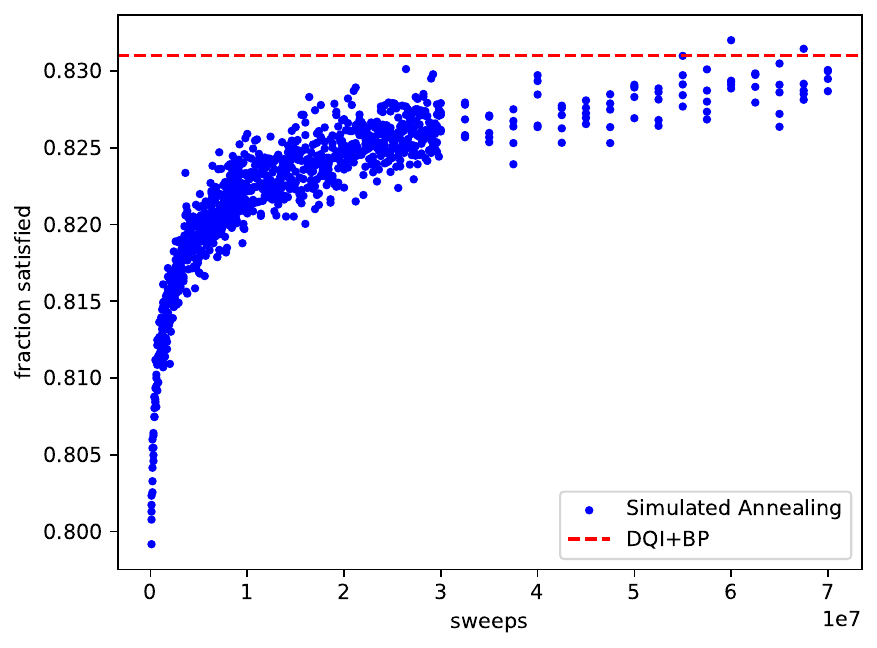}
    \end{center}
    \caption{\label{fig:mega_convergence} Here we show the satisfaction fraction achieved by simulated annealing as a function of the number of sweeps. For each number of sweeps we run five independent executions of simulated annealing with different pseudorandom seeds. At the right-hand side of the plot we incremented the number of sweeps by larger increments due to computational cost.}
\end{figure}

\subsection{Prange's Algorithm}
\label{sec:trunc}

Consider an instance of max-XORSAT
\begin{equation}
    \label{eq:vecpremise_repeat}
    B \mathbf{x} \stackrel{\max}{=} \mathbf{v}
\end{equation}
where $B$ is an $m \times n$ matrix over $\mathbb{F}_2$ and $m > n$. The system \eq{eq:vecpremise_repeat} is therefore overdetermined and we wish to satisfy as many equations as possible. In Prange's algorithm we simply throw away all but $n$ of the linear equations from this system so that it is no longer overdetermined. Then, provided the remaining system is not singular, we can simply solve it, \textit{e.g.} by Gaussian elimination. We thus obtain a bit string that definitely satisfies $n$ of the original $m$ constraints. Heuristically, one expects the remaining $m-n$ constraints each to be satisfied with probability $1/2$, independently. In other words the number of these $m-n$ constraints satisfied is binomially distributed. One reruns the above steps polynomially many times with different random choices of constraints to solve for. In this manner one can reach a logarithmic number of standard deviations onto the tail of this binomial distribution. Consequently, for max-XORSAT, the number of constraints one can satisfy using Prange's algorithm with polynomially many trials is $n + (m-n)/2 + \widetilde{\mathcal{O}}(\sqrt{m})$.

One can make the procedure more robust by bringing $B^T$ to reduced row-echelon form rather than throwing away columns and hoping that what is left is non-singular. In this case, as long as $B^T$ is full rank, one can find a bit string satisfying $n$ constraints with certainty. From numerical experiments one finds that matrices from Gallager's ensemble often fall just slightly short of full rank. Nevertheless, this procedure, when applied to Gallager's ensemble, matches very closely the behavior predicted by the above argument.

This heuristic generalizes straightforwardly to max-LINSAT. First, consider the case that $|f_i^{-1}(+1)| = r$ for all $i=1,\ldots,m$. In this case, we throw away all but $n$ constraints, and choose arbitrarily among the $r$ elements in the preimage $f_i^{-1}(+1)$ for each of those that remain. After solving the resulting linear system we will satisfy all of these $n$ constraints and on average we expect to satisfy a fraction $r/p$ of the remaining $n-m$ constraints assuming the $f_i$ are random. Hence, with polynomially many randomized repetitions of this scheme, the number of constraints satisfied will be $n + (m-n)(r/p) + \widetilde{\mathcal{O}}(\sqrt{m}).$ That is, the fraction $\phi_{\mathrm{PR}}$ of the $m$ constraints satisfied by a solution found by Prange's algorithm will be 
\begin{equation}
    \label{eq:truncfrac}
    \phi_{\mathrm{PR}} = \frac{r}{p} + \left( 1 - \frac{r}{p} \right) \frac{n}{m} + \widetilde{\mathcal{O}}(1/\sqrt{m}).
\end{equation}
If the preimages $f_i^{-1}(+1)$ for $i=1,\ldots,m$ do not all have the same size, then one should choose the $n$ constraints with smallest preimages as the ones to keep in the first step. The remaining $m-n$ are then those most likely to be satisfied by random chance.

\subsection{The AdvRand Algorithm}
\label{sec:advrand}

Prompted by some successes \cite{FGG14} of the Quantum Approximate Optimization Algorithm (QAOA), a simple but interesting algorithm for approximating max-XORSAT was proposed by Barak \textit{et al.} in \cite{BM15}, which the authors named AdvRand. Although designed primarily for the purpose of enabling rigorous average-case performance guarantees, the AdvRand algorithm can also be tried empirically as a heuristic, much like simulated annealing, and compared against DQI.

Given an instance of max-XORSAT with $n$ variables, the AdvRand algorithm works as follows. Select two parameters $R,F \in (0,1)$. Repeat the following sequence of steps polynomially many times. First assign $Rn$ of the variables uniformly at random. Substitute these choices into the instance, yielding a new instance with $(1-R)n$ variables. A constraint of degree $D$ will become a constraint of degree $D-r$ if $r$ of the variables it contains have been replaced by randomly chosen values. Thus, in the new instance some of the resulting constraints may have degree one. Assign the variables in such constraints to the values that render these constraints satisfied. If there are remaining unassigned variables, assign them randomly. Lastly, flip each variable independently with probability $F$.

In \cite{BM15}, formulas are given for $R$ and $F$ that enable guarantees to be proven about worst case performance. Here, we treat $R$ and $F$ as hyperparameters. We set $F=0$ and exhaustively try all values of $Rn$ from $0$ to $n$, then retain the best solution found. Our results on Gallager's ensemble at $k=3$ are displayed and compared against simulated annealing and DQI+BP in Fig. \ref{fig:advrand_etc} for constant $n$ and growing $D$. 

\begin{figure}
\begin{center}
    \begin{tabular}{c}
    \includegraphics[width=\textwidth]{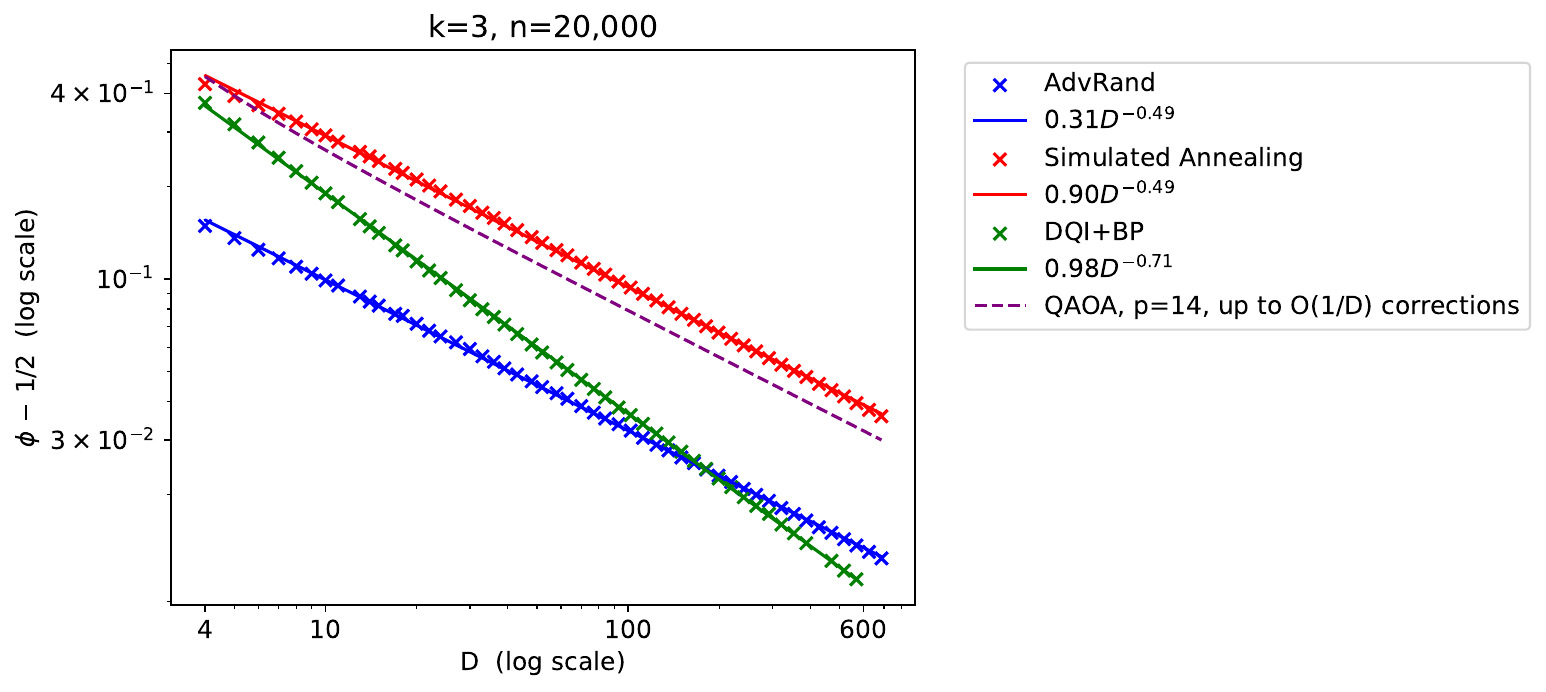}
    \end{tabular}
    \caption{\label{fig:advrand_etc} On Gallager's ensemble we compare the fraction $\phi$ of constraints satisfied in the solutions found by DQI+BP, AdvRand, and simulated annealing. We also show the approximate performance of $p=14$ QAOA on large girth $D$-regular max-3-XORSAT instances, which was calculated in \cite{BFM22}, up to $\mathcal{O}(1/D)$ corrections. We fix $k$, the number of variables in each constraint, at 3 and we vary $D$, the number of constraints that each variable is contained in, from $4$ to $687$. The red, blue, and green lines display the linear least-square fits to the log-log plot of $\phi-1/2$ versus $D$. We keep the number of variables fixed at $n=20,000$, thus $m = Dn/3$.}
    \end{center}
\end{figure}

In \cite{BM15} it was proven that, given max-$k$-XORSAT instances, AdvRand can in polynomial time find solutions satisfying a fraction $\frac{1}{2} + \frac{e^{-\mathcal{O}(k)}}{\sqrt{D}}$ of the constraints, even in the worst case, provided each variable is contained at most $D$ constraints. Our observed empirical performance in Fig. \ref{fig:advrand_etc} is in good agreement with this. In \cite{MH22} it was shown that at large $D$, for average-case degree-$D$ max-3-XORSAT, the exact optimum concentrates at $\frac{1}{2} + \frac{0.9959}{\sqrt{D}}$, in the limit of large $D$. Thus, for fixed $k$, the functional form of the scaling of AdvRand is provably optimal, and the key metric of performance for given $k$ is the specific value of the numerator $e^{-\mathcal{O}(k)}$. For our experiments at $n=20,000$ with $k=3$ we empirically observe a value of 0.31.

In the simulated annealing experiments shown in Fig. \ref{fig:advrand_etc}, we vary $\beta$ linearly from 0 to 3 and apply $5,000$ sweeps through the variables, \textit{i.e.} $5,000n$ Metropolis updates.

\subsection{Quantum Approximate Optimization Algorithm}
\label{sec:QAOA}

In 2014, Farhi, Goldstone, and Gutmann introduced a new quantum algorithm for optimization that they called the Quantum Approximate Optimization Algorithm (QAOA) \cite{FGG14a}. The QAOA algorithm is parameterized by a number of rounds, $p$. Allowing additional rounds can only improve the approximate optima found by QAOA, but this also makes the algorithm harder to analyze theoretically. The largest $p$ for which QAOA's performance has been analyzed on max-3-XORSAT is $p=14$, which was achieved in \cite{BFM22} using nontrivial tensor network techniques, which apply to all $D$-regular max-3-XORSAT instances whose hypergraphs have girth greater than $2p+1$. In \cite{BFM22} it was found that the fraction of satisfied clauses for every $D$-regular large-girth hypergraph, or random $D$-regular hypergraphs in the $n\to \infty$ limit, is
\begin{equation}
    \label{eq:QAOA}
    \phi_{\mathrm{QAOA}} = \frac{1}{2} + \bar{\nu}_{14}^{[3]}\sqrt{\frac{3}{2(D-1)}} \pm \mathcal{O}(1/D),
\end{equation}
where $\bar{\nu}_{14}^{[3]} = 0.6422$ \cite{bensgithub}.

In Fig. \ref{fig:advrand_etc} we include a plot of the line $\phi_{\mathrm{QAOA}} = \frac{1}{2} + \bar{\nu}_{14}^{[3]}\sqrt{\frac{3}{2(D-1)}}$, alongside the empirical average-case performance on the Gallager ensemble of DQI using a standard belief propagation decoder (DQI+BP). At small $D$ the comparison between $p=14$ QAOA and DQI+BP is not fully conclusive due to the unknown $\mathcal{O}(1/D)$ corrections in \eq{eq:QAOA}, but for all $D$ for which we can draw firm conclusions, \textit{i.e.} large $D$, QAOA at $p=14$ outperforms DQI+BP on max-3-XORSAT.
 
A second point of comparison between DQI and QAOA is the Sherrington-Kirkpatrick model, analyzed in \cite{BFM22,BKL25}. The analysis in \cite{BFM22} shows that for average-case max-2-XORSAT containing all $\binom{n}{2}$ possible constraints of the form $x_i \oplus x_j = v_k$, each with random $v_k \in \{0,1\}$, the number of satisfied minus unsatisfied clauses achieved by QAOA scales as $\nu_p n^{3/2}$, where $\nu_p$ is a constant that depends on $p$, the number of rounds in the QAOA algorithm. Using Theorem \ref{thm:DQI_shannon_limit} from \sect{sec:limits} one finds that DQI, even using a classical decoder saturating the Shannon bound, would achieve at best $O(n^{3/2}/\sqrt{\log n})$. Thus, DQI is not competitive on this problem, at least using classical decoders\footnote{One could also consider using quantum decoders such as BPQM \cite{PR22}, which take advantage of the coherence of the errors and are limited only by the Holevo bound rather than the Shannon bound.}.

A third point of comparison between DQI and QAOA is max-2-XORSAT where each variable is contained in exactly $D$ constraints. MaxCut for $D$-regular graphs is the special case of this where $\mathbf{v}$ is the all ones vector. In \cite{FGRV25} it was shown that QAOA with 17 rounds, when applied to the MaxCut problem on any 3-regular graph of girth at least 36 can achieve a cut fraction of $0.8971$. For random 3-regular graphs, and any constant $g$, as the number of vertices goes to infinity, the fraction of vertices that are involved in loops of size $g$ goes to zero. Thus for average-case instances of 3-regular MaxCut, QAOA with 17 rounds can asymptotically achieve cut fraction $0.8971$. Furthermore, since the performance of QAOA on max-XORSAT is independent of $\mathbf{v}$, QAOA can also satisfy fraction $0.8971$ of the constraints for 3-regular max-2-XORSAT on average-case graphs and arbitrary $\mathbf{v}$. By Theorem \ref{thm:max2limit}, the satisfaction fraction achievable by DQI with classical decoding is asymptotically upper bounded for 3-regular max-2-XORSAT by $0.75$ if $\mathbf{v}$ is chosen uniformly at random. Thus, at least for random $\mathbf{v}$, DQI with classical decoders is beaten by QAOA on 3-regular max-2-XORSAT.

In \cite{BFM22} it was shown that when applied to MaxCut problems on large-girth $D$-regular graphs, QAOA achieves a cut fraction
\begin{equation}
\frac{1}{2} + \frac{\nu_p}{\sqrt{D-1}} \pm \mathcal{O}(1/D),
\end{equation}
where $\nu_p$ is a constant that increases with the number of rounds $p$ and for which $\nu_{17} = 0.6773$. As in the case of $D=3$ described above, this implies the same asymptotic performance on average-case $D$-regular max-2-XORSAT. Since this analysis is only up to $\mathcal{O}(1/D)$ corrections, it cannot be quantitatively compared at finite $D$ against DQI. Nevertheless, Theorem \ref{thm:max2limit} shows that the approximation to average-case $D$-regular max-2-XORSAT achievable by DQI with classical decoders is limited to $1/2 + 1/(2D-2)$. Thus, QAOA outperforms DQI with classical decoders in the limit where $k=2$ and $D$ is large. To search for regimes of advantage for DQI one could instead consider increasing $k$ together with $D$, or using quantum decoders, as discussed in \sect{sec:limits}.

\subsection{Algebraic Attacks Based on List Recovery}
\label{sec:list-recovery}

For our OPI instances, we believe that the most credible classical algorithms to consider as competitors to DQI must be attacks that exploit algebraic structure. 

The max-LINSAT problem can be viewed as finding a codeword from $C$ that approximately maximizes $f$. With DQI we have reduced this to a problem of decoding $C^\perp$ out to distance $\ell$. For a Reed-Solomon code, as defined in \eq{eq:BRS}, its dual is also a Reed-Solomon code. Hence, both $C$ and $C^\perp$ can be efficiently decoded out to half their distances, which are $m-n+1$ and $n+1$, respectively. However, max-LINSAT is not a standard decoding problem, \textit{i.e.} finding the nearest codeword to a given string under the promise that the distance to the nearest string is below some bound. In fact, exact maximum-likelihood decoding for general Reed-Solomon codes with no bound on distance is known to be NP-hard \cite{GV05, GGG18}.

The OPI problem is very similar to a problem studied in the coding theory literature known as \emph{list-recovery}, applied in particular to Reed-Solomon codes.  In list-recovery, for a code $C \subseteq \mathbb{F}_p^m$, one is given sets $F_0, F_1, \ldots, F_{m-1} \subseteq \mathbb{F}_p$ (which correspond to our sets $f_i^{-1}(+1)$), and asked to return all codewords $c \in C$ so that $c_i \in F_i$ for as many $i$ as possible.  It is easy to see that solving this problem for Reed-Solomon codes will solve the OPI problem, assuming the list of all matching codewords is small.  However, existing list-recovery algorithms for Reed-Solomon codes rely on the size of the $F_i$ being quite small (usually constant, relative to $m$) and do not apply in this parameter regime.  In particular, the best known list-recovery algorithm for Reed-Solomon codes is the Guruswami-Sudan algorithm~\cite{GS98}; but this algorithm breaks down when the size of the $|F_i|$ is larger than $m/n$ (this is the \emph{Johnson bound} for list-recovery).  In our setting, $|F_i| = p/2 \approx m/2$, which is much larger than $m/n \approx 10$.  Thus, the Guruswami-Sudan algorithm does not apply.  Moreover, we remark that in this parameter regime, if the $f_i$ are random, we expect there to be exponentially many codewords satisfying all of the constraints; this is very different from the coding-theoretic literature on list-recovery, which generally tries to establish that the number of such codewords is at most polynomially large in $m$, so that they can all be returned efficiently.  Thus, standard list-recovery algorithms are not applicable in the parameter regime we consider.

\subsection{Lattice-Based Heuristics}
\label{sec:lattice}

A problem similar to our Optimal Polynomial Intersection problem---but in a very different parameter regime---has been considered before, and has been shown to be susceptible to lattice attacks.  In more detail, in the work~\cite{NP99}, Naor and Pinkas proposed essentially the same problem, but in the parameter regime where $p$ is exponentially large compared to $m$, and where $|f_i^{-1}(+1)| \ll p/2$ is very small (in particular, not balanced, like we consider).\footnote{In this parameter regime, unlike ours, for random $f_i$ it is unlikely that there are \emph{any} solutions $\mathbf{x}$ with $f(\mathbf{x})$ appreciably large, so the problem of \cite{NP99} also ``plants'' a solution $\mathbf{x}^*$ with $f(\mathbf{x}^*) = m$; the problem is to find this planted solution.  Another difference is that DQI attains \eqref{eq:truncfrac} for \emph{any} functions $f_i$, while in the conjecture of Naor and Pinkas the $f_i$ are random except for the values corresponding to the planted solution.} The work~\cite{NP99}  conjectured that this problem was (classically) computationally difficult.  This conjecture was challenged by Bleichenbacher and Nguyen in~\cite{BN00} using a lattice-based attack, which we describe in more detail below. However, this lattice-based attack does not seem to be effective---either in theory or in practice---against our OPI problem.  Intuitively, one reason is that in the parameter regime that the attack of \cite{BN00} works, a solution to the max-LINSAT problem---which will be unique with high probability---corresponds to a unique shortest vector in a lattice, which can be found via heuristic methods.  In contrast, in our parameter regime, there are many optimal solutions, corresponding to many short vectors; moreover, empirically it seems that there are much shorter vectors in the appropriate lattice that do not correspond to valid solutions.  Thus, these lattice-based methods do not seem to be competitive with DQI for our problem.

The target of the attack in \cite{BN00} is syntactically the same as our problem, but in a very different parameter regime.  Concretely, $p$ is chosen to be much larger than $m$ or $n$, while the size $r:=|f_i^{-1}(+1)|$ of the set of ``allowed'' symbols for each $i = 1, \ldots, m$ is very small. (Here, we assume that $f_i^{-1}(+1)$ has the same size $r$ for all $i$ for simplicity of presentation; this can be relaxed).  In \cite{BN00}, the $f_i$'s are chosen as follows.  Fix a planted solution $\mathbf{x}^*$, and set the $f_i$ so that $f_i(\mathbf{b}_i \cdot \mathbf{x}^*) = 1$ for all $i$; thus $f(\mathbf{x}^*) = m$.  Then for each $i$, $f_i(y)$ is set to $1$ for a few other random values of $y \in \F_p$, and $f_i(z) = -1$ for the remaining $z \in \F_p$.  With high probability, $\mathbf{x}^*$ is  the unique vector with large objective value, and the problem is to find it.

In our setting, where $r = |f_i^{-1}(+1)| \approx p/2$, $m=p-1$ and $n = \lceil m/10 \rfloor$, we expect there to be many vectors $\mathbf{x}$ with $f(\mathbf{x}) = m$ when the $f_i$ are random.  We have seen that DQI can find a solution with $f(\mathbf{x}) \approx 0.7179m$ (even for arbitrary $f_i$).  

The way the attack of \cite{BN00} works in our setting is the following.  For $\mathbf{x} \in \F_p^n$, define a polynomial $P_{\mathbf{x}}(Z) = \sum_{j=1}^n x_j Z^{j-1}$. Let $F_i = f_i^{-1}(+1)$, and write $F_i = \{v_{i,1}, v_{i,2}, \ldots, v_{i,r}\} \subseteq \F_p$.   If $f(\mathbf{x}) = m$, then $P_{\mathbf{x}}(\gamma^i) \in F_i$ for all $i = 1, \ldots, m$, where we recall from \sect{sec:OPI} that $\gamma$ is a primitive element of $\F_p$, and $B_{i,j} = \gamma^{ij}$.  Let $L_i(Z) = \prod_{j \neq i}\frac{Z - \gamma^j}{\gamma^i - \gamma^j}.$ By Lagrange interpolation, we have
\begin{equation}
    \label{eq:lagrange}
    P_{\mathbf{x}}(Z) = \sum_{i=1}^m P_\mathbf{x}(\gamma^i) L_i(Z) = \sum_{i=1}^m \left( 
    \sum_{j=1}^r \delta_{i,j}^{(P)} \ v_{i,j} \right) L_i(Z),
\end{equation}
where
$\delta_{i,j}^{(P)}$ is $1$ if $P_\mathbf{x}(\gamma^i) = v_{i,j}$ and $0$ otherwise.  Since $P_\mathbf{x}(Z)$ has degree at most $n-1$, the coefficients on $Z^k$ are equal to zero for $k=n,n+1, \ldots, p-1$.  Hence, \eq{eq:lagrange} gives us $p-n = m-n+1$ $\F_p$-linear constraints on the vector $\bm{\delta} \in \{0,1\}^{rm}$ that contains the $\delta_{i,j}^{(P)}$'s. Collect these constraints in a matrix $A \in \F_p^{m-n+1 \times rm}$, so that $A \bm{\delta} = 0$.  Consider the lattice $\Lambda := \{ \bm{\delta}' \in \mathbb{Z}^{rm}\,|\, A\bm{\delta}' = 0 \mod p \}.$  Our target vector $\bm{\delta}$ clearly lies in $\Lambda$, and moreover it has a very short $\ell_2$ norm: $\|\bm{\delta}\|_2 = \sqrt{m}$.  Thus, we may hope to find it using methods like LLL \cite{LLL82} or Schnorr's BKZ reduction \cite{S87,BKZ94}, and this is indeed the attack.\footnote{There are further improvements given in \cite{BN00}, notably passing to a sub-lattice that enforces the constraint that $\sum_{j=1}^r \delta_{i,j}^{(P)}$ is the same for all $i$.}
Upon finding a vector $\bm{\delta}$ of the appropriate structure (namely, so that $\delta_{i,j}^{(P)} = 1$ for exactly one $j \in \{1, \ldots, r\}$ for each $i$), we may read off the evaluations of $P_\mathbf{x}$ from the $F_i$, and hence recover $\mathbf{x}$.

Bleichenbacher and Nguyen show that in some parameter regimes, the target vector $\bm{\delta}$ with length $\sqrt{m}$ is likely to be the shortest vector in $\Lambda$.  However, these parameter regimes are very different from ours.  For example, their results hold for $p \approx 2^{80}$, and $r \leq 16$, with codes of rate $n/m$ at least $0.88$.  In contrast, in our setting we have {much} larger lists, with $r \approx p/2$, and much lower-rate codes, with $n/m \approx 1/10$ (although in \cite{BN00}, they take $m,n \ll p$ while we take $m \approx p$, so this is not a direct comparison).

Empirically, this attack does not seem to work in our parameter regime.  Indeed, the lattice heuristics do find short vectors in the lattice $\Lambda$, but these vectors are much shorter than $\sqrt{m}$ whenever $m$ and $n_{\text{distractors}}$ are comparable to $p$ (see Fig. \ref{fig:BleichenbacherNguyenAttackSVPinvalidSolutions}).
As a consequence, the success probability when applied to our OPI instances appears to decay exponentially with $p$.

\begin{figure}[!htb]
    \centering
    \includegraphics[height=0.29\linewidth]{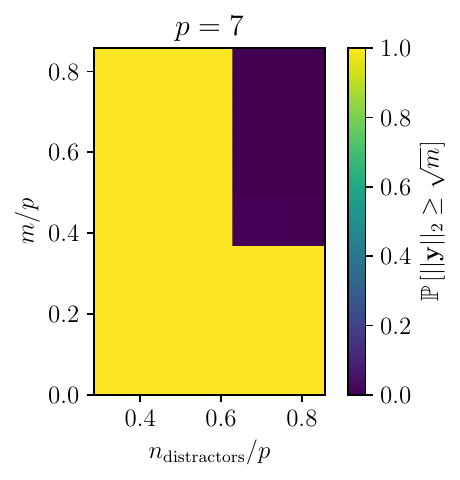}
    \includegraphics[height=0.29\linewidth]{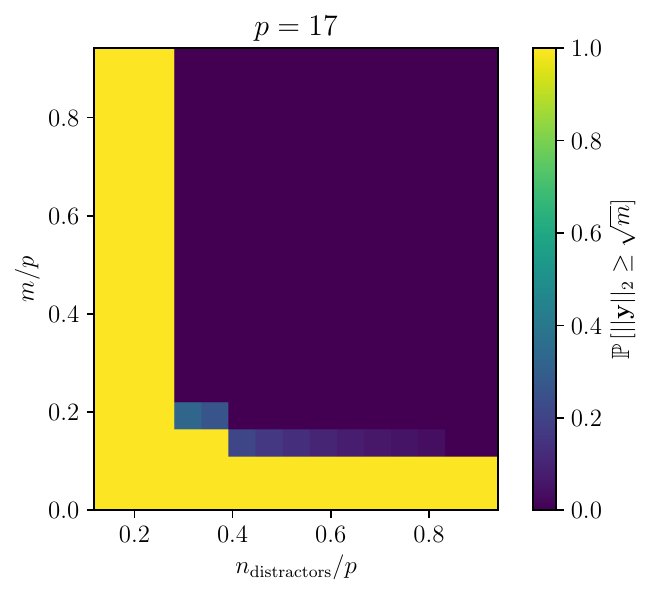}
    \includegraphics[height=0.29\linewidth]{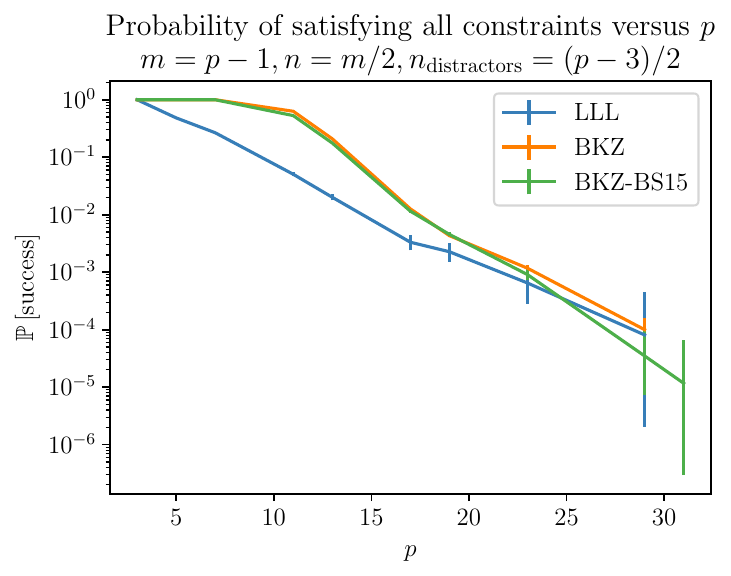}
    \caption{The attack of Bleichenbacher and Nguyen \cite{BN00} is workable when the shortest nonzero vector in a particular lattice has weight $\sqrt{m}$. Above, left, we apply the BKZ algorithm \cite{BKZ94} to find the shortest nonzero vector (under the 2-norm) in the lattices arising from random problem instances for various $m$ and $n_{\text{distractors}}.$ We observe that the shortest vector almost always has 2-norm $< \sqrt{m}$ in the regime where $m$ and $n_{\text{distractors}}$ are both a significant fraction of $p$. 
    Consequently, the success probability of the attack when applied to our OPI problem in the regime where $m=p-1, n=m/2, n_{\text{distractors}} = (p-3)/2$ appears to exhibit exponential decay with $p$ whether LLL, BKZ with a block size of 15 (``BKZ-15'') or BKZ with unlimited block size.
    In this regime we believe the lattice-based heuristic of \cite{BN00} does not succeed.}
    \label{fig:BleichenbacherNguyenAttackSVPinvalidSolutions}
\end{figure}

\section{Max-XORSAT Instances Advantageous to DQI Over Simulated Annealing}
\label{sec:wins}

In this section we construct a class of max-XORSAT instances such that DQI, using belief propagation decoders, achieves a better approximation than we are able to achieve using simulated annealing if we restrict simulated annealing to a comparable number of computational steps. DQI+BP also achieves a better approximation than we obtain from any of the other general-purpose optimization algorithms that we try: greedy optimization, Prange's algorithm, and AdvRand. However, we do not claim this as an example of quantum advantage because we are able to also construct a classical heuristic tailored to the class of instances which, within reasonable runtime, beats the approximation achieved by DQI+BP. Also, as noted in the introduction, using very long anneals (up to 118 hours) we are able to reach the satisfaction fraction achieved by DQI+BP for the instance considered here. We have left the systematic investigation the scaling with $n$ of the runtime of simulated annealing for these instances to future work.

Given a max-XORSAT instance $B \mathbf{x} \stackrel{\max}{=} \mathbf{v}$, the degree of a variable is the number of constraints in which it is contained. The degree of a constraint is the number of variables that are contained in it. Hence, the degree of the $i\th$ constraint is the number of nonzero entries in the $i\th$ row of $B$ and the degree of the $j\th$ variable is the number nonzero entries in the $j\th$ column of $B$. For an LDPC code, the degree of a parity check is the number of bits that it contains, and the degree of a bit is the number of parity checks in which it is contained. These degrees correspond to the number of nonzero entries in the rows and columns of the parity check matrix.

\begin{figure}
    \centering
\begin{tikzpicture}[scale=.7]
    \foreach \i in {0,1,2,3,4,5,6}{
    \node[draw=orange!30!red, circle,fill=orange!20,minimum size=0.5cm](a\i) at (0,\i) {};
    }
    \foreach \i in {0,1,2,3,4,5,6,7,8,9}{
    \node[draw=blue, rectangle,fill=blue!20, minimum size=0.5cm](b\i) at (4,\i-1.5) {};
    }
    \draw (a6)--(b9);
    \draw (a6)--(b7);
    \draw (a5)--(b9);
    \draw (a5)--(b8);
    \draw (a4)--(b6);
    \draw (a4)--(b2);
    \draw (a3)--(b7);
    \draw (a3)--(b5);
    \draw (a3)--(b3);
    \draw (a2)--(b8);
    \draw (a2)--(b4);
    \draw (a2)--(b3);
    \draw (a1)--(b1);
    \draw (a1)--(b4);
    \draw (a1)--(b3);
    \draw (a0)--(b0);
    \draw (a0)--(b4);
    \draw (a0)--(b5);
    \draw (a0)--(b6);
    \node[orange,anchor=east] at (0,7) {$n$ variables};
    \node[blue,anchor=west] at (4,8.5) {$m$ constraints};
    \draw[decorate,decoration={brace,amplitude=10pt}]
  (-.5,-.5) -- (-.5,.5) node[midway,xshift=-1em,anchor=east]{$\Delta_4 = \frac{1}{7}$};
   \draw[decorate,decoration={brace,amplitude=10pt}]
  (-.5,.5) -- (-.5,3.5) node[midway,xshift=-1em,anchor=east]{$\Delta_3 = \frac{3}{7}$};
  \draw[decorate,decoration={brace,amplitude=10pt}]
  (-.5,3.5) -- (-.5,6.5) node[midway,xshift=-1em,anchor=east]{$\Delta_2 = \frac{3}{7}$};
  \draw[decorate,decoration={brace,mirror,amplitude=10pt}]
  (4.5,-2) -- (4.5,1) node[midway,xshift=1em,anchor=west]{$\kappa_1 = \frac{3}{10}$};
  \draw[decorate,decoration={brace,mirror,amplitude=10pt}]
  (4.5,1) -- (4.5,3) node[midway,xshift=1em,anchor=west]{$\kappa_3 = \frac{1}{5}$};
  \draw[decorate,decoration={brace,mirror,amplitude=10pt}]
  (4.5,3) -- (4.5,8) node[midway,xshift=1em,anchor=west]{$\kappa_2 = \frac{1}{2}$};
\end{tikzpicture}
    \caption{Tanner graph for a sparse irregular LDPC code illustrating the notation introduced in \sect{sec:wins}.}
    \label{fig:sparse_irregular}
\end{figure}

Given a max-XORSAT instance, let $\Delta_j$ be the fraction of variables that have degree $j$. Let $\kappa_i$ be the fraction of constraints that have degree $i$. This is illustrated in Fig.~\ref{fig:sparse_irregular}. Via DQI, a max-XORSAT instance with degree distribution $\Delta$ for the variables and $\kappa$ for the constraints is reduced to a decoding problem for a code with degree distribution $\Delta$ for the parity checks and $\kappa$ for the bits. For regular LDPC codes, in which every bit has degree $k$ and every constraint has degree $D$, the error rate from which belief propagation can reliably decode deteriorates as $D$ increases. However, it has been discovered that belief propagation can still work very well for certain \emph{irregular} codes in which the average degree $\bar{D}$ of the parity checks is large \cite{RSU01}. In contrast, we find that the approximate optima achieved by simulated annealing on the corresponding irregular max-XORSAT instances of average degree $\bar{D}$ are typically no better than on regular instances in which every variable has degree exactly $\bar{D}$. This allows us to find examples where DQI achieves a better approximation than simulated annealing.

If the $m$ bits in an LDPC code have degree distribution $\kappa$ the total number of nonzero entries in the parity check matrix is $m \sum_i i \kappa_i$. If the $n$ parity checks in an LDPC code have degree distribution $\Delta$ then the total number of nonzero entries in the parity check matrix is $n \sum_j j \Delta_j$. Hence to define a valid LDPC code, a pair of degree distributions must satisfy
\begin{eqnarray}
    \textstyle{\sum_{i=1}^m \kappa_i} & = & 1 \\
    \textstyle{\sum_{j=1}^n \Delta_j} & = & 1 \\
    \textstyle{m \sum_{i=1}^m i \kappa_i} & = & \textstyle{n \sum_{j=1}^n j \Delta_j}.
\end{eqnarray}
Given any $\kappa_1,\ldots,\kappa_m$ and $\Delta_1,\ldots,\Delta_n$ satisfying these constraints, it is straightforward to sample uniformly from the set of all parity check matrices with $m$ bits whose degree distribution is $\kappa$ and $n$ parity checks whose degree distribution is $\Delta$. Furthermore, the maximum error rate from which belief propagation can reliably correct on codes from this ensemble can be computed in the limit of $m \to \infty$ by a method called density evolution, which numerically solves for the fixed point of a certain stochastic process \cite{RU01}. Using this asymptotic maximum error rate as an objective function, one can optimize degree distributions to obtain irregular LDPC codes that outperform their regular counterparts \cite{RSU01}.

As a concrete example, we consider the degree distribution shown in Fig.~\ref{fig:degree-dist-instance-beats-sa}.
\begin{figure}
    \centering
    \includegraphics[width=0.4\linewidth]{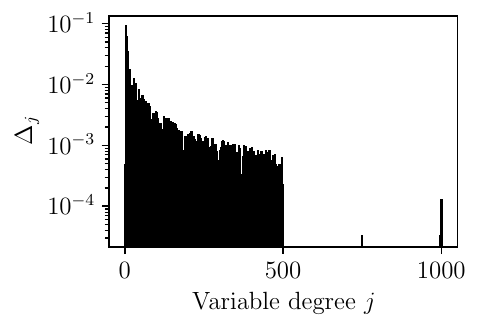}
    \includegraphics[width=0.4\linewidth]{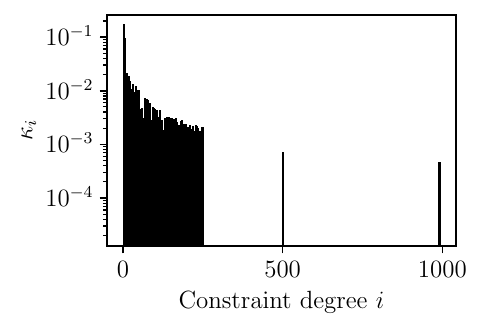}
    \caption{Degree distribution for an irregular instance sampled with $m=50,000$ and $n=31,216$. The full table of variable and constraint degrees is available in our Zenodo record \url{https://doi.org/10.5281/zenodo.13327870}.}
    \label{fig:degree-dist-instance-beats-sa}
\end{figure}
After generating a random max-XORSAT instance with 50,000 constraints and 31,216 variables consistent with this degree distribution we find that belief propagation fails on $9$ out of $10,000$ trials of decoding from uniformly random errors of Hamming weight $6,350$.
This bit flip error rate of $6,350/50,000$ is close to the asymptotic threshold of $\approx 13\%$ predicted by density evolution. Thus for average case $\mathbf{v}$, Theorem \ref{thm:imperfect_decoding} shows that DQI can asymptotically find solutions that satisfy at least $0.831m$ constraints with high probability.
In contrast, when we run simulated annealing on this instance with $10^6$ sweeps, for example, in sixteen trials the fraction of satisfied constraints ranges from $0.81024$ to $0.81488$. The performance of simulated annealing on this instance as a function of number of sweeps is discussed in detail in \sect{sec:sa_convergence}. Prange's algorithm would be predicted to achieve approximately $0.8122$, and experimentally we saw $0.8124$. The greedy algorithm performs far worse than simulated annealing on this instance. In sixteen trials its best satisfaction fraction was $0.666$. Our trial of AdvRand achieved $0.5536$.

The following classical algorithm can exceed the approximation achieved by DQI+BP on the above example, by exploiting the highly unbalanced degree distribution of its constraints. We modify simulated annealing (which is described in \sect{sec:local}) by adding a $\beta$-dependent factor to each term in the objective function. Letting $n_i$ denote the number of variables contained in constraint $i$, we use the objective function
\begin{equation}
    f^{(\beta)}(\mathbf{x}) = \sum_{i=1}^{m}\max\left(0, 1-e^{-\beta / n_i}\right)f_i\left(\sum_{j=1}^{n} B_{ij}x_j\right).
\end{equation}
We apply 1 million sweeps (since there are $50,000$ constraints this corresponds to fifty billion Metropolis updates), interpolating linearly from $\beta=0$ to $\beta = 5$. After this, we are left with solutions that satisfy approximately $0.88m$ clauses. We call this algorithm {\it irregular annealing} since it takes advantage of the irregularity of the instance by prioritizing the lower-degree constraints early in the annealing process.

\section{Limitations of DQI}
\label{sec:limits}

In addition to the power of DQI for solving optimization problems it is also interesting to delineate its fundamental limits. Because DQI reduces optimization problems to decoding problems, some limitations of DQI can be deduced from information theory. In this section we use information-theoretic considerations to prove upper bounds on the performance of DQI on max-XORSAT. We first consider the case where DQI uses a classical decoding algorithm implemented reversibly, as is done throughout the rest of this manuscript. We then consider the case of intrinsically quantum decoders. Lastly, we consider the special case of max-2-XORSAT, which in many respects behaves differently from max-$k$-XORSAT with $k \geq 3$.

\subsection{General Limitations of DQI Under Classical Decoding}

DQI reduces max-XORSAT to a decoding problem for a code $C^\perp$ with $m$ bits and $n$ parity checks. Hence its rate $R$ is
\begin{equation}
\label{eq:rate_def}
R = 1 - \frac{n}{m}.
\end{equation}
The decoding is required to succeed with high probability when an error string $\mathbf{e}$ with Hamming weight $\ell$ has been added to the codeword. We can model this by an error channel where each bit is independently flipped with probability 
\begin{equation}
\label{eq:bsc_def}
    p=\ell/m.
\end{equation}
This model is called the binary symmetric channel $\mathrm{BSC}(p)$. Although the distribution over error weights in DQI is given by $\mathrm{Pr}(|\mathbf{e}| = k) = |w_k|^2$, whereas $\mathrm{BSC}(p)$ has $\mathrm{Pr}(|\mathbf{e}| = k) = p^k (1-p)^{m-k} \binom{m}{k}$, the channels with these error distributions both asymptotically yield the same information-theoretic capacity because in both cases the probability distribution over error weights is narrowly peaked around $\ell$. As shown by Shannon, the rate of a code that can reliably transmit information over $\textrm{BSC}(p)$ is limited by
\begin{equation}
\label{eq:shannon_bound}
R \leq 1 - H_2(p),
\end{equation}
where $H_2(p) = -p\log_2(p)-(1-p)\log_2(1-p)$ is the binary entropy function. Substituting \eq{eq:rate_def} and \eq{eq:bsc_def} into \eq{eq:shannon_bound} yields $\frac{n}{m} \geq H_2 \left( \frac{\ell}{m} \right)$, which implies
\begin{equation}
\frac{\ell}{m} \leq H_2^{-1} \left( \frac{n}{m} \right)
\end{equation}
where the inverse $H_2^{-1}$ is well-defined, because $\ell/m \leq 1/2$. Substituting this into the semicircle law \eq{eq:dqifrac} yields the following bound
\begin{equation}
    \label{eq:s_shannon}
    \frac{\langle s \rangle_{\mathrm{Shannon}}}{m} \leq \frac{1}{2} + \sqrt{ H_2^{-1} \left( \frac{n}{m} \right) \left( 1 - H_2^{-1} \left( \frac{n}{m} \right) \right)}.
\end{equation}
Since $m H_2^{-1}(n/m)$ in general exceeds $d^\perp/2$, Theorem \ref{thm:semicircle} does not apply, though Theorem \ref{thm:imperfect_decoding} still does. That is, this must be interpreted as a bound on the performance of DQI for max-XORSAT with average case $\mathbf{v}$.

Let's now consider this limit in relation to the ensemble of degree-$D$ max-$k$-XORSAT instances in which $B$ is chosen from the $(k,D)$-regular Gallager ensemble. In this case we have $n/m = k/D$, which can be substituted into \eq{eq:s_shannon} to yield concrete upper bounds on $\langle s \rangle_{\mathrm{Shannon}} / m$. In Fig. \ref{fig:shannon_regions} we compare these upper bounds on the performance of DQI with classical decoding against the empirical performance of simulated annealing as well as the known asymptotic performance of Prange's algorithm.

We note that in Fig. \ref{fig:shannon_regions}, DQI is analyzed asymptotically, whereas simulated annealing results are obtained empirically at finite $n$. For simulated annealing we take $n \simeq 2,000$ and use $500,000$ sweeps, as we found these parameters to achieve a good tradeoff between asymptotic informativeness and computational convenience. To improve the results of simulated annealing we run it five times for each instance with different random seeds and report the maximum number of satisfied constraints achieved by any of these five trials. This technique is referred to as ``restarts,'' and it is often used in simulated annealing because running $r$ repetitions of simulated annealing and keeping the best result may often outperform the equally costly procedure of running a single anneal with $r$ times as many sweeps. We have $n \simeq 2,000$ rather than $n=2,000$ because in the Gallager ensemble $n$ must always be a multiple of $D$.

\begin{figure}[ht]
    \begin{center}
        \includegraphics[width=0.5\textwidth]{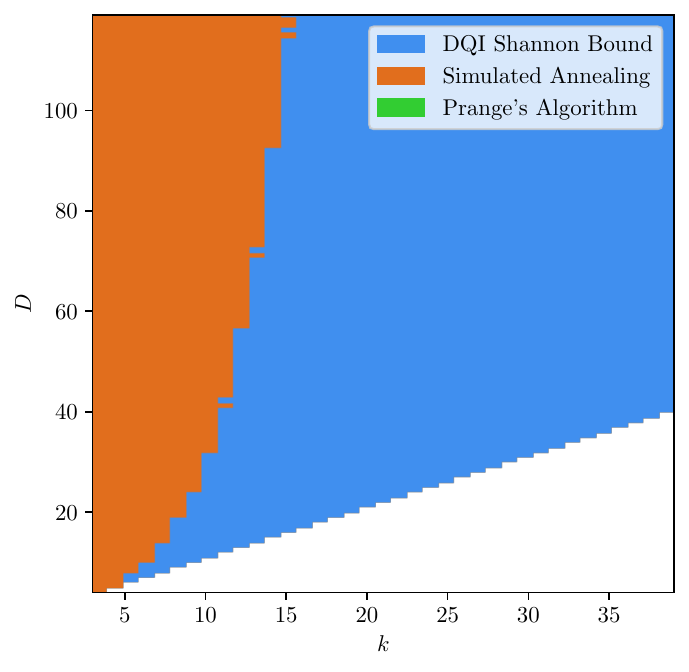}
    \end{center}
    \caption{\label{fig:shannon_regions} Here we consider degree-$D$ max-$k$-XORSAT instances $B\mathbf{x} \stackrel{\max}{=} \mathbf{v}$ where $\mathbf{v}$ is uniformly random and $B$ is drawn from the $(k,D)$-regular Gallager ensemble. In the orange region it is information-theoretically impossible for DQI with classical decoding to outperform simulated annealing on average-case instances from Gallager's ensemble. In the blue region, DQI with maximum-likelihood (\textit{i.e.} Shannon-limit) decoding achieves a higher satisfaction fraction than simulated annealing, but realizing this advantage with polynomial-time decoders remains an open problem. Prange's algorithm does not win on any region of this plot.}
\end{figure}

Although our main focus in this section is on limitations of DQI it is also worth discussing the possibilities of DQI. In the blue region of Fig. \ref{fig:shannon_regions} it is information-theoretically possible for DQI to achieve average-case quantum advantage using classical decoders. Based on computational experiments, we believe that this potential advantage is not realized by belief propagation decoding. This is because the number of errors that belief propagation can successfully correct falls increasingly short of the Shannon limit as the parity check matrices defining $C^\perp$ become denser. It is an open question whether efficient classical decoders can be devised that approach the Shannon limit closely enough for denser codes to allow DQI to outperform simulated annealing on average case instances from Gallager's ensemble. Efficient classical decoding of LDPC codes with denser than usual parity check matrices has so far not been a subject of intensive research but some results in this direction are obtained in \cite{FW05,FM07,DD08,TS08,YLB09,TRH09,TRH10}.

To make \eq{eq:s_shannon} less unwieldy, we can use the following useful bound from \cite{Cal09}.

\begin{theorem}
For all $x \in [0,1]$,
\begin{equation}
\label{eq:hbound}
\frac{x}{2 \log_2 \left( \frac{6}{x} \right)} \leq H_2^{-1}(x) \leq \frac{x}{\log_2 \left( \frac{1}{x} \right)}.
\end{equation}
\end{theorem}

\noindent
Substituting \eq{eq:hbound} into \eq{eq:s_shannon} yields the simpler but looser bound $\frac{\langle s \rangle}{m} \leq \frac{1}{2} + \sqrt{ \frac{n/m}{\log(m/n)}}.$ In summary, we have the following.
\begin{theorem}
\label{thm:DQI_shannon_limit}
Consider a max-XORSAT instance with $m$ constraints and $n$ variables where $\mathbf{v} \in \mathbb{F}_2^m$ is chosen uniformly at random. The expected number of satisfied clauses $\langle s \rangle$ obtained by DQI in the limit of large $n$ using a classical decoder is bounded by
\begin{equation}
\frac{\langle s \rangle}{m} \ \leq \ \frac{1}{2} + \sqrt{ H_2^{-1} \left( \frac{n}{m} \right) \left( 1 - H_2^{-1} \left( \frac{n}{m} \right) \right)} \ \leq \ \frac{1}{2} + \sqrt{ \frac{n/m}{\log(m/n)}}.
\end{equation}
\end{theorem}

\subsection{General Limitations of DQI Under Quantum Decoding}

The bit flip errors that must be decoded in DQI are in coherent superposition. One can treat these errors classically by implementing a classical decoding algorithm as a reversible circuit and separately performing classical error correction on each ``branch'' of the superposition. This is the strategy we describe and analyze throughout this manuscript. However, this is not necessarily optimal. Information-theoretically, at least, coherent errors are more advantageous than random errors.

In the preceding section, to avoid complications arising from the details of the distribution $\mathrm{Pr}(|\mathbf{e}| = k) = |w_k|^2$ over error weights arising in DQI, we approximated this by the binary symmetric channel with bit flip probability $p=\ell/m$. Similarly, we may approximate our distribution over coherent errors using the following channel
\begin{eqnarray*}
\ket{0} & \to & \ket{0_p} \\
\ket{1} & \to & \ket{1_p}
\end{eqnarray*}
where
\begin{eqnarray*}
\ket{0_p} & = & \sqrt{1-p} \ket{0} + \sqrt{p} \ket{1} \\
\ket{1_p} & = & \sqrt{p} \ket{0} + \sqrt{1-p} \ket{1}
\end{eqnarray*}
and $p = \ell/m$. The capacity of this channel is limited by Holevo's bound, which states
\begin{equation}
R \leq \chi(p),
\end{equation}
where
\begin{equation}
\chi(p) = S \left( \frac{1}{2} \ket{0_p}\bra{0_p} + \frac{1}{2} \ket{1_p} \bra{1_p} \right),
\end{equation}
and $S$ denotes the von Neumann entropy. By direct computation, one finds
\begin{equation}
\label{eq:chi}
\chi(p) = H_2 \left( \frac{1}{2} -\sqrt{p (1-p)} \right),
\end{equation}
where, as before, $H_2$ denotes the binary entropy function. When applying DQI to a max-XORSAT problem with $n$ variables and $m$ constraints, our decoding problem is for a code $C^\perp$ of rate $R=1-n/m$. Hence, the Holevo bound implies via the semicircle law \eq{eq:dqifrac}
\begin{equation}
\label{eq:sholevo}
\frac{\langle s \rangle_{\mathrm{Holevo}}}{m} \leq \frac{1}{2} + \sqrt{ \chi^{-1}\left( 1 - \frac{n}{m} \right) \left( 1 - \chi^{-1} \left( 1 - \frac{n}{m} \right) \right)}
\end{equation}
where the inverse $\chi^{-1}$ is well-defined, because $p=\ell/m \leq 1/2$. Using \eq{eq:chi} and simplifying one can equivalently write \eq{eq:sholevo} as
\begin{equation}
\frac{\langle s \rangle_{\mathrm{Holevo}}}{m} \leq 1 - H_2^{-1} \left(1-\frac{n}{m} \right).
\end{equation}

Next we apply this limit to Gallager's ensemble where $n/m = k/D$. If we compare this limit to the performance of simulated annealing and Prange's algorithm for average-case instances of max-XORSAT drawn from Gallager's ensemble, analogously to in Fig. \ref{fig:shannon_regions}, we do not find any region of the $(k,D)$-plane in which simulated annealing or Prange's algorithm beat this upper bound on DQI's performance. Thus, the Holevo bound does not allow us to rule out quantum advantage by DQI with quantum decoding for any region of the $(k,D)$ plane for Gallager-ensemble instances. But for other ensembles it may be successful in doing so.

In this section we consider the Holevo bound only as a tool for ruling out quantum advantage. But it also suggests that in regions where quantum advantage is not achievable using DQI with classical decoders it might be achievable using DQI with quantum decoders, if efficient quantum circuits can be found to implement them. Some exciting results on efficient quantum circuits for decoding coherent bit flip errors can be found in \cite{chailloux2024quantum, CLZ22, PR22}.

\subsection{Limitations of DQI with Classical Decoding for max-2-XORSAT}

The special case of max-$k$-XORSAT where $k=2$ behaves somewhat differently in the context of DQI than $k > 2$ and benefits from separate analysis. The max-2-XORSAT problem is widely studied, particularly the special case of max-2-XORSAT where $\mathbf{v}$ is the all-ones vector, which is known as MaxCut. Additionally, max-2-XORSAT is the unweighted special case of the Quadratic Unconstrained Binary Optimization (QUBO) problem.

The code $C^\perp$ dual to an instance of max-2-XORSAT is one in which each bit is contained in exactly two parity checks. Such codes are sometimes referred to as cycle codes (not to be confused with cyclic codes, which are unrelated). It is known that cycle codes have minimum distance that is at most logarithmic in their block length \cite{DZ97}. Although decoding of adversarial errors is impossible beyond half the code distance, a large fraction of random errors of far greater Hamming weight may be decodable.

Interestingly, for cycle codes, unlike for general LDPC codes, polynomial-time decoders can achieve exact maximum-likelihood decoding, \textit{i.e.} saturate the information-theoretic limit. A cycle code can be associated with a graph whose edges represent bits and whose vertices represent the parity checks. A given syndrome corresponds to a subset $T$ of the vertices, and the lowest Hamming weight error (which is the maximum likelihood error for the binary symmetric channel) is given by the minimum-weight T-join, which can be found in polynomial time \cite{EJ73}.

In \cite{DZ97}, the following theorem is proven.
\begin{theorem}
\label{thm:dz_limit}
    Consider an asymptotic family of LDPC codes in which each bit is contained in exactly two parity checks and each parity check contains exactly $D$ bits. The rate of these codes is then $R=1-2/D$. Let $p_2$ be the largest probability such that, if each bit is independently flipped with probability $p_2$, then maximum-likelihood decoding will recover the original codeword with probability converging to one in the limit of large block size. Then,
    \begin{equation}
    \label{eq:p2}
    p_{2} \leq \frac{1}{2} \frac{(1-\sqrt{R})^2}{1+R}.
    \end{equation}
\end{theorem}
This theorem provides new information specific to $k=2$ because, as one can easily verify, the above bound on $p_2$ lies below the corresponding Shannon bound $p_{\textrm{general}} \leq H_2^{-1}(1-R)$ for general codes. We also note that our computer experiments suggest that, for $D=3$ the bound \eq{eq:p2} is essentially saturated by the cycle codes arising from random 3-regular graphs.

For max-2-XORSAT, we are mainly interested in the case $2 \ell + 1 > d^\perp$ since, as noted above, $d^\perp = O(\log n)$. We therefore rely on Theorem \ref{thm:imperfect_decoding}, which together with Theorem \ref{thm:dz_limit} implies the following.

\begin{theorem}
\label{thm:max2limit}
Consider an asymptotic family of max-2-XORSAT instances in which $\mathbf{v} \in \mathbb{F}_2^m$ is chosen uniformly at random and each variable is contained in exactly $D$ constraints. In the limit of large $n$ the performance of DQI using classical decoders is limited by
\begin{equation}
\label{eq:cycle_limit}
\frac{\langle s \rangle}{m} \leq \frac{1}{2} + \frac{1}{2 (D-1)}.
\end{equation}
\end{theorem}

\begin{proof}
Rewriting \eq{eq:imperfect_asymptotic_beyond_avg} from Theorem \ref{thm:imperfect_decoding} in terms of the expected number of constraints satisfied $\langle s \rangle$ instead of the expected objective value $\langle f \rangle$ yields
\begin{equation}
\label{eq:proofstep}
\frac{\langle s \rangle}{m} = \frac{1}{2} + \sqrt{\frac{\ell}{m} \left( 1 - \frac{\ell}{m} \right)} - \varepsilon.
\end{equation}
From Theorem \ref{thm:dz_limit}, the information-theoretic limit of decoding cycle codes is at 
\begin{equation}
    \label{eq:inflimit}
    \frac{\ell}{m} = \frac{1}{2} \frac{(1-\sqrt{R})^2}{1+R} \quad \mathrm{and} \quad \varepsilon = 0
\end{equation}
Substituting \eq{eq:inflimit} and $R=1-2/D$ into \eq{eq:proofstep} and simplifying yields \eq{eq:cycle_limit}.
\end{proof}
 
\section{DQI for Folded Codes and over Extension Fields}
\label{sec:folded}

In \sect{sec:DQI}, we describe the DQI algorithm for the max-LINSAT problem over prime fields.  However, DQI works over extension fields as well, and also works for so-called \emph{folded codes}. In this section, we go through the details to extend DQI to these settings.  Our main motivation is to show that DQI is applicable to the problem considered by Yamakawa and Zhandry in \cite{YZ22}, which is similar to our OPI problem, but for folded Reed-Solomon codes. Moreover, as discussed in \sect{sec:prior}, if a variant of Theorem \ref{thm:semicircle} applies in the regime studied by Yamakawa and Zhandry, then equation \eqref{eq:semicircle_general} implies that DQI can find a solution satisfying all constraints.

\subsection{Folded max-LINSAT problem}

In \cite{YZ22}, Yamakawa and Zhandry define the following oracle problem, which they prove can be solved in polynomially many queries by a quantum computer but requires exponentially many queries for classical computers. 

\begin{definition}\label{def:YZ}
    Fix a prime power $q$ and integers $m,n,r$ such that $r$ divides $m$ and $m > n$.
    Let $\mathcal{O}:\{1, \ldots, m/r\} \times \F_q^r \to \{0,1\}$ be a random function.   
    Let $B \in \F_q^{m\times n}$ be a Vandermonde matrix (so that $B_{i,j} = \gamma^{ij}$ for $i\in\{0,\ldots,m-1\}$, $j\in\{0,\ldots,n-1\}$ where $\gamma$ is a primitive element of $\F_q$), written as
    \begin{equation}
        B = \left[\begin{array}{c}
        B_1 \\
        \hline
        B_2 \\
        \hline
        \vdots \\
        \hline
        B_{m/r}
        \end{array}\right]
    \end{equation}
    where $B_i \in \F_q^{r \times n}$. The Yamakawa-Zhandry problem is, given $B$ and query access to $\mathcal{O}$, to efficiently find $\mathbf{x} \in \F_q^n$ such that $\mathcal{O}(i, B_i\mathbf{x}) = 1$ for all $i \in \{1,\ldots,m/r\}.$
\end{definition}

The problem in Definition~\ref{def:YZ} has some similarities to Definition~\ref{def:linsat}, and especially to our OPI example in \sect{sec:OPI}, but is not exactly the same, as the problem in Definition~\ref{def:YZ} is for \emph{folded} Reed-Solomon codes, over an extension field $\mathbb{F}_q$.  Below, we extend DQI to this setting.  More precisely, we consider the following generalization of Definition~\ref{def:linsat}.

\begin{definition}[Folded max-LINSAT]\label{def:folded_maxlinsat}
Let $\mathbb{F}_q$ be a finite field, where $q$ is any prime power. For $i=1,\ldots, m/r$, let $f_i:\mathbb{F}_q^r \to \{+1,-1\}$ be arbitrary functions. Given a matrix $B \in \mathbb{F}_q^{m \times n}$ written as
\begin{equation}
    B = \left[\begin{array}{c}
    B_1 \\
    \hline
    B_2 \\
    \hline
    \vdots \\
    \hline
    B_{m/r}
    \end{array}\right]
\end{equation}
with $B_i \in \mathbb{F}_q^{r \times n}$, the \emph{$r$-folded max-LINSAT problem} is to find $\mathbf{x} \in \mathbb{F}_q^n$ maximizing the objective function
\begin{equation}\label{eq:folded_maxlinsat_objective}
    f(\mathbf{x}) = \sum_{i=1}^{m/r} f_i(B_i \mathbf{x}).
\end{equation}
\end{definition}

We now describe how to adapt the presentation in \sect{sec:genp} to the folded max-LINSAT problem. As before, we first discuss the properties of the DQI state $|P(f)\rangle := \sum_{\mathbf{x}\in\mathbb{F}_p^n}P(f(\mathbf{x}))|\mathbf{x}\rangle$ and then describe the algorithm for creating it.

Again, mirroring \sect{sec:genp}, we assume that $2\ell + 1 < d^\perp$ where $d^\perp$ is the minimum distance of the folded code $C^\perp = \{ \mathbf{d} \in \mathbb{F}_q^m : B^T \mathbf{d} = \mathbf{0} \}$. Note that folding affects the definition of $d^\perp$. In a folded code, we view every codeword $\mathbf{y}\in\mathbb{F}_q^m$ as an $m/r$-tuple $(\mathbf{y}_1,\ldots\mathbf{y}_{m/r})$ of elements of $\mathbb{F}_q^r$ and regard each $\mathbf{y}_i$ as a symbol. Consequently, the Hamming weight $|.|:\mathbb{F}_q^m \to \{0,\ldots,m/r\}$ associated with the folded code is the number of $\mathbf{y}_i$ not equal to $\mathbf{0}\in\mathbb{F}_q^r$.

\subsection{DQI Quantum State for Folded max-LINSAT}

As in \sect{sec:genp}, we assume that no $f_i$ is constant and that the preimages $F_i := f_i^{-1}(+1)$ have the same cardinality for all $i=1,\ldots,m/r$. This allows us to define $g_i$ as $f_i$ shifted and rescaled so that its Fourier transform
\begin{equation}
\label{eq:folded_tildeg}
\tilde{g}_i(\mathbf{y}) = \frac{1}{\sqrt{q^r}} \sum_{\mathbf{x} \in \mathbb{F}_q^r} \omega_p^{\mathrm{tr}(\mathbf{y} \cdot \mathbf{x})} g_i(\mathbf{x}),
\end{equation}
where $\mathrm{tr}: \mathbb{F}_q \to \mathbb{F}_p$ given by $\mathrm{tr}(x) = x + x^p + x^{p^2} + \ldots + x^{q/p}$ is the field trace, vanishes at $\mathbf{y}=\mathbf{0}\in\mathbb{F}_q^r$ and is normalized, \textit{i.e.}  $\sum_{\mathbf{x}\in\mathbb{F}_q^r}|g_i(\mathbf{x})|^2=\sum_{\mathbf{y}\in\mathbb{F}_q^r}|\tilde{g}_i(\mathbf{y})|^2=1$. More explicitly, we define
\begin{equation}
    g_i(\mathbf{x}) := \frac{f_i(\mathbf{x}) - \overline{f}}{\varphi}
\end{equation}
where $\overline{f} := \frac{1}{q^r}\sum_{\mathbf{x}\in\mathbb{F}_q^r} f_i(\mathbf{x})$ and $\varphi := \left(\sum_{\mathbf{y}\in\mathbb{F}_q^r}|f_i(\mathbf{y})-\overline{f}|^2\right)^{1/2}$. The sums $f(\mathbf{x})=\sum_{i=1}^{m/r}f_i(B_i\mathbf{x})$ and $g(\mathbf{x})=\sum_{i=1}^{m/r}g_i(B_i\mathbf{x})$ are related by $f(\mathbf{x})=g(\mathbf{x})\varphi + m\overline{f}/r$. Substituting this relationship for $f$ in $P(f)$, we obtain an equivalent polynomial $Q(g)$ which, by Lemma \ref{thm:esp} in Appendix \ref{app:esp}, can be expressed as a linear combination of elementary symmetric polynomials $P^{(k)}$
\begin{equation}
    Q(g(\mathbf{x})) := \sum_{l=0}^\ell u_k P^{(k)}\left(g_1(B_1\mathbf{x}), \ldots, g_{m/r}(B_{m/r}\mathbf{x})\right).
\end{equation}
We will write the DQI state
\begin{equation}
    |P(f)\rangle = \sum_{\mathbf{x}\in\mathbb{F}_q^n} P(f(\mathbf{x}))|\mathbf{x}\rangle = \sum_{\mathbf{x}\in\mathbb{F}_q^n} Q(g(\mathbf{x}))|\mathbf{x}\rangle = |Q(g)\rangle
\end{equation}
as a linear combination of $|P^{(0)}\rangle, \ldots, |P^{(\ell)}\rangle$ defined as
\begin{equation}
    \ket{P^{(k)}} := \frac{1}{\sqrt{q^{n-rk}{m/r \choose k}}} \sum_{\mathbf{x} \in \F_q^n} P^{(k)}\left(g_1(B_1\mathbf{x}),\ldots,g_{m/r}(B_{m/r}\mathbf{x})\right) \ket{\mathbf{x}}.
\end{equation}
By definition,
\begin{eqnarray}
    P^{(k)}(g_1(B_1 \mathbf{x}),\ldots,g_{m/r}(B_{m/r}\mathbf{x})) & = & \sum_{\substack{i_1,\ldots,i_k \\ \mathrm{distinct}}} \prod_{i \in \{i_1,\ldots,i_k\}} g_i(B_i \mathbf{x}) \\
    & = & \sum_{\substack{i_1,\ldots,i_k \\ \mathrm{distinct}}} \prod_{i \in \{i_1,\ldots,i_k\}} \left( \frac{1}{\sqrt{q^r}} \sum_{\mathbf{y}_i \in \mathbb{F}_q^r} \omega_p^{ -\mathrm{tr}(\mathbf{y}_i \cdot B_i \mathbf{x})} \ \tilde{g}_i(\mathbf{y}_i) \right) \\
    & = &  \sum_{\substack{\mathbf{y} \in \mathbb{F}_q^m \\ |\mathbf{y}| = k }} \frac{1}{\sqrt{q^{rk}}} \  \omega_p^{-\mathrm{tr}((B^T \mathbf{y}) \cdot \mathbf{x})} \ \prod_{\substack{i=1 \\ \mathbf{y}_i \neq 0}}^{m/r} \tilde{g}_i(\mathbf{y}_i)\label{eq:esp_state_for_folded_dqi}
\end{eqnarray}
where $|.|:\mathbb{F}_q^m \to \{0,\ldots,m/r\}$ is the Hamming weight associated with the folded code. From \eqref{eq:esp_state_for_folded_dqi} we see that the Quantum Fourier Transform of $\ket{P^{(k)}}$ is
\begin{equation}
    |\widetilde{P}^{(k)}\rangle := F^{\otimes n}|P^{(k)}\rangle = \frac{1}{\sqrt{\binom{m/r}{k}}} \sum_{\substack{\mathbf{y}\in\mathbb{F}_q^m\\|\mathbf{y}|=k}} \left( \prod_{\substack{i=1\\\mathbf{y}_i \neq 0}}^{m/r} \tilde{g}_i(\mathbf{y}_i) \right)
    |B^T \mathbf{y}\rangle.
\end{equation}
As in \sect{sec:genp}, if $|\mathbf{y}| < d^\perp/2$, then $B^T \mathbf{y}$ are all distinct and $\ket{P^{(0)}},\ldots,\ket{P^{(\ell)}}$ form an orthonormal set. Consequently,
\begin{equation}
    \ket{P(f)} = \sum_{k=0}^\ell w_k \ket{P^{(k)}}
\end{equation}
where
\begin{equation}
    w_k = u_k\sqrt{q^{n-rk} \binom{m/r}{k}},
\end{equation}
and $\langle P(f)|P(f)\rangle = \|\mathbf{w}\|^2$.

\subsection{DQI Algorithm for Folded max-LINSAT}

Similarly to DQI for general max-LINSAT, the algorithm for folded max-LINSAT uses three quantum registers: a \textit{weight register} comprising $\lceil\log_2\ell\rceil$ qubits, an \textit{error register} with $m\lceil\log_2 q\rceil$ qubits, and a \textit{syndrome register} with $nr\lceil\log_2 q\rceil$ qubits. We will consider the error and syndrome registers as consisting of $m/r$ and $n$ subregisters, respectively, where each subregister consists of $r\lceil\log_2 q\rceil$ qubits. We will also regard the rightmost qubits from all subregisters of the error register as forming the \textit{mask register} of $m/r$ qubits. We assume that the encoding of $\mathbb{F}_q$ into each of the $r$ components of a subregister uses a basis that contains $1$, so that $1\in\mathbb{F}_q$ is encoded as $\ket{0,\ldots,0,1}$.

We begin by initializing the weight register in the normalized state $\sum_{k=0}^\ell w_k\ket{k}$. Next, we prepare the mask register in the Dicke state corresponding to the weight register
\begin{equation}
    \to \sum_{k=0}^\ell w_k \ket{k} \frac{1}{\sqrt{\binom{m/r}{k}}} \sum_{\substack{\boldsymbol{\mu} \in \{0,1\}^{m/r} \\ |\boldsymbol{\mu}| = k}} \ket{\mathbf{\boldsymbol{\mu}}}
\end{equation}
and then uncompute the weight register, obtaining
\begin{equation}
    \to \sum_{k=0}^\ell w_k \frac{1}{\sqrt{\binom{m/r}{k}}} \sum_{\substack{\boldsymbol{\mu} \in \{0,1\}^{m/r} \\ |\boldsymbol{\mu}| = k}} \ket{\mathbf{\boldsymbol{\mu}}}.
\end{equation}
Let $G_i$ denote a unitary acting on $r\lceil\log_2 q\rceil$ qubits that sends $\ket{\mathbf{0}}$ to $\ket{\mathbf{0}}$ and $\ket{0,\ldots,0,1}$ to
\begin{equation}
    \sum_{\mathbf{c}\in\mathbb{F}_q^r} \tilde{g}_i(\mathbf{c}) \ket{\mathbf{c}} = \sum_{\mathbf{c}\in\mathbb{F}_q^r \setminus \{\mathbf{0}\}} \tilde{g}_i(\mathbf{c}) \ket{\mathbf{c}}.
\end{equation}
See \sect{sec:dqi_folded_oracle} below for an implementation of $G_i$ in the oracle setting. As in the case of DQI for general max-LINSAT, parallel application $G:=\prod_{i=1}^{m/r}G_i$ of $G_i$ to all subregisters of the error register preserves the Hamming weight, so that
\begin{equation}
    \sum_{\substack{\boldsymbol{\mu} \in \{0,1\}^{m/r} \\ |\boldsymbol{\mu}| = k}} G\ket{\mathbf{\boldsymbol{\mu}}} = \sum_{\substack{\mathbf{y} \in \mathbb{F}_q^m \\ |\mathbf{y}|=k}} \tilde{g}_{y(1)}\left(\mathbf{y}_{y(1)}\right) \ldots \tilde{g}_{y(k)}\left(\mathbf{y}_{y(k)}\right) \ket{\mathbf{y}}
\end{equation}
where $\mathbf{y}_i$ for $i\in\{1,\ldots,m/r\}$ denotes the $i\th$ entry of $\mathbf{y}$, and $y(j)$ denotes the index of the $j\th$ nonzero entry of $\mathbf{y}$. Consequently, by applying $G$ to the error register, we obtain
\begin{gather}
    \to \sum_{k=0}^\ell w_k \frac{1}{\sqrt{\binom{m/r}{k}}} \sum_{\substack{\mathbf{y} \in \mathbb{F}_q^m \\ |\mathbf{y}|=k}} \tilde{g}_{y(1)}(\mathbf{y}_{y(1)}) \ldots \tilde{g}_{y(k)}(\mathbf{y}_{y(k)}) \ket{\mathbf{y}}.
\end{gather}
Next, we reversibly compute $B^T\mathbf{y}$ into the syndrome register
\begin{gather}
    \to \sum_{k=0}^\ell w_k \frac{1}{\sqrt{\binom{m/r}{k}}} \sum_{\substack{\mathbf{y} \in \mathbb{F}_q^m \\ |\mathbf{y}|=k}} \tilde{g}_{y(1)}(\mathbf{y}_{y(1)}) \ldots \tilde{g}_{y(k)}(\mathbf{y}_{y(k)}) \ket{\mathbf{y}} \ket{B^T\mathbf{y}}.
\end{gather}
The task of finding $\mathbf{y}$ from $B^T \mathbf{y}$ is the bounded distance syndrome decoding problem on the folded code $C^\perp = \{ \mathbf{y} \in \mathbb{F}_q^m : B^T \mathbf{y} = \mathbf{0} \}$. Consequently, uncomputing the content of the error register can be done efficiently whenever the bounded distance decoding problem on $C^\perp$ can be solved efficiently out to distance $\ell$.

Uncomputing disentangles the syndrome register from the error register, leaving behind
\begin{align}
    \to & \sum_{k=0}^\ell w_k \frac{1}{\sqrt{\binom{m/r}{k}}} \sum_{\substack{\mathbf{y} \in \mathbb{F}_q^m \\ |\mathbf{y}|=k}} \tilde{g}_{y(1)}(\mathbf{y}_{y(1)}) \ldots \tilde{g}_{y(k)}(\mathbf{y}_{y(k)}) \ket{\mathbf{y}} \ket{B^T\mathbf{y}} \\
    = & \sum_{k=0}^\ell w_k \frac{1}{\sqrt{\binom{m/r}{k}}} \sum_{\substack{\mathbf{y} \in \mathbb{F}_q^m \\ |\mathbf{y}|=k}} \left( \prod_{\substack{i=1\\\mathbf{y}_i \neq 0}}^{m/r} \tilde{g}_i(\mathbf{y}_i) \right) \ket{B^T \mathbf{y}} \\
    = & \sum_{k=0}^\ell w_k \ket{\widetilde{P}^{(k)}}
\end{align}
which becomes the desired state $\ket{P(f)}$ after applying the Quantum Fourier Transform.

\subsection{Oracle Access to Objective Function}
\label{sec:dqi_folded_oracle}

In our discussion of DQI for general max-LINSAT, we assumed that the field size $p$ is polynomial in $n$. In that setting, the objective functions $f_i$ can be given explicitly by their values at all elements of $\mathbb{F}_p$ and hence the gates $G_i$ can be realized efficiently using techniques from \cite{LKS18}. By contrast, in Yamakawa-Zhandry problem, the random function $\mathcal{O}$ is available via query access to an oracle.

Here, we show how to realize the gates $G_i$ in a setting where the functions $f_i(.) = \mathcal{O}(i,.)$ are provided by oracles and without assuming that the field size $q$ is polynomial in $n$. For simplicity, we assume that $q$ is a power of two, $f_i$ is balanced, \textit{i.e.} $|f^{-1}(+1)| = q^r/2$, and $f_i(\mathbf{0}) = +1$.

Suppose oracle $U_i$ for $f_i$ is defined as
\begin{equation}
    U_i\ket{\mathbf{x}} = f_i(\mathbf{x}) \ket{\mathbf{x}}
\end{equation}
for $\mathbf{x}\in\mathbb{F}_q^r$ and let
\begin{equation}
    \label{eq:qft_over_fqr}
    F\ket{\mathbf{x}} = \frac{1}{\sqrt{q^r}} \sum_{\mathbf{y} \in \mathbb{F}_q^r} (-1)^{\mathrm{tr}(\mathbf{x} \cdot \mathbf{y})} \ket{\mathbf{y}}
\end{equation}
be the Quantum Fourier Transform on a subregister of the error register. Then $G_i$ can be realized using a single auxiliary qubit as shown in Fig.~\ref{fig:g_i_from_oracle}. When the input is $\ket{\mathbf{0}}$, then all operations act as identity, so $G_i\ket{\mathbf{0}}=\ket{\mathbf{0}}$. When the input is $\ket{0,\ldots,0,1}$, then the SWAP gate sets the auxiliary qubit to the $\ket{1}$ state and the input qubits into $\ket{\mathbf{0}}$. Subsequent three operations yield
\begin{align}
    FU_iF\ket{\mathbf{0}} = \sum_{\mathbf{y}\in\mathbb{F}_q^r} \tilde{g}_i(\mathbf{y})\ket{\mathbf{y}}.
\end{align}
Moreover, $\tilde{g}_i(\mathbf{0})=0$, so the last operation uncomputes the auxiliary qubit.

\begin{figure}
\begin{center}
\begin{quantikz}
\lstick{$\ket{0}$} & \swap{1} & \ctrl{1}    &               & \ctrl{1}    & \targ{} & \rstick{$\ket{0}$} \\
& \targX{} & \gate[3]{F} & \gate[3]{U_i} & \gate[3]{F} & \gate[3]{\ne 0} \wire[u]{q} & \\
\lstick{\vdots} \setwiretype{n} & & & & & & \rstick{\vdots} \\
& & & & & & 
\end{quantikz}
\end{center}
\caption{Quantum circuit implementing $G_i$ gate using oracle access. The circuit takes as input a single subregister of the error register and employs an auxiliary qubit initialized in $\ket{0}$. It begins by swapping the qubit corresponding to $1$ in the input with the auxiliary qubit. Then it applies the Quantum Fourier Transform (QFT) $F$ over $\mathbb{F}_q^r$, followed by a call to the oracle $U_i$, followed by another QFT. The QFT gates are conditional on the auxiliary qubit. The circuit ends with the uncomputation of the auxiliary qubit by flipping it when the subregister is non-zero.}
\label{fig:g_i_from_oracle}
\end{figure}

\section{Multivariate Optimal Polynomial Intersection}
\label{sec:multivariate}

In this section we describe the application of DQI to the multivariate generalization of the OPI problem.

\begin{definition}
    \label{def:GPR}
    Let $r,u,m,q$ be integers where $q$ is a prime power, $1 \leq u \leq (q-1)m$, and $1 \leq r < q$. For each $\mathbf{z} \in \mathbb{F}_q^m$ let $L(\mathbf{z}) \subset \mathbb{F}_q$ with $|L(\mathbf{z})| = r$. Given such subsets specified by explicit tables, the multivariate optimal polynomial intersection problem $\mathrm{mOPI}(r,u,m,q)$ is to find a polynomial $Q \in \mathbb{F}_q[z_1,\ldots,z_m]$ of total degree at most $u$ that maximizes the objective function $f$ defined by
    \begin{equation}
        f[Q] = \left| \left\{ \mathbf{z} \in \mathbb{F}_q^m : Q(\mathbf{z}) \in L(\mathbf{z}) \right\} \right|.
    \end{equation}
\end{definition}

\noindent
An instance of $\mathrm{mOPI}(r,u,m,q)$ is specified by the lists $L(\mathbf{z}) \subset \mathbb{F}_q$ for all $\mathbf{z} \in \mathbb{F}_q^m$. Specifying these requires $q^{m+1}$ bits. Hence, throughout this section when we say ``polynomial-time'' we mean polynomial in $q^m$. The OPI problem is the special case of mOPI where $m=1$ and $q$ is a prime of polynomial magnitude. 

We next recount some background information about Reed-Muller codes, which we will need to describe how DQI can be applied to the mOPI problem. The following exposition is based on \cite{ASY20}.

For a given multivariate polynomial $Q(z_1,\ldots,z_m)$ over $\mathbb{F}_q$ let $\mathrm{Eval}(Q)$ be the symbol string in $\mathbb{F}_q^n$ obtained by evaluating $Q$ at all possible assignments to $z_1,\ldots,z_m$ in lexicographical order. Hence, 
\begin{equation}
    n = q^m.
\end{equation}
Let $\mathcal{P}_{q,m,u}$ be the set of polynomials in $\mathbb{F}_q[z_1,\ldots,z_m]$ of total degree at most $u$. In $\mathbb{F}_q$, raising a variable to power $q-1$ yields the identity. Thus, the number of distinct monomials from which elements of $\mathcal{P}_{q,m,u}$ can be constructed is
\begin{equation}
    \label{eq:kgen}
    k = \left| \left\{ (i_1,\ldots,i_m) \ : \ 0 \leq i_j \leq q-2 \textrm{ and } \sum_{j=1}^m i_j \leq u\right\} \right|.
\end{equation}
A polynomial in $\mathcal{P}_{q,m,u}$ is determined by choosing the coefficients from $\mathbb{F}_q$ for each of these $k$ monomials. Hence, $|\mathcal{P}_{q,m,u}| = q^k$.
\begin{definition}
    Given a prime power $q$ and integers $u,m$ satisfying $1 \leq u \leq m(q-1)$, the corresponding Reed-Muller code $\mathrm{RM}_q(u,m)$ is $\{ \mathrm{Eval}(Q) : Q \in \mathcal{P}_{q,m,u}\}$.
\end{definition}
\noindent
We observe that $\mathrm{RM}_q(u,m)$ is an $\mathbb{F}_q$-linear code; it linearly maps from the $k$ coefficients that define $Q$ to the $n$ values in $\mathrm{Eval}(Q)$. \\

\noindent
As discussed in \cite{PW04}, the dual of a Reed-Muller code is also Reed-Muller code.
\begin{theorem}
    \label{thm:rmdual}
    The dual code to $\mathrm{RM}_q(u,m)$ is $\mathrm{RM}_q(u^\perp,m)$, where
    \begin{equation}
        \label{eq:uperp}
        u^\perp = m(q-1)-u-1. 
    \end{equation}
\end{theorem}

\noindent
From \cite{DGM70}, based on \cite{KLP68}, we have the following.

\begin{theorem}
    \label{thm:rmdist}
    Let $\alpha$ and $\beta$ be the quotient and remainder obtained when dividing $u$ by $q-1$. That is, $u = \alpha(q-1)+\beta$ with $0 \leq \beta < q-1$. Then the distance of $\mathrm{RM}_q(u,m)$ is $d = (q-\beta)q^{m-\alpha-1}$.
\end{theorem}

\noindent
As discussed in Chapter 13 of \cite{GRM23}, the methods of \cite{PW04} imply a polynomial-time classical reduction from decoding Reed-Muller codes to decoding Reed-Solomon codes, and therefore the following theorem holds.

\begin{theorem}
    \label{thm:rmdecode}
    The code $\mathrm{RM}_q(u,m)$ can be decoded from errors up to weight $\left \lfloor \frac{d-1}{2} \right \rfloor$ with perfect reliability by a polynomial time classical algorithm.
\end{theorem}

From the above facts we see that DQI reduces the problem $\mathrm{mOPI}(r,u,m,q)$ to decoding of a Reed-Muller code $\mathrm{RM}_q(u^\perp,m)$, where $u^\perp = m(q-1) - u - 1$. Given an algorithm that can decode $\mathrm{RM}_q(u^\perp, m)$ out to $\ell$ errors, DQI will achieve an expected value of $f$ given by the following theorem, which is a straightforward generalization of Lemma \ref{thm:genw}. (Here we include finite-size corrections since taking a limit of large problem size while keeping the ratio of constraints to variables fixed is not straightforward in the context of mOPI.)

\begin{theorem}
    \label{thm:semicircle_general}
    Suppose we have an efficient algorithm that decodes $C^\perp = \mathrm{RM}_q(u^\perp, m)$ out to $\ell$ errors. Let $d^\perp$ be the distance of $C^\perp$. If $2 \ell + 1 < d^\perp$, then for any instance of $\mathrm{GPR}(r,u,m,q)$, DQI produces in polynomial time samples from polynomials $Q$ such that expected value of the objective $f[Q]$ is
    \begin{equation}
        \langle f \rangle = q^m \frac{r}{q} + \frac{\sqrt{r(q-r)}}{q} \lambda_{\max}^{(q)}
    \end{equation}
    where $\lambda_{\max}^{(q)}$ is the largest eigenvalue of the following $(\ell+1)\times(\ell+1)$ symmetric tridiagonal matrix
    \begin{equation}
        A^{(q)}=\begin{bmatrix}
        0   & a_1 & \\
        a_1 & d   & a_2   & \\
            & a_2 & 2d    & \ddots  & \\
            &     & \ddots &        & a_\ell \\
            &     &        & a_\ell & \ell d
        \end{bmatrix}
    \end{equation}
    with $a_k=\sqrt{k(q^m-k+1)}$ and $d = \frac{q-2r}{\sqrt{r(q^m-r)}}$.
\end{theorem}

Together, theorems \ref{thm:semicircle_general}, \ref{thm:rmdecode}, and \ref{thm:rmdist} yield a strong performance guarantee for DQI applied to $\mathrm{mOPI}(r,u,m,q)$, particularly when $q \geq u$. Comparing this performance against competing classical algorithms remains for future work.

\section{Resource Estimation for OPI}
\label{sec:resources}

In this section we look at the resources required to construct a quantum circuit for syndrome decoding of Reed Solomon codes using the Berlekamp-Massey decoding algorithm \cite{B15}. Since the decoding step is the dominant cost in DQI, this gives us an estimate of the resource requirements for DQI to solve OPI. Whereas our example in \sect{sec:OPI} uses a ratio of ten constraints per variable, here we consider two constraints per variable, as this appears to be a more optimal choice for the purpose of solving classically-intractable instances of OPI using as few quantum gates and qubits as possible. That is, throughout this section, the number of constraints is $p-1$ and the number of variables is $n \simeq p/2$.

An outline of the key steps of Berlekamp-Massey syndrome decoding algorithm is given in \cref{algo:bkm_decoding}. Readers seeking more detail can see Appendix E of \cite{dodis2008fuzzy}, where the algorithm used here is explained in the context of BCH decoding. Asymptotically, the most computationally intensive step is the subroutine  BerlekampMasseyLFSR, responsible for finding the shortest linear feedback shift register (LFSR). Standard irreversible implementations of this subroutine often rely on conditional branching and variable assignments that do not directly translate to efficient reversible circuits. Consequently, a naive reversible implementation of this subroutine can lead to significant overhead in terms of quantum resources, potentially undermining the overall efficiency of the DQI algorithm.

We address this challenge by presenting an optimized reversible implementation of the Berlekamp-Massey algorithm for finding the shortest LFSR in \cref{algo:bkm_lfsr}. Our implementation is a generalization of the implementation given by \cite{chevignard2024reducing}, and works for any finite field $\F_q$. For a sequence of length $n$ and a retroaction polynomial of maximum degree $\ell$, \cref{algo:bkm_lfsr} can be implemented as a quantum circuit using $\mathcal{O}(n \cdot \ell)$ multiplications in $\mathbb{F}_{q}$ and using $2 \cdot(n + \ell)\cdot \lceil\log_2{q}\rceil + n + \log_2{\ell}$ qubits. For finite fields $\F_p$, where $p$ is a prime, we list the costs for performing modular arithmetic operations in \cref{tab:arithmetic_costs}. In our Zenodo record (\url{https://doi.org/10.5281/zenodo.13327870}) we provide an implementation of \cref{algo:bkm_lfsr} for prime fields $\F_p$ using Qualtran \cite{harrigan2024expressing}. We present the resulting resource estimates in \cref{tab:lfsr_resource_estimates}.

\begin{table}
    \centering
    \scalebox{0.8}{
    \begin{tabular}{|c|c|c|c|c|c|c|}
    \hline
    & Subroutine Name & Subroutine Action & Toffoli & Ancilla & Reference \\
    \hline
    \multirow{6}{*}{\rotatebox[origin=c]{90}{Arithmetic}}
    & & & & & \\
    & quantum-classical addition &$\ket{x}\rightarrow\ket{x+K}$ 
    & $n$  & $n$ & \cite{fedoriaka2025newcircuitquantumadder} \\

    & controlled quantum-classical addition &$\ket{c}\ket{x}\rightarrow\ket{c}\ket{x+cK}$ 
    & $n$  & $2n$ & \cite{häner2020improvedquantumcircuitselliptic} \\ 
    
    & quantum-quantum addition &$\ket{x}\ket{y}\rightarrow\ket{x}\ket{x+y}$ 
    & $n$  & $n$ & \cite{Gidney_2018} \\

    & controlled quantum-quantum addition  
    &$\ket{c}\ket{x}\ket{y}\rightarrow\ket{x}\ket{x+cy}$
    & $2n$ & $n$ & \cite{Gidney_2018}\\
    & & & & & \\
    \hline
    \multirow{9}{*}{\rotatebox[origin=c]{90}{Modular Arithmetic}}
    & & & & & \\
    
    & modular quantum-classical addition 
    &$\ket{x}\rightarrow\ket{(x+K)\bmod P}$ 
    & $2.5n$ & $n$ & \cite{luongo2024measurementbaseduncomputationquantumcircuits} \\

    & modular quantum-classical controlled addition 
    &$\ket{c}\ket{x}\rightarrow\ket{c}\ket{(x+cK)\bmod P}$ 
    & $2.5n$  & $2n$ & \cite{luongo2024measurementbaseduncomputationquantumcircuits} \\

    & modular quantum-quantum addition  
    &$\ket{x}\ket{y}\rightarrow\ket{x}\ket{(x+y)\bmod P}$ 
    & $3.5n$  & $n$ & \cite{luongo2024measurementbaseduncomputationquantumcircuits} \\

    & modular quantum-quantum controlled addition  
    &$\ket{c}\ket{x}\ket{y}\rightarrow\ket{c}\ket{x}\ket{(x+cy)\bmod P}$ 
    & $4.5n$  & $n$ & \cite{luongo2024measurementbaseduncomputationquantumcircuits} \\

    & modular controlled scaled addition 
    &$\ket{x}\ket{y} \rightarrow \ket{x}\ket{(y + x*K) \bmod P}$ &
    $2.5n^2$  & $2n$ &  \cite{Gidney_2021} \\

    & modular quantum-quantum multiplication 
    &$\ket{x}\ket{y} \rightarrow \ket{x}\ket{y}\ket{(x*y) \bmod P}$ &
    $3n^2 + 2n - 1$  & $2n$ &  \cite{litinski2023compute256bitellipticcurve} \\

    & modular multiplicative inverse 
    & $\ket{x} \rightarrow \ket{x^{-1} \bmod P}$  
    & $26n^2 + 2n$ &  $5n$ & \cite{litinski2023compute256bitellipticcurve} \\
    & & & & & \\
    \hline
    \end{tabular}
    }
    \caption{Quantum circuit costs for modular arithmetic operations on $n$-bit operands in $\F_p$.}
    \label{tab:arithmetic_costs}
\end{table}

\begin{table}[H]
    \centering
    \begin{tabular}{|c|c|c|c|c|}
    \hline
    \hline
    $RS_{p}(N, K)$  &  BerlekampMasseyLFSR($n$, $\ell$, $\lceil\log_2{p}\rceil$) & Toffoli & Clifford & Qubits \\
    \hline
    \hline
    $RS_{67}(66, 34)$  &  $(32, 16, 7)$ & \num{999850} & \num{5153696} & \num{907} \\
    \hline
    $RS_{131}(130, 66)$  &  $(64, 32, 8)$ & \num{4607692} & \num{23798358} & \num{1912} \\
   \hline
    $RS_{257}(256, 128)$  &  $(128, 64, 9)$ & \num{21546662} & \num{110333236} & \num{4101} \\
   \hline
    $RS_{521}(520, 264)$  &  $(256, 128, 10)$ & \num{101011904} & \num{521039438} & \num{8850} \\
   \hline
    $RS_{1031}(1030, 518)$  &  $(512, 256, 11)$ & \num{471606346} & \num{2444572208} & \num{19103} \\
   \hline
    $RS_{2153}(2152, 1128)$  &  $(1024, 512, 12)$ & \num{2186280548} & \num{11380033666} & \num{41132} \\
   \hline
\end{tabular}
\caption{Cost of finding shortest linear feedback shift register (LFSR) using \cref{algo:bkm_lfsr}, implemented and analyzed using Qualtran \cite{harrigan2024expressing}. 
This is the most expensive step of Berlekamp-Massey syndrome decoding algorithm \cite{B15} for Reed Solomon codes, as presented in \cref{algo:bkm_decoding}.  
Here $n = N - K$ is the number of syndromes, which is equal to the length of the input sequence to Berlekamp-Massey algorithm, and $\ell = \frac{n}{2}$ is the maximum number of correctable errors, which is equal to the degree of retroaction polynomial.}
\label{tab:lfsr_resource_estimates}
\end{table}

As a point of comparison, we can estimate the classical cost of solving the instances in Table \ref{tab:lfsr_resource_estimates} by repeating Prange's algorithm with different size-$n$ random subsets of the constraints until the target satisfaction fraction is reached. As an example, consider the case $p=521$. Here, we have 260 variables and 521 constraints, of which DQI is able to satisfy 486. In a given trial of Prange's algorithm one uses Gaussian elimination to obtain a solution that is guaranteed to satisfy 260 of the constraints and which satisfies each of the remaining 261 constraints with probability $r/p = 1/2$. To beat DQI Prange's algorithm needs to satisfy at least 226 among these remaining 261 constraints. The probability of this on a given trial is
\begin{equation}
    \frac{1}{2^{261}} \sum_{m=226}^{261} \binom{261}{m} \simeq 10^{-35}.
\end{equation}
Thus solving this instance by the repeat-until-success version of Prange's algorithm should require on the order of $10^{35}$ repetitions. Therefore, assuming a CPU can execute on the order of a billion elementary operations per second this yields a total cost of at least $10^{26}$ CPU-seconds, even if each trial could be parallelized into a single clock-cycle.

As noted in \cite{BMN21}, achieving such large separation between classical and quantum resource costs at reasonable problem size is only possible when the underlying source of advantage goes beyond quadratic speedups based on Grover's algorithm and its generalizations such as amplitude amplification and quantum walks. More specifically, in \cite{CKM19} the cost of solving random 14-SAT at the satisfiability phase transition using Grover's algorithm and quantum backtracking methods was estimated, and it was found that Grover's algorithm required lower resources, namely a T-depth of $10^{14}$ and a Toffoli count of $10^{19}$ on instances that they estimate would require $10^{10}$ CPU-seconds classically. 

Although the repeat-until-success version of Prange's algorithm is a standard technique that provides a useful point of comparison, we do not specifically claim it to be optimal. In \sect{sec:list-recovery} and \sect{sec:lattice} we have surveyed all of the classical algorithms that we can find in the literature applicable to OPI and find that none yield polynomial-time solutions in the parameter regime discussed here. Nevertheless, there may be classical techniques to achieve better scaling than Prange's algorithm, while still requiring exponential time; such incremental improvements to exponential runtimes are quite common in the cryptanalysis literature (see for example \cite{BCL17}). We leave the investigation of this possibility to future work.

\algnewcommand{\LineComment}[1]{\State \(\triangleright\) #1}
\begin{algorithm}
    \caption{Syndrome decoding of $RS_{q}(\vec{\gamma}, N, K)$ using Berlekamp Massey algorithm}
    \begin{algorithmic}[1]
        \State \textbf{Input:} A list $s = [s_1, s_2, \dots, s_{n}]$ of syndromes where $n = N - K$.
        \State \textbf{Output:} A list $e = [e_1, e_2, \dots, e_{N}]$ of error values such that $|\text{supp}(e)| \le \ell$ where $\ell = \frac{n}{2}$
        \LineComment{\underline{Step-1:} Find error locator polynomial $\sigma(Z)$ of degree $\le \ell$, by solving for shortest linear feedback shift register (LFSR) using Berlekamp-Massey algorithm.}
        \LineComment{Uses $\mathcal{O}(N^2)$ multiplications and $\mathcal{O}(N)$ inversions of elements in $\F_q$.}
        \State $\sigma(Z) \gets \text{BerlekampMasseyLFSR}(s, \ell)$ 
        \LineComment{\underline{Step-2:} Find error evaluator polynomial $\Omega(Z)$ via fast polynomial multiplication using Number Theoretic Transform (NTT).}
        \LineComment{Uses $\mathcal{O}(N \log{N})$ multiplications of elements in $\F_q$.}
        \State $\Omega(Z) \gets \sigma(Z) \times S(Z) \mod Z^{n+1}$
        \LineComment{\underline{Step-3:} Find roots of $\sigma(Z)$ to determine locations of errors by evaluating $\sigma(Z)$ for all $N$ roots of unity ($\vec{\gamma}$) using NTT.}
        \LineComment{Uses $\mathcal{O}(N \log{N})$ multiplications of elements in $\F_q$.}
        \State $\text{sigma\_roots} \gets \text{chien\_search}(\sigma(Z), \vec{\gamma})$
        \LineComment{\underline{Step-4:} Use evaluations of $\Omega(Z)$ and $\sigma^{\prime}(Z)$ for all N roots of unity $(\vec{\gamma})$ to determine the values of errors $e$.}
        \LineComment{Uses $\mathcal{O}(N \log{N})$ multiplications and $\mathcal{O}(N)$ divisions of elements in $\F_q$.}
        \State $e \gets \text{forneys\_algorithm}(\Omega(Z), \sigma(Z), \text{sigma\_roots}, \vec{\gamma}$)
    \end{algorithmic}
    \label{algo:bkm_decoding}
\end{algorithm}

\begin{algorithm}
    \caption{Reversible Berlekamp-Massey algorithm to find shortest LFSR over $\F_q$.}
    \begin{algorithmic}[1]
        \State \textbf{Input:} A list $s = [s_0, s_1, \dots, s_{n-1}]$ in $\F_q$ and an integer $\ell \leq n$. 
        \State \textbf{Output:} Retroaction polynomial $C(X) \in \F_q[X]$ such that $\deg(C(X)) \leq \ell$.
        \State \textbf{Storage:} Register for $s$ ($n$ elements in $\F_q$) and $C(X)$ ($\ell$ elements in $\F_q$)
        \State \textbf{Garbage:} Register for $L$ ($\log_2{\ell}$ bits), $B(X)$ ($\ell$ elements in $\F_q$), $d=[d_0, d_2, \dots, d_{n-1}]$ ($n$ elements in $\F_q$) and $v = [v_0, v_2, \dots, v_{n-1}]$ ($n$ bits)
        \State \textbf{Total Qubits:} $2\times(n + \ell)\times \lceil\log_2{q}\rceil + \log_2(\ell) + n$ 
        \State $c \gets [-1]$
        \Comment{$c$ and $b$ correspond to list of coefficients for $C(X)$ and $B(X)$; $C(X) = \sum_{i=0}^{\ell}c_{i}x^{i}$}
        \State $b \gets [-1]$
        \Comment{$-1$ is the additive inverse of $1$ in $F_q$}
        
        \State $L \gets 0$
        \For{$i \text{ in } 0, \dots, N-1$}
            \State $d_{i} \gets \sum_{j=0}^{\text{len}(c)}s_{i - j} \times c_{j}$ 
            \Comment{$\leq \ell + 1$ quantum-quantum multiplications and additions}
            \State $v_{i} \gets (2L \leq i)$
            \Comment{quantum-classical comparator to decide between case 2 and case 3}
            \If{$i < \ell$} 
                \State $c \gets \text{concatenate}(c, [0])$
                \Comment{Extends register for $C(X)$ upto a maximum length $\ell+1$}
                \State $b \gets \text{concatenate}([0], b)$
                \Comment{Equivalent to performing $B(X) \gets X\cdot B(X)$}
            \Else
                \State $b \gets \text{roll}(b, 1)$
                \Comment{Equivalent to performing $B(X) \gets X\cdot B(X)$}
            \EndIf 
            \State $r \gets (1 \textbf{ if } d_i = 0 \textbf{ else } d_i)$
            \State $b \gets b \times r$ 
            \Comment{$\leq \ell + 1$ in-place quantum-quantum multiplications}
            \If{$d_i \neq 0$}
                \State $c \gets c - b$
                \Comment{$\leq \ell + 1$ quantum-quantum controlled subtractions}
            \EndIf
            \If {$d_i \neq 0 \textbf{ and } v_{i}$}
                \State $L \gets i + 1 - L$
                \State $b \gets b + c$
                \Comment{$\leq \ell + 1$ quantum-quantum controlled additions}
            \EndIf
            \State $b \gets b \times r^{-1}$
            \Comment{$\leq \ell + 1$ in-place quantum-quantum multiplications and $1$ modular inverse}
        \EndFor
        \State \Return c
    \end{algorithmic}
    \label{algo:bkm_lfsr}
\end{algorithm}

\clearpage

\appendix

\section{Elementary Symmetric Polynomials}
\label{app:esp}

In this section we prove Lemma \ref{thm:esp}. Recall that our max-LINSAT objective function takes the form $f(\mathbf{x}) = f_1(\mathbf{b}_1\cdot\mathbf{x}) + \ldots + f_m(\mathbf{b}_m\cdot\mathbf{x})$, where $f_j:\mathbb{F}_p \to \{+1,-1\}$. Hence $f_j^2 - 1 = 0$ for all $j$. Consequently, the rescaled functions $g_j(x) = (f_j(x) - \overline{f}) / \varphi$ also obey a quadratic identity that does not depend on $j$. (Namely, $\varphi^2 g_j^2 + 2 \overline{f} \varphi g + \overline{f}^2  - 1 = 0$.) Thus the conditions of Lemma \ref{thm:esp} are met by both $f_1,\ldots,f_m$ and $g_1,\ldots,g_m$.

\begin{lemma}
\label{thm:esp}
Let $P$ be any degree-$\ell$ polynomial in a single variable. Let $x_1,\ldots,x_m$ be variables that each obey a quadratic identity $a x_j^2 + b x_j + c = 0$, where $a,b,c$ are independent of $j$ and $a \ne 0$. Then $P(x_1 + \ldots + x_m)$ can be expressed as a linear combination of elementary symmetric polynomials: $P(x_1 + \ldots + x_m) = \sum_{k=0}^\ell u_k P^{(k)}(x_1,\ldots,x_m)$.
\end{lemma}

\begin{proof}
    Reducing modulo the quadratic identities:
    \begin{equation}
        \label{eq:identities}
        a x_j^2 + b x_j + c = 0 \quad j=1,\ldots,m
    \end{equation}
    brings $P(x_1 + \ldots + x_m)$ into a form where none of $x_1,\ldots,x_m$ is raised to any power greater than one. We can think of the resulting expression as a multilinear multivariate polynomial $P(x_1,\ldots,x_m)$, which is of total degree $\ell$. Since the coefficients in the $m$ identities described in \eq{eq:identities} are independent of $j$, this multilinear multinomial is symmetric. That is, for any permutation $\pi \in S_m$ we have $P(x_{\pi(1)}, \ldots, x_{\pi(m)}) = P(x_1,\ldots,x_m)$. By the fundamental theorem of symmetric polynomials, multilinearity and symmetry together imply that the degree-$\ell$ polynomial $P(x_1,\ldots,x_m)$ can be expressed as a linear combination of elementary symmetric polynomials of degree at most $\ell$.
\end{proof}

\section{Gallager's Ensemble}
\label{app:random_regular}

In \cite{G63} Gallager defined the following ensemble of matrices over $\mathbb{F}_2$. This ensemble is widely used in coding theory because when parity check matrices are drawn from Gallager's ensemble, the resulting LDPC codes have good parameters with high probability. Together with a choice of $\mathbf{v}$, a matrix $B$ sampled from Gallager's ensemble also induces a natural ensemble of $D$-regular max-$k$-XORSAT instances.

Given parameters $(k,D,b)$, a sample from Gallager's ensemble of matrices $B \in \mathbb{F}_2^{bk \times bD}$ is generated as follows. Let $A$ denote the horizontal concatenation of $D$ identity matrices, each $b \times b$, yielding $A = [I_1 I_2 \ldots I_D]$. For $i=1,\ldots,k$ let $M_i = A P_i$ where $P_1,\ldots,P_k$ are independent uniformly random $bD \times bD$ permutation matrices. Concatenate these vertically yielding
\begin{equation}
B^T = \left[\begin{array}{c}
M_1 \\
\hline
M_2 \\
\hline
\vdots \\
\hline
M_k
\end{array}\right]
\end{equation}
The matrix $B^T$ thus has $n = bk$ rows and $m = bD$ columns, with $k$ ones in each row and $D$ ones in each column. By choosing each function $f_1,\ldots,f_m$ independently at random to be either $f_i(x) = (-1)^x$ or $f_i(x) = -(-1)^x$ we obtain a $D$-regular instance of max-$k$-XORSAT as a special case of max-LINSAT over $\mathbb{F}_2$.

Gallager's $(k,D,b)$ ensemble is not equivalent to sampling uniformly from all $bk \times bD$ matrices with $k$ ones in each row and $D$ ones in each column. It shares many properties with such a distribution but is more convenient to sample from.

\section{Simulated Annealing Applied to OPI}
\label{app:RSSA}

By the results of \sect{sec:local} we expect simulated annealing to yield solutions to the OPI instances of \sect{sec:OPI} where the fraction $\phi_{\max}$ of constraints satisfied scales like 
\begin{equation}
    \phi_{\max} \simeq \frac{1}{2} + \frac{c}{D^\nu}.
\end{equation}
where $c$ and $\nu$ are free parameters of the fit. According to our crude theoretical model $\nu$ should be $1/2$. However, extrapolating from our empirical results with sparse instances we would expect $\nu$ to be slightly smaller than $1/2$. As shown in Fig. \ref{fig:RSSA}, experimental results match well to this prediction with $\nu = 0.45$. In OPI every variable is contained in every constraint so the degree $D$ is equal to the number of constraints $m$.

\begin{figure}[ht]
    \begin{center}
        \includegraphics[width=0.6\textwidth]{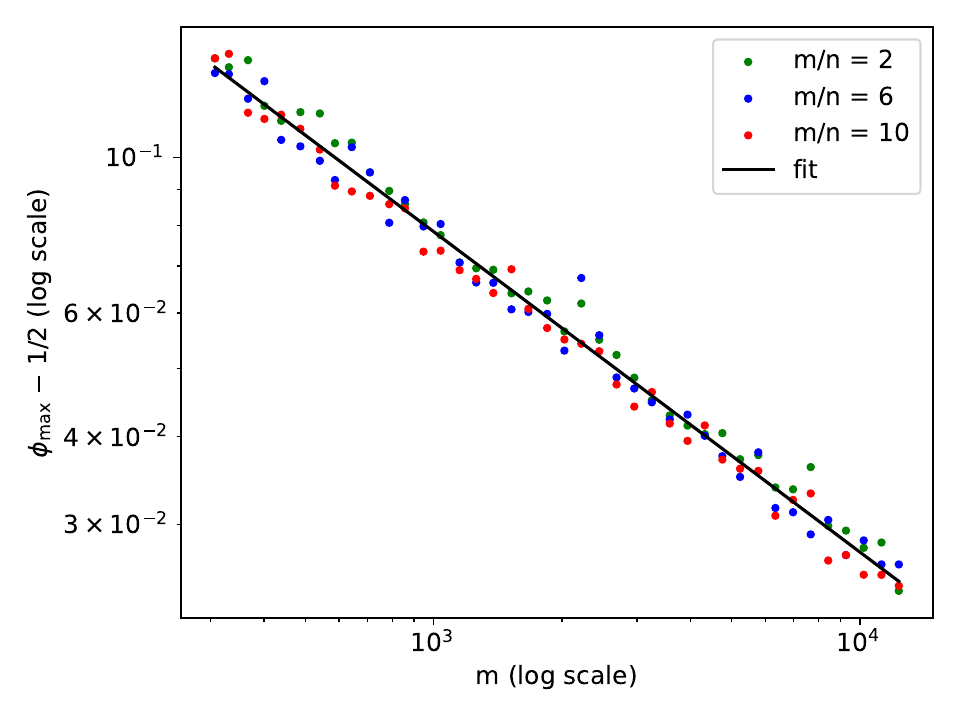}
    \end{center}
    \caption{Here we generate OPI instances over $\mathbb{F}_p$ where $p$ takes prime values from $307$ to $12,343$. The number of constraints is $m = p-1$. For each $m$ we take $n\in\{m/2,m/6,m/10\}$, rounded to the nearest integer. We find that, independent of $n/m$, the approximation achieved by simulated annealing with $10,000$ sweeps fits well to $\phi_{\max} = 1/2+1.8 D^{-0.45}$. Note that, in OPI the degree $D$ equals the number of constraints $m$, since every variable is contained in every constraint.
    \label{fig:RSSA}}
\end{figure}

\bibliography{dqi}

\end{document}